%% file: main.tex
\documentclass[a4paper]{article}
\usepackage{authblk}
\input{preamble.tex}

\usetikzlibrary {decorations.shapes}
\usetikzlibrary {decorations.markings} 
\usetikzlibrary {arrows.meta}
\usetikzlibrary{calc}
\usetikzlibrary{bending} %
\usetikzlibrary{decorations.pathreplacing} 
\usetikzlibrary{external}

\title{String Diagrams for Closed Symmetric Monoidal Categories}
\author{Callum Reader}
\affil{Tallinn University of Technology, Estonia}
\author{Alessandro Di Giorgio}
\affil{Tallinn University of Technology, Estonia}
\date{}

\begin{document}

\maketitle

\begin{abstract}
    We introduce a graphical language for closed symmetric monoidal categories based on an extension of string diagrams with special bracket wires representing internal homs. These bracket wires make the structure of the internal hom functor explicit, allowing standard morphism wires to interact with them through a well-defined set of graphical rules.
    We establish the soundness and completeness of the diagrammatic calculus, and illustrate its expressiveness through examples drawn from category theory, logic and programming language semantics.
\end{abstract}
\input{sections/intro.tex}
\input{sections/presentation3.tex}

\input{sections/examples2.tex}
\input{sections/completeness.tex}

\input{sections/conclusions.tex}

\bibliography{main}%
\appendix
\input{sections/appendix/presentation2.tex}

\end{document}

%% file: preamble.tex
\usepackage{mathrsfs}   %
\usepackage[dvipsnames]{xcolor}
\usepackage{tikz-cd}    %
\usepackage{tikz}       %
\usepackage{xspace}     %
\usepackage{graphicx}
\usepackage{booktabs}
\usepackage{ebproof}
\usepackage{bbm}
\usepackage{leftidx}
\usepackage{listofitems}
\usepackage{caption}
\usepackage[outline]{contour}
\usepackage{float}

\usepackage{wrapfig}

\graphicspath{{Figures/}}   %
\usepackage{comment}        %
\usepackage{xstring}

\usepackage{afterpage}

\usepackage{multicol}

\DeclareFontFamily{U}{mathb}{}%
\DeclareFontShape{U}{mathb}{m}{n}{
  <5> <6> <7> <8> <9> <10>
  <10.95> <12> <14.4> <17.28> <20.74> <24.88>
  mathb10
  }{}
\DeclareSymbolFont{mathb}{U}{mathb}{m}{n}
\DeclareMathSymbol{\mminus}{\mathbin}{mathb}{"01}   %

\usepackage[all,knot,poly]{xy}
\usepackage{standalone}
\usepackage{amsmath,amssymb,amsthm,amsfonts}
\usepackage{multirow}
\usepackage{pdfpages}
\usepackage{mathtools}
\usepackage{nicefrac}
\usepackage{accents}
\usepackage{caption}
\usepackage{mathtools}
\usepackage{thmtools}
\usepackage{thm-restate}
\usepackage{tabularx}
\usepackage{scalerel}
\usepackage{geometry}[margin=1in]
\usepackage{pgffor}
\usepackage{circuitikz}
\usepackage[nomessages]{fp}
\usepackage{xstring} %
\usepackage{xargs}
\usepackage{ifthen}
\usepackage{cmll}
\usepackage{longtable}

\usepackage{mathpartir}
\usepackage{xcolor}
\usepackage{soul}
\usepackage{cleveref}
\usepackage{subcaption}
\newcommand{\autoref}[1]{\cref{#1}}

\sethlcolor{orange}

\RequirePackage{snapshot}

\makeatletter
\DeclareFontFamily{U}{MnSymbolA}{}
\DeclareFontShape{U}{MnSymbolA}{m}{n}{
    <-6>  MnSymbolA5
   <6-7>  MnSymbolA6
   <7-8>  MnSymbolA7
   <8-9>  MnSymbolA8
   <9-10> MnSymbolA9
  <10-12> MnSymbolA10
  <12->   MnSymbolA12}{}
\DeclareFontShape{U}{MnSymbolA}{b}{n}{
    <-6>  MnSymbolA-Bold5
   <6-7>  MnSymbolA-Bold6
   <7-8>  MnSymbolA-Bold7
   <8-9>  MnSymbolA-Bold8
   <9-10> MnSymbolA-Bold9
  <10-12> MnSymbolA-Bold10
  <12->   MnSymbolA-Bold12}{}
\DeclareSymbolFont{MnSyA}{U}{MnSymbolA}{m}{n}
\SetSymbolFont{MnSyA}{bold}{U}{MnSymbolA}{b}{n}

\DeclareRobustCommand{\overleftharpoon}{\mathpalette{\overarrow@\leftharpoonfill@}}
\DeclareRobustCommand{\overrightharpoon}{\mathpalette{\overarrow@\rightharpoonfill@}}
\def\leftharpoonfill@{\arrowfill@\leftharpoondown\mn@relbar\mn@relbar}
\def\rightharpoonfill@{\arrowfill@\mn@relbar\mn@relbar\rightharpoonup}

\DeclareMathSymbol{\leftharpoondown}{\mathrel}{MnSyA}{'112}
\DeclareMathSymbol{\rightharpoonup}{\mathrel}{MnSyA}{'100}
\DeclareMathSymbol{\mn@relbar}{\mathrel}{MnSyA}{'320}
\makeatother

\usetikzlibrary{patterns}
\usetikzlibrary{through}
\usetikzlibrary{knots}
\usetikzlibrary{decorations}
\usetikzlibrary{fit}
\usetikzlibrary{shadows,shadings,shapes.symbols}
\usetikzlibrary{math}
\usetikzlibrary{hobby}%

\usepackage{stmaryrd}
\newcommand{\semFunct}[1]{\left \llbracket #1 \right \rrbracket}
\newcommand{\sort}{\mathcal{S}}
\newcommand{\sign}{\Sigma}
\newcommand{\cate}[1]{\mathscr{#1}}
\newcommand{\interpretation}[1]{\mathcal{#1}}
\newcommand{\objects}[1]{Obj(#1)}
\newcommand{\arrows}[1]{Ar(#1)}
\newcommand{\arity}{ar}
\newcommand{\coarity}{coar}

\newcommand{\freeCat}{\cate{C}_\sign}
\newcommand{\normalised}[1]{\uparrow \! {#1} \! \downarrow}
\newcommand{\freeNormalCat}{\normalised{\freeCat}}

\newcommand{\lTerm}[3]{#1 \vdash #2 \colon #3}

\usetikzlibrary{decorations.pathmorphing}

\tikzset{%
    symbol/.style={%
        draw=none,
        every to/.append style={%
            edge node={node [sloped, allow upside down, auto=false]{$#1$}}}
    }
}

\newcommand{\outin}[2]{
out=\IfEqCase{#1}{
{ne}{-45}
{se}{45}
{sw}{135}
{nw}{-135}
}[0],
in=\IfEqCase{#2}{
{ne}{-45}
{se}{45}
{sw}{135}
{nw}{-135}
}
}

\newcommand{\compassToDegrees}[1]{
\IfEqCase{#1}{
{ne}{-45}
{se}{45}
{sw}{135}
{nw}{-135}
}[#1]
}

\newcommand{\inedge}[1]{
to [in=#1]
}
\newcommand{\inedgetest}[1]{\expandafter\inedge\expandafter{\compassToDegrees{#1}}}

\theoremstyle{plain}
\newtheorem{theorem}{Theorem}[section]
\newtheorem*{theorem*}{Theorem}

\newtheorem{lemma}[theorem]{Lemma}
\newtheorem*{lemma*}{Lemma}

\renewcommand{\autoref}{\cref}          %

\crefname{equation}{Equation}{Equations}
\crefname{theorem}{Theorem}{Theorems}
\crefname{lemma}{Lemma}{Lemmata}
\crefname{example}{Example}{Examples}

\theoremstyle{definition}
\newtheorem{definition}[theorem]{Definition}
\newtheorem{example}[theorem]{Example}

\newtheorem*{remark}{Remark}

\renewcommand{\phi}{\varphi}                	    %

\newcommand{\blank}{\phantom{ }}                    %

\newcommand{\tensor}{\otimes}                       %

\newcommand{\C}{\mathscr{C}}

\renewcommand{\hom}[1]{\left[#1\right]}

\DeclareSymbolFont{bbold}{U}{bbold}{m}{n}
\DeclareSymbolFontAlphabet{\mathbbold}{bbold}

\DeclareMathSymbol{\mminus}{\mathbin}{mathb}{"01}

\newcommand{\id}{\mathrm{id}}

\DeclareMathOperator{\Set}{\mathsf{Set}}

\DeclareFontFamily{U}{min}{}
\DeclareFontShape{U}{min}{m}{n}{<-> udmj30}{}

\newsavebox{\pullbackdrawing}
\sbox\pullbackdrawing{%
\begin{tikzpicture}%
\draw (0,0) -- (1ex,0ex);%
\draw (1ex,0ex) -- (1ex,1ex);%
\end{tikzpicture}}

\newsavebox{\pushforwarddrawing}
\sbox\pushforwarddrawing{%
\begin{tikzpicture}[rotate=180, transform shape]%
\draw (0,0) -- (1ex,0ex);%
\draw (1ex,0ex) -- (1ex,1ex);%
\end{tikzpicture}}

\makeatletter
\def\slashedarrowfill@#1#2#3#4#5{%
  $\m@th\thickmuskip0mu\medmuskip\thickmuskip\thinmuskip\thickmuskip
  \relax#5#1\mkern-7mu%
  \cleaders\hbox{$#5\mkern-2mu#2\mkern-2mu$}\hfill
  \mathclap{#3}\mathclap{#2}%
  \cleaders\hbox{$#5\mkern-2mu#2\mkern-2mu$}\hfill
  \mkern-7mu#4$%
}
\def\rightslashedarrowfill@{%
  \slashedarrowfill@\relbar\relbar\mapstochar\rightarrow}
\newcommand\xprofto[2][]{%
  \ext@arrow 0055{\rightslashedarrowfill@}{#1}{#2}}
\makeatother

\tikzset{commutative diagrams/.cd, 
    prof/.style = "\shortmid"{marking, sloped},
    }

\DeclareMathSymbol\simplex  \mathord{bbold}{"01}

\DeclareFontFamily{U}{mathx}{\hyphenchar\font45}
\DeclareFontShape{U}{mathx}{m}{n}{<->mathx10}{}
\DeclareFontSubstitution{U}{mathx}{m}{n}
\DeclareSymbolFont{mathx}{U}{mathx}{m}{n}

\DeclareMathSymbol{\endint}{\mathop}{mathx}{"B3}

\DeclareFontFamily{U}{mathx}{\hyphenchar\font45}
\DeclareFontShape{U}{mathx}{m}{n}{
      <5> <6> <7> <8> <9> <10>
      <10.95> <12> <14.4> <17.28> <20.74> <24.88>
      mathx10
      }{}
\DeclareSymbolFont{mathx}{U}{mathx}{m}{n}
\DeclareFontSubstitution{U}{mathx}{m}{n}
\DeclareMathAccent{\widecheck}{0}{mathx}{"71}

\newcommand{\lift}{\mathrel{\multimap}}
\tikzset{shorten <>/.style={shorten >=#1,shorten <=#1}}

\newcommand{\elift}{
\mathbin{
    \vcenter{
      \hbox{${-}$}
    }
      \kern-1.25pt
    {
      \vcenter{\hbox{${\scriptstyle\diamond}$}}
    }
  }
}

\DeclareMathSymbol{\boxbackslash}{2}{mathb}{"6E}

    \usepackage[a]{esvect}
    \usepackage{array}

    \usetikzlibrary{positioning,arrows.meta,calc,fit,quotes,cd,math,arrows,backgrounds,shapes.geometric}
    
    \makeatletter
\DeclareRobustCommand{\rvdots}{%
  \vbox{
    \baselineskip4\p@\lineskiplimit\z@
    \kern-\p@
    \hbox{.}\hbox{.}\hbox{.}
  }}
\makeatother
\makeatletter
\def\instring#1#2{TT\fi\begingroup
  \edef\x{\endgroup\noexpand\in@{#1}{#2}}\x\ifin@}
    \let\nvec\Vec
    \def\Vec#1{\nvec{\vphantom b\smash{#1}}}
    \renewcommand{\Vec}[1]{\overrightarrow{\vphantom a\smash{#1}}}
    \renewcommand{\Vec}[1]{
    \if\instring{#1}{ABCDEFGHIJKLMNOPQRSTUVWXYZ}
    \vv{\vphantom{A}\smash{#1}}
    \else
    \vv{\vphantom{\tiny{a}}\smash{#1}}
    \fi}

  \makeatletter
  \newcommand{\superimpose}[2]{%
  {\ooalign{$#1\@firstoftwo#2$\cr\hfil$#1\@secondoftwo#2$\hfil\cr}}}
  \makeatother

    \newcommand{\forgetdagger}[1]{\ifthenelse{\equal{#1}{(-)}}{^{\dagger}}{#1^{\dagger}}}
    
    \newcommand{\freedagger}[1]{#1^{\dagger}}
    \let\f\freedagger
    \renewcommand{\freedagger}[1]{\ifthenelse{\equal{#1}{}}{\f{(-)}}{\f{#1}}}
    
    \newcommand{\freemulti}[1]{#1^{\triangleright}}
    \let\f\freemulti
    \renewcommand{\freemulti}[1]{\ifthenelse{\equal{#1}{}}{\f{(-)}}{\f{#1}}}

    \newcommand{\forgetmulti}[1]{#1_{\triangleright}}
    \let\f\forgetmulti
    \renewcommand{\forgetmulti}[1]{\ifthenelse{\equal{#1}{}}{\f{(-)}}{\f{#1}}}

    \let\f\freemonoidal
    \renewcommand{\forgetmulti}[1]{\ifthenelse{\equal{#1}{}}{\f{(-)}}{\f{#1}}}

    \let\f\forgetmonoidal
    \renewcommand{\forgetmulti}[1]{\ifthenelse{\equal{#1}{}}{\f{(-)}}{\f{#1}}}
    
    \renewcommand{\forgetdagger}[1]{L_{\dagger}#1}
    \renewcommand{\freedagger}[1]{R_{\dagger}#1}
    \renewcommand{\freemulti}[1]{L_{\triangleright}#1}
    \renewcommand{\forgetmulti}[1]{R_{\triangleright}#1}

\input{closedMacros.tex}

    \tikzstyle{finish}=[to path={(\tikztostart) -- (\tikztotarget.#1)}]
    \tikzstyle{start}=[to path={(\tikztostart.#1) -- (\tikztotarget)}]
    \tikzstyle{vertical}=[
      to path={
          (\tikztostart) 
          -- (\tikztostart|-\tikztotarget.north)\tikztonodes}
    ]

    \tikzstyle{verticalup}=[
      to path={
          (\tikztostart) 
          -- (\tikztostart|-\tikztotarget.south)\tikztonodes}
    ]

    \makeatletter
    \newcommand{\asti}{%
      \check@mathfonts
      \leavevmode
      {\ooalign{%
        \hidewidth
        \raisebox{.8ex}{\fontsize{\ssf@size}{0}\selectfont*}%
        \hidewidth\cr
        \i\cr
      }}%
    }
    \makeatother

\newcommand{\diagLabel}[1]{
  #1
}
\newcommand{\brCup}{\underline{{\color{red} \cup}}}
\newcommand{\brCap}{\underline{{\color{red} \cap}}}
\newcommand{\dXin}{\underrightarrow{\check{\chi}}{\color{red} \!|\,}}
\newcommand{\dXout}{\underleftarrow{\check{\chi}}{\color{red} \!|\,}}
\newcommand{\uXin}{\underrightarrow{\hat{\chi}}{\color{red} \!|\,}}
\newcommand{\uXout}{\underleftarrow{\hat{\chi}}{\color{red} \!|\,}}
\newcommand{\Xcup}{\underrightarrow{\psi}{\color{red} \!\mid\,}}
\newcommand{\rightCap}{\underrightarrow{\cap}}
\newcommand{\rightCup}{\underrightarrow{\Cup}}

\newcommand{\YdXin}{\Yflip{\dXin}}
\newcommand{\YdXout}{\Yflip{\dXout}}
\newcommand{\YuXin}{\Yflip{\uXin}}
\newcommand{\YuXout}{{\Yflip{\uXout}}}

\newcommand{\YXcup}{\Yflip{\Xcup}}
\newcommand{\YrightCap}{{\Yflip{\rightCap}}}
\newcommand{\YrightCup}{{\Yflip{\rightCup}}}

\newcommand{\flip}[1]{\raisebox{0.5\depth}{\scalebox{1}[-1]{$#1$}}}
\newcommand{\Yflip}[1]{{\scalebox{-1}[1]{$#1$}}}
\newcommand{\XYflip}[1]{\Yflip{\flip{#1}}}

\newcommand{\subfig}[2]{
  \begin{subfigure}{#1\textwidth}
    \centering
    #2
  \end{subfigure}\ignorespaces
}

\newcommand{\figLab}[1]{\\\vspace{0.2cm}\textbf{(#1)}}
\newcommand{\figSpace}{\vspace{0.75cm}}

\crefname{relation}{Equation}{Equations}
\crefname{figure}{Figure}{Figures}

\pgfdeclarelayer{background}
\pgfdeclarelayer{edgelayer}
\pgfdeclarelayer{nodelayer}
\pgfsetlayers{background,edgelayer,nodelayer,main}
\tikzstyle{none}=[inner sep=0pt]
\tikzstyle{solidString}=[{myArrow/.list}={{0.25}{},{0.75}{}}, solid]
\tikzstyle{dottedString}=[{myArrow/.list}={{0.25}{},{0.75}{}}, dotted]
\tikzstyle{lBracket}=[{myArrow/.list}={{0.25}{left},{0.75}{left}}, red]
\tikzstyle{rBracket}=[{myArrow/.list}={{0.25}{right},{0.75}{right}}, red]

 \newcommand{\normF}{\mathcal{N}}
 \newcommand{\unref}[1]{#1^{\rightarrow}}

 \newcommand{\minimal}[1]{#1^\bot}

%% file: closedMacros.tex
\definecolor{cb-orange}{cmyk}{0,0.5,1,0}
\definecolor{cb-light-blue}{cmyk}{0.8,0,0,0}
\definecolor{cb-green}{cmyk}{0.97,0,0.75,0}
\definecolor{cb-yellow}{cmyk}{0.1,0.05,0.9,0}
\definecolor{cb-blue}{cmyk}{1,0.5,0,0}
\definecolor{cb-red}{cmyk}{0,0.8,1,0}
\definecolor{cb-pink}{cmyk}{0.1,0.7,0,0}

\definecolor{cb-purple}{RGB}{62,35,115}

\definecolor{alt-plum}{RGB}{51,34,136}
\definecolor{alt-green}{RGB}{17,119,51}
\definecolor{alt-turq}{RGB}{68,170,153}
\definecolor{alt-blue}{RGB}{136,204,238}
\definecolor{alt-yellow}{RGB}{221,204,119}
\definecolor{alt-pink}{RGB}{204,102,119}
\definecolor{alt-purple}{RGB}{170,68,153}
\definecolor{alt-red}{RGB}{136,34,85}

\colorlet{c1}{cb-orange}
\colorlet{c2}{cb-light-blue}
\colorlet{c3}{cb-green}
\colorlet{c4}{cb-pink}
\colorlet{c5}{cb-blue}
\colorlet{c6}{cb-red}
\colorlet{c7}{cb-yellow}
\colorlet{c8}{alt-red}

\tikzset{
  blank/.style={
    inner sep = 0cm, 
    outer sep = 0cm, 
    align=center, 
    minimum size=0pt
  },
  label/.style={
    minimum size=5mm, 
    inner sep=0.1cm,
    fill=white
  },
  midarrow/.style={
    currarrow,
    sloped,
    scale=0.5,
    allow upside down
  },
  shorten <>/.style={
    shorten >=#1,
    shorten <=#1
  },
  bullet/.style={
    rectangle,
    draw,
    fill=black,
    minimum size=3pt,
    inner sep=0pt,
  },
  colourBullet/.style={
    circle,
    draw=#1,
    fill=#1,
    minimum size=3pt,
    inner sep=0pt,
  },
  double colour/.style 2 args={fill=#1, 
  path picture={
    \fill[#2, sharp corners] 
    (path picture bounding box.south) -|
    (path picture bounding box.north east) --
    (path picture bounding box.north) --
    cycle;
    }
  },
  modification/.style 2 args={
    rounded corners=2,
    double colour={#1}{#2},
    minimum size=5mm,
    inner sep = 0.1cm, 
    outer sep = 0cm, 
    align=center,
    text=white
  },
  functorNode/.style={
    modification={#1}{#1}
  },
  twocellarr/.style={
    shorten <= 10pt, 
    shorten >= 10pt
  },
  twocell/.style={
    functorNode={black}
  },
  twocellX/.style ={
    rectangle,
    thick,
    fill=white,
    rounded corners=0.5,
    inner sep = 0.1cm, 
    outer sep = 0cm, 
    align=center,
    minimum width={#1 pt},
    draw=black,
    line width=\lineWidth
    },
  mod/.style ={
    rectangle,
    thick,
    fill=white,
    rounded corners=0.5,
    inner sep = 0.1cm, 
    outer sep = 0cm, 
    align=center,
    minimum width={#1 pt},
  },
  mod2/.style ={
    rectangle,
    dashed,
    thick,
    fill=white,
    rounded corners=0.5,
    inner sep = 0.1cm, 
    outer sep = 0cm, 
    align=center,
    minimum width={#1 pt},
  },
  colChange/.style ={
    rectangle,
    thick,
    fill=black,
    rounded corners=0.5,
    inner sep = 0.1cm, 
    outer sep = 0cm, 
    align=center,
    minimum width={#1 pt},
  },
  only in/.style={
    to path={
      let
      \p1 = ($ (\tikztotarget) - (\tikztostart) $)
      in
      .. controls 
      ($(\tikztostart)!{veclen(\x1,\y1)/3}!($(\tikztotarget)+(#1:{veclen(\x1,\y1)/3})$)$)
      and
      ($(\tikztotarget)+(#1:{veclen(\x1,\y1)/3})$)
      .. (\tikztotarget)
    }
  },
  only out/.style={
    \path draw (\tikztostart) --node[](temp1){} (\tikztotarget-|)
    to path={
      let
      \p1 = ($ (\tikztotarget) - (\tikztostart) $)
      in
      .. controls 
      ($(\tikztostart)+(#1:{veclen(\x1,\y1)/3})$)
      and
      ($(\tikztotarget)!{veclen(\x1,\y1)/3}!($(\tikztostart)+(#1:{veclen(\x1,\y1)/3})$)$)
      .. (\tikztotarget)
    }
  },
  braidTop/.style={
    to path={
      let
      \p1=($(\tikztotarget)-(\tikztostart)$)
      in
      ..controls +(0,\y1/2).. (\tikztotarget)
    }
  },
  braidBottom/.style={
    to path={
      let
      \p1=($(\tikztotarget)-(\tikztostart)$)
      in
      ..controls +(\x1/2,0).. (\tikztotarget)
    }
  },
  NTDash/.style={
    dash pattern= on 0.20cm off 0.05cm,
  },
}

\makeatletter
\newcommand{\gettikzxy}[3]{%
    \tikz@scan@one@point\pgfutil@firstofone#1\relax
    \edef#2{\the\pgf@x}%
    \edef#3{\the\pgf@y}%
}
\makeatother

\setsepchar[.]{;./}

\xdef\gapWidth{1.2}
\xdef\hlWidth{1.2}
\xdef\lineWidth{0.4}
\xdef\outerLineWidth{0.2}

\tikzmath{
    function angle(\a){
        return Mod(\a,360);
    };
    function bisect(\i,\o){
        if angle(\i)<angle(\o) then{
            return angle((angle(\i)+angle(\o))/2);
        }
        else{
            return angle(((angle(\i)+angle(\o))/2)-180);
        };
    };
    function shiftToFactor(\s,\x1,\x2){
        if \x1==\x2 then{
            return 0;
        } else{
            return (1+(\s/(\x1-\x2)))/(28.45274);
        };
    };
    function vecAngle(\x,\y){
        if \y<0 then {
            \s=-1;
        } else {
            \s=1;
        };
        return 180+\s*acos(\x/(veclen(\x,\y)));
    };
    function hlDblSep(\n){
        return \lineWidth+2*\gapWidth+2*(\n-1)*(\hlWidth+\gapWidth);
    };
    function hlDblSep2(\n){
        return 2*\outerLineWidth+4*\gapWidth+\lineWidth+2*\gapWidth+2*(\n-1)*(\hlWidth+\gapWidth);
    };
    function TCLShft(\n){
        return ((\lineWidth-hlDblSep(\n))/2)-\hlWidth;
    };
    function TCelShft(\l,\r){
        return (1/2)*(\r-\l)*(\hlWidth+\gapWidth);
    };
    function TCWidth(\n,\l,\r){
        return \n+2*\hlWidth+hlDblSep(\l)+hlDblSep(\r)+10;
    };
    function TCWidthX(\n,\l,\r){
        return \n+(\l+\r)*(\hlWidth+\gapWidth)+10;
    };
    function symWidth(\n,\l,\r){
        return (max(\l*2*(\hlWidth+\gapWidth),\r*2*(\hlWidth+\gapWidth)))+\n;
    };
    function testing(\l,\r){
        return (max(\l*2*(\hlWidth+\gapWidth),\r*2*(\hlWidth+\gapWidth)));
    };
    function twoCellWidth(\n,\l,\r,\n2,\l2,\r2){
        return 10+max(symWidth(\n,\l,\r), symWidth(\n2,\l2,\r2));
    };
    function outerLineSep(\n){
        return \lineWidth+2*\gapWidth+2*\n*(\hlWidth+\gapWidth);
    };
    function joinAngle(\i,\n){
        return 180*(1-(1/(2*\n)))+(1-\i)*(180/\n);
    };
    function spread(\i,\n,\w){
        if \n==1 then{
            return 0;
        }else{
            return (\i-1)*(\w/(\n-1))-(\w/2);
        };
    };
    function idRad(\i){
        return (1/2)*\lineWidth+(1/2)*\hlWidth+\gapWidth+(\i-1)*(\hlWidth+\gapWidth);
    };
}

\newcommand{\drawStartNodes}[1]{
    \xdef\xCoord{0}
      \foreach \i/\j/\k in {#1}
      {
        \coordinate (\i) at (\xCoord,0);
        \draw (\i) node[label,anchor=south] {\j};
        \xdef\xCoord{\xCoord+\k}
      }
    }
\newcommand{\drawEndNodes}[1]{
    \foreach \i/\j in {#1}
    {
        \draw(\i) node[label,anchor=north]{\j};
    }
}

\newcommandx{\joinTarget}[5][4=\nodesWidth,5=-1,usedefault]{
    \readlist*\inputJT{#1}
    \xdef\xCoordinateList{}
    \xdef\yCoordinateList{}

    \foreachitem \n \in \inputJT{
        \gettikzxy{(\n)}{\ax}{\ay}
        \xdef\xCoordinateList{\ax,\xCoordinateList}
        \xdef\yCoordinateList{\ay,\yCoordinateList}
    }
    \tikzmath{
        \minX=min(\xCoordinateList);
        \maxX=max(\xCoordinateList);
        \xCoordinate=(\maxX+\minX)/2;
        \minY=min(\yCoordinateList);
        \maxY=max(\yCoordinateList);
        \width=\maxX-\minX;
    }
    \ifdim #5 pt < 0 pt
        \xdef\yCoordinate{\minY}
    \else
        \xdef\yCoordinate{\maxY}
    \fi{}
    \xdef#4{\width}
    \coordinate[shift={(0,#2)}] (#3) at (\xCoordinate pt, \yCoordinate pt);
}

\tikzset{
  myArrow/.style 2 args={
    draw=black,
    line width=0.3mm,
    postaction={
      draw,decorate,decoration={
        markings,%
        mark=at position {#1*\pgfdecoratedpathlength-0.15cm} with {\coordinate (bent arrow 1);},
        mark=at position {#1*\pgfdecoratedpathlength-0.05cm} with {\coordinate (bent arrow 2);},
        mark=at position {#1*\pgfdecoratedpathlength+0.05cm} with {\coordinate (bent arrow 3);},
        mark=at position {#1*\pgfdecoratedpathlength+0.15cm} with {\coordinate (bent arrow 4);
          \draw[-{Latex[bend,#2, width=1.5mm, length=1.5mm]}, ] (bent arrow 1) to[curve through={(bent arrow 2) .. (bent arrow 3)}]  (bent arrow 4);
        }
      }
    }
  },
  bracket/.style={
      draw=white,
      double=#1,
      double distance=0.3mm,
      line width=0.5mm,
  },
}

\xdef\dashed{dotted}

\newcommandx{\drawArrowShift}[3][3=,usedefault]{
    \coordinate[shift={(#2)}] (temp) at (#1);
    \draw[myArrow/.list={{0.25}{},{0.75}{}},#3] (#1) to[out=-90, in=90] (temp);
    \coordinate (#1) at (temp);
}

\newcommandx{\dShiftE}[3][3=, usedefault]{
    \coordinate[shift={(#2)}] (temp) at (#1);
    \draw[#3] (#1) to[out=-90, in=90] (temp);
    \coordinate (#1) at (temp);
}

\newcommandx{\uShiftE}[3][3=,usedefault]{
    \coordinate[shift={(#2)}] (temp) at (#1);
    \draw[#3] (temp) to[out=90, in=-90] (#1);
    \coordinate (#1) at (temp);
}

\newcommandx{\dShift}[2]{
  \drawArrowShift{#1}{#2}
}

\newcommandx{\ddShift}[2]{
  \drawArrowShift{#1}{#2}[\dashed]
}

\newcommandx{\sShift}[2]{
    \coordinate[shift={(#2)}] (temp) at (#1);
    \draw[line width=0.3mm, \dashed] (#1) to[out=-90, in=90] (temp);
    \coordinate (#1) at (temp);
}

\newcommandx{\drawArrowTest}[2]{
    \coordinate[shift={(#2)}] (temp) at (#1);
    \draw[myArrow/.list={{0.5}{},{0.5}{}}] (#1) to[out=-90, in=90] (temp);
    \coordinate (#1) at (temp);
}

\newcommandx{\drawBracketShift}[4][3=left, 4=red, usedefault]{
    \coordinate[shift={(#2)}] (temp) at (#1);
    \xdef\tempath{(#1) to[out=-90, in=90] (temp)}
    \draw[bracket=#4] \tempath;
    \draw[myArrow/.list={{0.25}{#3},{0.75}{#3}}, draw=#4,] \tempath;
    \coordinate (#1) at (temp);
}

\newcommandx{\ubShift}[4][3=left, 4=red, usedefault]{
    \coordinate[shift={(#2)}] (temp) at (#1);
    \xdef\tempath{(temp) to[out=90, in=-90] (#1)}
    \draw[bracket=#4] \tempath;
    \draw[myArrow/.list={{0.25}{#3},{0.75}{#3}}, draw=#4,] \tempath;
    \coordinate (#1) at (temp);
}

\newcommandx{\bShift}[4][3=left, 4=red, usedefault]{
  \drawBracketShift{#1}{#2}[#3][#4]
}
\newcommandx{\upArrowShift}[3][3=,usedefault]{
    \coordinate[shift={(#2)}] (temp) at (#1);
    \draw[myArrow/.list={{0.25}{},{0.75}{}},#3] (temp) to[out=90, in=-90] (#1);
    \coordinate (#1) at (temp);
}

\newcommandx{\uShift}{
  \upArrowShift
}

\newcommandx{\udShift}[2]{
  \upArrowShift{#1}{#2}[\dashed]
}
\newcommandx{\leftBracketShift}[4][3=right, 4=red, usedefault]{
    \coordinate[shift={(#2)}] (temp) at (#1);
    \xdef\tempath{(#1) to[out=-90, in=90] (temp)}

    \draw[bracket=#4] \tempath;
    \draw[myArrow/.list={{0.25}{#3},{0.75}{#3}}, draw=#4] \tempath;

    \coordinate (#1) at (temp);
}

\newcommandx{\drawCup}[4][3=0.5,4=, usedefault]{
    \path (#1) to node[midway](temp){} (#2);
    \coordinate[shift={(0,-#3)}] (temp) at (temp);
    \xdef\tempath{(#1) to[out=-90, in=180] (temp) to[out=0, in=-90] (#2)}
    \draw[myArrow/.list={{0.25}{},{0.75}{}},#4]\tempath;
}
\newcommandx{\dCup}{
  \drawCup
}

\newcommandx{\bCupAcross}[4][3=0.5,4=left,usedefault]{
    \path (#1) to node[midway](temp){} (#2);
    \coordinate[shift={(0,-#3)}] (temp) at (temp);
    \xdef\tempath{(#1) to[out=-90, in=180] (temp) to[out=0, in=-90] (#2)}
    \draw[myArrow/.list={{0.25}{#4},{0.75}{#4}},red]\tempath;
}

\newcommandx{\ddCup}[3][3=0.5, usedefault]{
  \drawCup{#1}{#2}[#3][\dashed]
}

\newcommandx{\dbCup}[4][3=0.5,4=, usedefault]{
    \path (#1) to node[midway](temp){} (#2);
    \coordinate[shift={(0,-#3)}] (temp) at (temp);
    \xdef\tempath{(#1) to[out=-90, in=0] (temp) to[out=180, in=-90] (#2)}
    \draw[myArrow/.list={{0.25}{},{0.75}{}},#4]\tempath;
}

\newcommandx{\ddbCup}[3][3=0.5,usedefault]{
  \dbCup{#1}{#2}[#3][\dashed]
}

\newcommandx{\drawCap}[4][3=0.5,4=, usedefault]{
    \path (#1) to node[midway](temp){} (#2);
    \coordinate[shift={(0,#3)}] (temp) at (temp);
    \xdef\tempath{(#1) to[out=90, in=180] (temp) to[out=0, in=90] (#2)}
    \draw[myArrow/.list={{0.25}{},{0.75}{}},#4]\tempath;
}

\newcommandx{\dCap}{
  \drawCap
}

\newcommandx{\dbCap}[4][3=0.5,4=, usedefault]{
    \path (#1) to node[midway](temp){} (#2);
    \coordinate[shift={(0,#3)}] (temp) at (temp);
    \xdef\tempath{(#1) to[out=90, in=0] (temp) to[out=180, in=90] (#2)}
    \draw[myArrow/.list={{0.25}{},{0.75}{}},#4]\tempath;
}

\newcommandx{\ddCap}[3][3=0.5, usedefault]{
  \drawCap{#1}{#2}[#3][\dashed]
}

\newcommandx{\ddbCap}[3][3=0.5, usedefault]{
  \dbCap{#1}{#2}[#3][\dashed]
}

\newcommandx{\bracketCup}[4][3=0.5, 4=red, usedefault]{
    \path (#1) to node[midway](temp){} (#2);
    \coordinate[shift={(0,-#3)}] (temp) at (temp);
    
    \xdef\tempath{(#1) to[out=-90, in=180] (temp)}
    \draw[bracket=#4] \tempath;
    \draw[myArrow/.list={{0.4}{right}}, draw=#4]\tempath;

    \xdef\tempath{(#2) to[out=-90, in=0] (temp)}
    \draw[bracket=#4] \tempath;
    \draw[myArrow/.list={{0.4}{left}}, draw=#4]\tempath;
}

\newcommandx{\upBCup}[4][3=0.5, 4=red, usedefault]{
    \path (#1) to node[midway](temp){} (#2);
    \coordinate[shift={(0,-#3)}] (temp) at (temp);
    
    \xdef\tempath{(temp) to[out=180, in=-90] (#1)}
    \draw[bracket=#4] \tempath;
    \draw[myArrow/.list={{0.4}{left}}, draw=#4]\tempath;

    \xdef\tempath{(temp) to[out=0, in=-90] (#2)}
    \draw[bracket=#4] \tempath;
    \draw[myArrow/.list={{0.4}{right}}, draw=#4]\tempath;
}

\newcommandx{\bCup}[4][3=0.5, 4=red, usedefault]{
  \bracketCup{#1}{#2}[#3][#4]
}
\newcommandx{\bracketCupBack}[4][3=0.5, 4=red, usedefault]{
    \path (#1) to node[midway](temp){} (#2);
    \coordinate[shift={(0,-#3)}] (temp) at (temp);
    
    \xdef\tempath{(#1) to[out=-90, in=180] (temp)}
    \draw[bracket=#4] \tempath;
    \draw[myArrow/.list={{0.4}{left}}, draw=#4]\tempath;

    \xdef\tempath{(#2) to[out=-90, in=0] (temp)}
    \draw[bracket=#4] \tempath;
    \draw[myArrow/.list={{0.4}{right}}, draw=#4]\tempath;
}

\newcommandx{\XbbCup}[4][3=0.5, 4=red, usedefault]{
    \path (#1) to node[midway](temp){} (#2);
    \coordinate[shift={(0,-#3)}] (temp) at (temp);
    
    \xdef\tempath{(temp) to[in=-90, out=180] (#1)}
    \draw[bracket=#4] \tempath;
    \draw[myArrow/.list={{0.6}{left}}, draw=#4]\tempath;

    \xdef\tempath{(temp) to[in=-90, out=0] (#2)}
    \draw[bracket=#4] \tempath;
    \draw[myArrow/.list={{0.6}{right}}, draw=#4]\tempath;
}

\newcommandx{\bbCup}[4][3=0.5, 4=red, usedefault]{
  \bracketCupBack{#1}{#2}[#3][#4]
}

\newcommandx{\bracketCap}[4][3=0.5, 4=red, usedefault]{
    \path (#1) to node[midway](temp){} (#2);
    \coordinate[shift={(0,#3)}] (temp) at (temp);
    
    \xdef\tempath{(temp) to[out=180, in=90] (#1)}
    \draw[bracket=#4] \tempath;
    \draw[myArrow/.list={{0.6}{right}}, draw=#4]\tempath;

    \xdef\tempath{(temp) to[out=0, in=90] (#2)}
    \draw[bracket=#4] \tempath;
    \draw[myArrow/.list={{0.6}{left}}, draw=#4]\tempath;
}

\newcommandx{\buCap}[4][3=0.5, 4=red, usedefault]{
    \path (#1) to node[midway](temp){} (#2);
    \coordinate[shift={(0,#3)}] (temp) at (temp);
    
    \xdef\tempath{(#1) to[out=90, in=180] (temp)}
    \draw[bracket=#4] \tempath;
    \draw[myArrow/.list={{0.6}{left}}, draw=#4]\tempath;

    \xdef\tempath{(#2) to[out=90, in=0] (temp)}
    \draw[bracket=#4] \tempath;
    \draw[myArrow/.list={{0.6}{right}}, draw=#4]\tempath;
}

\newcommandx{\XbbCap}[4][3=0.5, 4=red, usedefault]{
    \path (#1) to node[midway](temp){} (#2);
    \coordinate[shift={(0,#3)}] (temp) at (temp);
    
    \xdef\tempath{(#1) to[in=180, out=90] (temp)}
    \draw[bracket=#4] \tempath;
    \draw[myArrow/.list={{0.4}{left}}, draw=#4]\tempath;

    \xdef\tempath{(#2) to[out=90, in=0] (temp)}
    \draw[bracket=#4] \tempath;
    \draw[myArrow/.list={{0.4}{right}}, draw=#4]\tempath;
}

\newcommandx{\XbCap}[4][3=0.5, 4=red, usedefault]{
    \path (#1) to node[midway](temp){} (#2);
    \coordinate[shift={(0,#3)}] (temp) at (temp);
    
    \xdef\tempath{(#1) to[in=180, out=90] (temp)}
    \draw[bracket=#4] \tempath;
    \draw[myArrow/.list={{0.4}{right}}, draw=#4]\tempath;

    \xdef\tempath{(#2) to[out=90, in=0] (temp)}
    \draw[bracket=#4] \tempath;
    \draw[myArrow/.list={{0.4}{left}}, draw=#4]\tempath;
}

\newcommandx{\bCap}[4][3=0.5, 4=red, usedefault]{
  \bracketCap{#1}{#2}[#3][#4]
}

\newcommandx{\bead}[3]{
  \path (#1) to node[midway](temp){} (#2);

  \path let \p1=($(#1.west)-(#2.east)$), \n1 = {veclen(\p1)}
  in node[minimum width=\n1+8, draw=black, rectangle, fill=white, font=\tiny, rounded corners=0.5mm]  at (temp){$#3$};
}

\newcommandx{\coord}[3]{
  \coordinate[shift={(#3)}] (#1) at (#2);
}

\newcommandx{\dCopy}[7][4=1, 5=1, 6=1, 7=, usedefault]{
  \coordinate[shift={(0,-#5)}] (temp) at (#1);
  \coordinate[shift={(-#4,-#6)}] (tempLeft) at (temp);
  \coordinate[shift={(#4,-#6)}] (tempRight) at (temp);

  \draw[myArrow/.list={{0.5}{}},#7] (#1) to (temp);
  \draw[myArrow/.list={{0.5}{}},#7] (temp) to[out=-170, in=90] (tempLeft);
  \draw[myArrow/.list={{0.5}{}},#7] (temp) to[out=-10, in=90] (tempRight);

  \node[minimum width=0.15cm, inner sep=0pt, draw=black, fill=black, circle](temp) at (temp){};

  \coordinate (#2) at (tempLeft);
  \coordinate (#3) at (tempRight);
}

\newcommandx{\dDel}[1]{
  \node[minimum width=0.15cm, inner sep=0pt, draw=black, fill=black, circle](temp) at (#1){};
}

\newcommandx{\dashbead}[2]{
  \path (#1) to node[midway](temp){} (#2);

  \path let \p1=($(#1.west)-(#2.east)$), \n1 = {veclen(\p1)}
  in node[minimum width=\n1+8, minimum height=((\n1+8)/2), draw=black, dashed, rectangle, fill=white, font=\tiny, rounded corners=0.5mm]  at (temp){};
}

\newcommandx{\midway}[3]{
    \path (#2) to node[midway](#1){} (#3);
}

%% file: sections/intro.tex
\section{Introduction}\label{sec:intro}

Diagrams have long played a central role in mathematics and computer science, serving both as intuitive aids and as rigorous tools for reasoning. 
From commutative diagrams in algebraic topology~\cite{eilenberg2015foundations} to Penrose notation in phyiscs~\cite{Penrose-tensornotation}, Peirce's existential graphs in mathematical logic~\cite{existentialGraphs}, flowcharts in programming~\cite{nassi1973flowchart} and Petri nets in concurrency~\cite{peterson1977petri}, the use of pictures to represent abstract structures predates the formalization of many of the theories they illustrate. %
In particular, diagrammatic reasoning has been crucial in areas where morphisms, transformations, or processes are primary, allowing mathematicians and computer scientists to track complex compositions and dependencies visually.

String diagrams emerged from this tradition as a powerful formalism for reasoning in symmetric monoidal categories. 
Introduced in the context of category theory by Joyal and Street~\cite{joyal1991geometry}, string diagrams provide a graphical language where morphisms are depicted as boxes and wires, and categorical structures—such as tensor product, composition, and symmetry—are encoded geometrically. 
This approach makes algebraic reasoning not only easier to manage but often self-evident.
For instance, the equation $(f ; g) \otimes (h ; i) = (f \otimes h) ; (g \otimes i),$ expressing functoriality of the tensor product, is represented by a single string diagram—both sides correspond to exactly the same picture, shown below. 
\[
\ensuremath{\vcenter{\hbox{\scalebox{0.75}{\input{ClosedDiags/functorialitySD.tex}}}}}
\]
Other equations, such as the naturality of the symmetry, are captured by geometric transformations: for example, sliding a box past a crossing of wires, as illustrated by the pair of diagrams below.
\[
    \ensuremath{\vcenter{\hbox{\scalebox{0.75}{\input{ClosedDiags/symmNat1.tex}}}}} = \ensuremath{\vcenter{\hbox{\scalebox{0.75}{\input{ClosedDiags/symmNat2.tex}}}}}
\]  
As such, string diagrams have found applications across quantum computing~\cite{Abramsky2004,Coecke2017}, database theory~\cite{GCQ}, programming language semantics~\cite{jeffrey1997premonoidal,bonchi2025diagrammatic}, control theory~\cite{Bonchi2014b}, logic~\cite{bonchi2024diagrammatic} and more~\cite{gale2023categorical, piedeleu2025complete, DBLP:journals/corr/abs-2106-07763, Ghica2016}.

With the traction gained by the growing number of applications, string diagrams were naturally extended to accommodate richer categorical settings. These extensions reflect the need to reason diagrammatically in settings where additional layers of structure—such as duals, traces, or multiple monoidal products—are present. For example, compact closed categories~\cite{freyd1989braided} enrich string diagrams with cups and caps to make duality explict as a bending of wires. Traced monoidal categories~\cite{joyal1996traced} introduce loops to represent fixed-point behaviour or feedback, crucial in the semantics of recursion and iteration. In bicategories~\cite{BenabouBicategories}, where composition is associative only up to isomorphism, diagrams~\cite{DBLP:journals/corr/Marsden14} often include 2-cells and higher-dimensional surfaces to manage the extra level of structure. Similarly, diagrams~\cite{comfort2020sheet,bonchi2023deconstructing} for bimonoidal categories~\cite{laplaza_coherence_1972} have to deal with multiple tensor products, requiring new graphical conventions to keep the interactions between monoidal structures clear. 

In this paper, we develop a string diagrammatic language for \emph{closed} symmetric monoidal categories—a class of categories that plays a foundational role in logic and computer science. These categories are distinguished by the presence of an internal hom functor, allowing e.g. to represent morphisms from $X \otimes Y$ to $Z$ as morphisms from $X$ to $[Y, Z]$. In $\mathsf{Set}$, the category of sets and functions, this operation corresponds to currying, with $[Y,Z]$ being the set of functions from $Y$ to $Z$. In our language currying is represented as the following diagram.
\[ \ensuremath{\vcenter{\hbox{\scalebox{0.75}{\input{ClosedDiags/curryingIntro.tex}}}}} \]
In particular we employ red ‘‘bracket'' strings to form the internal hom $[Y,Z]$; the direction on the wires indicates whether the objects are covariant or contravariant in the internal hom functor.

This structure captures the essence of higher-order computational models, underlying the categorical semantics of typed $\lambda$-calculi. In proof theory, the internal hom corresponds to implication, and its adjunction with the tensor product models the process of introducing and eliminating assumptions.%

Despite the centrality of these categories, their string diagrammatic treatment has traditionally been more limited, due to the added complexity of representing internal homs in a geometric way.
The aim of this paper is to fill this gap by introducing a formal diagrammatic language that is enough to support a wide range of applications, from logic to programming language semantics, as already suggested in this introduction.%

{\bf Synopsis.} In Section~\ref{sec:presentation}, we present our diagrammatic language via generators and equations, and we prove that it forms a closed symmetric monoidal category. Furthermore, we prove some derived equations that reinforce the visual intuition. In Section~\ref{sec:examplesApp}, we provide several examples from category theory and computer science, including a diagrammatic encoding of the simply typed $\lambda$-calculus. In Section~\ref{sec:soundness and completeness}, we prove completeness—showing that the string diagrams are the free closed symmetric monoidal category generated by a closed monoidal signature. To this end, we establish a normalization and a decomposition result for our diagrams.

\subsection{Related Work} 

There have been some attempts to develop a diagrammatic syntax for closed monoidal categories, most notably the \emph{bubble-and-clasp} notation in~\cite{BaezRosetta} and the \emph{hierarchical string diagrams} in~\cite{DBLP:conf/fscd/Alvarez-Picallo22, ghica2023string}. 
Higher-order mechanisms are represented in both of them as a bubble enclosing a diagram. For example, currying is depicted as the diagrams below.
\[
\includegraphics[scale=0.4]{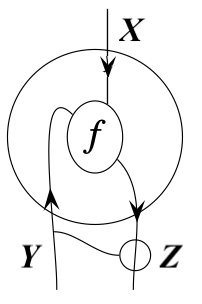} \qquad\qquad\qquad \includegraphics[scale=0.35]{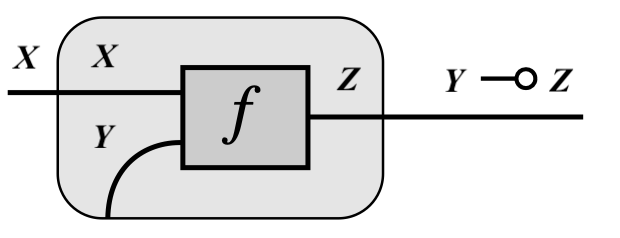}
\]

They are close in spirit to our language, however there are some notable differences. The first diagram binds together the strings for $Y$ and $Z$ via a clasp to form the object $[Y,Z]$, while in the second diagram, the $Y$ wire is curryed away and the output type is annotated with $Y \multimap Z$, denoting the object $[Y,Z]$.

For the case of~\cite{BaezRosetta}, we are not aware of any completeness result with respect to closed symmetric monoidal categories. The language proposed in~\cite{DBLP:conf/fscd/Alvarez-Picallo22, ghica2023string} has been formalised using functorial boxes~\cite{mellies_functorial_2006}. However, the approach is not fully diagrammatic, as it relies on a hybrid syntax combining graphical elements with algebraic notation, rather than a purely visual representation.

The same holds true also for the diagrams for closed monoidal theories appearing in~\cite{garner2009variable, garner2008graphical}.

In~\cite{shulman2020autonomous}, it has been shown that every closed monoidal category fully embeds into a $\ast$-autonomous category via a functor that preserves the closed structure. This result suggests that the graphical language for $\ast$-autonomous categories~\cite{Selinger2009} can, in principle, be used to reason about closed monoidal categories as well. However, this approach does not yield a complete diagrammatic language for the latter: one must take care to avoid diagrammatic manipulations that are valid in the $\ast$-autonomous setting but do not correspond to valid reasoning within the closed monoidal category.

Recent, unpublished work by Willerton~\cite{willerton} provides a more geometric account of non-symmetric, closed monoidal categories. In this language, the internal hom is represented by a surface which surrounds the strings, acting as a boundary.

Finally, although they arise in different contexts, additional instances of string diagrams equipped with graphical mechanisms for managing different kinds of bracketing can be found in~~\cite{acclavio2019proof,bonchi2023deconstructing,bonchi2024diagrammatic,comfort2020sheet,tiurin2025equivalence}.

\paragraph{Acknowledgments} We would like to acknowledge Matt Earnshaw, Alkis Ioannidis, Mario Roman, Pawel Sobocinski and Simon Willerton for several helpful discussions. We are also grateful to the anonymous reviewers of CSL '26 for their insightful comments and feedback.

%% file: ClosedDiags/functorialitySD.tex
\includegraphics{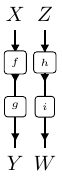}

%% file: ClosedDiags/symmNat1.tex
$\vcenter{\hbox{\includegraphics{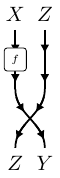}}}$

%% file: ClosedDiags/symmNat2.tex
$\vcenter{\hbox{\includegraphics{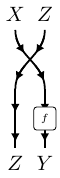}}}$

%% file: ClosedDiags/curryingIntro.tex
\includegraphics{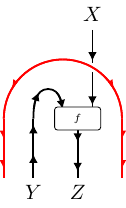}

%% file: sections/presentation3.tex
\section{String Diagrams with Brackets}\label{sec:presentation}

The string diagram language that we present here relies on a kind of context-dependency: certain manipulations of strings can only take place in certain `contexts' so to speak. These contexts represent the closed structure, typically denoted $\hom{A,B}$ or $A\lift B$. Other similar languages also have such contexts, in~\cite{ghica2023string} this means using a functor box, while in~\cite{BaezRosetta} this means using a bubble. 

The key difference here is that our contexts are not machinery that we impose on string diagrams. They are built from special types of strings -- referred to as `bracketing strings' -- and these bracketing strings are subject to manipulations and equations that we might expect from any other string. In fact, the equations on bracketing strings correspond closely to symmetries as well as yankings, poppings and mergings -- the same equations imposed on string diagrams for adjoint equivalences. 

To see how we use these bracketing strings, consider our simplest example: diagram (a) in \autoref{fig:basic}, which represents the identity map on $\hom{A,B}$. Here we read our diagrams from top to bottom. The upwards string -- or `upstring' -- indicates that $A$ is contravariant in the internal hom functor, whilst the `downstring' $B$ is covariant. Note that an upstring may only appear between a bracketing pair of strings, and that there is a left bracketing string, and a right bracketing string.

The reader might also notice that our bracketing strings are directional. This is so that we can nest our brackets and represent double dualisations. So diagram (b) in \autoref{fig:basic} represents the identity on $\hom{\hom{A,B},C}$.

As in string diagrams for monoidal categories, we represent the tensor product by horizontal concatenation, so diagram (c) in \autoref{fig:basic} represents the identity on $A\tensor \hom{B\tensor C, D\tensor E}$.  In this example note that all of the upstrings are on the left and all of the downstrings are on the right. This is a sort of `normal form' for our diagrams that we discuss in \autoref{sec:normal}. Since our language is for closed \emph{symmetric} monoidal categories, it is convenient to allow our upstrings and downstrings to appear in any order, and we add in a syntax for crossing these strings, subject to the usual symmetry equations\footnote{Note that, this manipulation can be considered purely syntactically, but might also be seen as a kind of partially unbiased definition of closed symmetric monoidal in the vein of \cite{leinster2004higher}. It also corresponds to the fact that a closed symmetric monoidal category is both left- and right-closed, but we do not explore that perspective here.}. For example, the identity on $A\tensor \hom{B\tensor C, D\tensor E}$ can also be represented by diagram (d) in \autoref{fig:basic}.

Of course we include the usual boxes to represent morphisms. So, for example, if we have $g\colon A_0\to A_1$ and $f\colon \hom{A_0,B}\to C$ then the composite $
  \hom{A_1,B} \xrightarrow{\hom{g,B}} \hom{A_0,B} \xrightarrow{f} C
$ is represented by diagram (g) of \autoref{fig:basic}. 

Boxes allow us to represent certain morphisms between internal homs. In order to represent all of the \emph{canonical} morphisms found in a closed symmetric monoidal category we have a collection of manipulations that can be performed on diagrams. Locally, these manipulations look like the those of \autoref{fig:local}. However, to prevent the generation of invalid diagrams, these manipulations are -- informally -- subject to the following restrictions.\vspace{0.1cm}

\begin{enumerate}
  \item Bracketing strings must occur in balanced pairs.
  \item Upstrings can only occur between pairs of bracketing strings.
  \item Downstrings can only leave a downward bracketing pairs of strings, if that bracketing pair of strings contains no upstrings. 
  \item Upstrings can only enter a downwards bracketing pair from another, if the outer bracketing pair has no other downstrings.
  \item The corresponding restrictions on upwards bracketing pairs of strings.
\end{enumerate}\vspace{0.1cm}

\begin{figure}[p]
\subfig{1}{
\subfig{0.2}{\ensuremath{\vcenter{\hbox{\scalebox{0.75}{\input{ClosedDiags/Local/14}}}}}}
\subfig{0.2}{\ensuremath{\vcenter{\hbox{\scalebox{0.75}{\input{ClosedDiags/Local/5}}}}}}
\subfig{0.2}{\ensuremath{\vcenter{\hbox{\scalebox{0.75}{\input{ClosedDiags/Local/5a}}}}}}
\subfig{0.2}{\ensuremath{\vcenter{\hbox{\scalebox{0.75}{\input{ClosedDiags/Relations/cupNat7}}}}}}
\subfig{0.2}{\ensuremath{\vcenter{\hbox{\scalebox{0.75}{\input{ClosedDiags/Relations/cupNat8}}}}}}\\[-0.5cm]
\subfig{0.2}{\ensuremath{\vcenter{\hbox{\scalebox{0.75}{\input{ClosedDiags/Local/15}}}}}}
\subfig{0.2}{\ensuremath{\vcenter{\hbox{\scalebox{0.75}{\input{ClosedDiags/Relations/cupNat5a}}}}}}
\subfig{0.2}{\ensuremath{\vcenter{\hbox{\scalebox{0.75}{\input{ClosedDiags/Relations/cupNat6a}}}}}}
\subfig{0.2}{\ensuremath{\vcenter{\hbox{\scalebox{0.75}{\input{ClosedDiags/Relations/cupNat7a}}}}}}
\subfig{0.2}{\ensuremath{\vcenter{\hbox{\scalebox{0.75}{\input{ClosedDiags/Relations/cupNat8a}}}}}}\\[-0.5cm]
\subfig{0.2}{\ensuremath{\vcenter{\hbox{\scalebox{0.75}{\input{ClosedDiags/Local/15a}}}}}}
\subfig{0.2}{\ensuremath{\vcenter{\hbox{\scalebox{0.75}{\input{ClosedDiags/Relations/cupNat10}}}}}}
\subfig{0.2}{\ensuremath{\vcenter{\hbox{\scalebox{0.75}{\input{ClosedDiags/Local/10a}}}}}}
\subfig{0.2}{\ensuremath{\vcenter{\hbox{\scalebox{0.75}{\input{ClosedDiags/Local/12}}}}}}%
\subfig{0.2}{\ensuremath{\vcenter{\hbox{\scalebox{0.75}{\input{ClosedDiags/Local/12a}}}}}}\\[-0.5cm]
\subfig{0.2}{\ensuremath{\vcenter{\hbox{\scalebox{0.75}{\input{ClosedDiags/Local/16}}}}}}
\subfig{0.2}{\ensuremath{\vcenter{\hbox{\scalebox{0.75}{\input{ClosedDiags/Local/11}}}}}}%
\subfig{0.2}{\ensuremath{\vcenter{\hbox{\scalebox{0.75}{\input{ClosedDiags/Local/11a}}}}}}%
\subfig{0.2}{\ensuremath{\vcenter{\hbox{\scalebox{0.75}{\input{ClosedDiags/Local/13}}}}}}%
\subfig{0.2}{\ensuremath{\vcenter{\hbox{\scalebox{0.75}{\input{ClosedDiags/Local/13a}}}}}}}
\caption{Local manipulations.}
\label{fig:local}
\figSpace
\subfig{1}{
\begin{subfigure}{\textwidth}
\subfig{0.33}{\ensuremath{\vcenter{\hbox{\scalebox{0.75}{\input{ClosedDiags/basic1.tex}}}}}\figLab{a}}
\subfig{0.33}{\ensuremath{\vcenter{\hbox{\scalebox{0.75}{\input{ClosedDiags/Example/3.tex}}}}}\figLab{b}}
\subfig{0.33}{\ensuremath{\vcenter{\hbox{\scalebox{0.75}{\input{ClosedDiags/basic1a.tex}}}}}\figLab{c}}
\\
\subfig{0.33}{\ensuremath{\vcenter{\hbox{\scalebox{0.75}{\input{ClosedDiags/Example/2.tex}}}}}\figLab{d}}
\subfig{0.33}{\ensuremath{\vcenter{\hbox{\scalebox{0.75}{\input{ClosedDiags/basic5a.tex}}}}}\figLab{e}}
\subfig{0.33}{\ensuremath{\vcenter{\hbox{\scalebox{0.75}{\input{ClosedDiags/basic4.tex}}}}}\figLab{f}}
\\
\subfig{0.5}{\ensuremath{\vcenter{\hbox{\scalebox{0.75}{\input{ClosedDiags/Example/1.tex}}}}}\figLab{g}}
\subfig{0.5}{\ensuremath{\vcenter{\hbox{\scalebox{0.75}{\input{ClosedDiags/Example/4.tex}}}}}\figLab{h}}
\end{subfigure}}%
\caption{Examples of valid diagrams.}
\label{fig:basic}
\figSpace
\subfig{1}{
\begin{subfigure}{\textwidth}
\subfig{0.5}{\ensuremath{\vcenter{\hbox{\scalebox{0.75}{\input{ClosedDiags/Invalid/2}}}}}\figLab{a}}
\subfig{0.5}{\ensuremath{\vcenter{\hbox{\scalebox{0.75}{\input{ClosedDiags/Invalid/1}}}}}\figLab{b}}
\end{subfigure}}
\caption{Examples of invalid diagrams.}
\label{fig:invalid}
\end{figure}

For example, we can represent the canonical morphism $\hom{A,B}\tensor C\to \hom{A,B\tensor C}$ by diagram (f) in \autoref{fig:basic}. In, for example, $\Set$ this takes the pair $(f,c)$ to the function $f_c$ where $f_c(a)=(f(a),c)$. 

In diagram (e) of \autoref{fig:basic} we see the internal tensor-hom adjunction $\hom{A,\hom{B,C}}\to \hom{A\tensor B, C}$.  This is a use of syntactic sugar, where the crossing of the cupped bracket string represents two separate crossings, followed by a bracket cup. This syntactic sugar is unambiguous, since we later derive the equations in \autoref{fig:derCrossingRel}.

Note that in diagrams (e) and (h) we cross an upstring with a bracketing string. We cannot, however, draw diagram (a) in \autoref{fig:invalid}. This is because the presence of the downstring $C$ \emph{prevents} $A$ from crossing the bracketing string. This should certainly be the case, since there is no canonical morphism $\hom{A, \hom{I,B}\tensor C} \to \hom{I, \hom{I,A \tensor B}\tensor C}$. %

It is also worth noting that the cups in our language must always either cross a bracketing string or follow a bracketing string, and this prevents the occurence of traces. For example, diagram (b) of \autoref{fig:invalid} is not a valid diagram.

Finally, note that the language is designed in such a way that the contents of bracketing pairs of strings can really be thought of as a string in itself -- this is why we use half arrowheads to denote the left and right bracketing strings. 
So, for example, diagram (h) in \autoref{fig:basic} represents the canonical evaluation $\hom{A,B}\tensor \hom{\hom{A,B},C} \to C \xrightarrow{\sim} \hom{I,C}$.

As for traditional string diagrams, we are interested in generating the language starting from a signature, that is a set of symbols $f \colon X \to Y$,  each with their own name and type.

\begin{definition}\label{def:monoidal closed signature}
  A \emph{(unbiased) closed monoidal signature}, often simply denoted by $\sign$, is a quadruple $(\sort, \sign, \arity, \coarity)$ consisting of a set $\sort$ of \emph{sorts}, a set $\sign$ of \emph{generators}, and two functions $\arity, \coarity \colon \sign \to \sort^\sharp$ that assign to each generator its domain and codomain, respectively. Here, $\sort^\sharp$ is the set of strings generated by the first layer of the grammar below
  \[
    \begin{array}{lcl@{\qquad\mid\qquad}l@{\qquad\mid\qquad}l}
    X &\coloneqq& I & AX & [\mathbf{X}]X 
    \\
    X^\ast &\coloneqq& I & A^\ast X^\ast & [\mathbf{X}]^\ast X^\ast
    \\
    \mathbf{X} &\coloneqq& I & X^\ast \mathbf{X} & X \mathbf{X}
  \end{array}
  \]
  where $I$ is the empty string and $A$ is a symbol in $\sort$. Elements of $\sort^\sharp$ are thus bracketed strings of sorts and their formal dual. Brackets $[-]$ and duals $-^\ast$ serve the same role of the brackets and upstrings in the diagrammatic language.

As a convention, we use letters $A,B,C,\ldots$ to denote elements of $\sort$ and $X,Y,Z,\ldots,U,V,W,\ldots$ to denote elements of $\sort^\sharp$.

\end{definition}

Since our string types involve nested brackets, as well as upstrings and downstrings, in the first subsection we give an inductive definition of how we represent each of these types. In the second subsection we give an inductive definition of our terms. Rather than generating our diagrams from the manipulations in \autoref{fig:local}, and discarding the invalid manipulations, we instead introduce a syntax that can \emph{only} generate valid diagrams. In the third subsection we give a collection of equations that we impose on these diagrams such as symmetry, naturality, and yanking conditions. We then prove that additional, expected equations can be derived from these, before proving that the category they generate is, in fact, closed symmetric monoidal.

\subsection{Objects}
Since we now have ``types'' of string -- upstrings, downstrings, upwards brackets, and downwards brackets -- we must first generate all of the identity diagrams in a sensible way. We need to make sure that our bracketing strings are balanced, and that we can generate nested brackets, as well as strings in up- and downward directions. This is done via the following three layer grammar, that is a diagrammatic version of the one in Definition~\ref{def:monoidal closed signature}.
\begin{alignat*}{6}
  \ensuremath{\vcenter{\hbox{\scalebox{0.75}{\input{ClosedDiags/StringGen/down1.tex}}}}}&\coloneqq \quad
    &&\ensuremath{\vcenter{\hbox{\scalebox{0.75}{\input{ClosedDiags/StringGen/blank.tex}}}}}\quad &&\mid \quad 
  &&\ensuremath{\vcenter{\hbox{\scalebox{0.75}{\input{ClosedDiags/StringGen/down3.tex}}}}} \quad &&\mid \quad
  &&\ensuremath{\vcenter{\hbox{\scalebox{0.75}{\input{ClosedDiags/StringGen/down4.tex}}}}}\\
  \ensuremath{\vcenter{\hbox{\scalebox{0.75}{\input{ClosedDiags/StringGen/up1.tex}}}}}&\coloneqq  \quad
    &&\ensuremath{\vcenter{\hbox{\scalebox{0.75}{\input{ClosedDiags/StringGen/blank.tex}}}}}\quad &&\mid \quad 
  &&\ensuremath{\vcenter{\hbox{\scalebox{0.75}{\input{ClosedDiags/StringGen/up3.tex}}}}} \quad &&\mid \quad
  &&\ensuremath{\vcenter{\hbox{\scalebox{0.75}{\input{ClosedDiags/StringGen/up4.tex}}}}}\\
  \ensuremath{\vcenter{\hbox{\scalebox{0.75}{\input{ClosedDiags/StringGen/1.tex}}}}}&\coloneqq  \quad
  &&\ensuremath{\vcenter{\hbox{\scalebox{0.75}{\input{ClosedDiags/StringGen/blank.tex}}}}} \quad &&\mid \quad   
  &&\ensuremath{\vcenter{\hbox{\scalebox{0.75}{\input{ClosedDiags/StringGen/3.tex}}}}} \quad &&\mid \quad &&\ensuremath{\vcenter{\hbox{\scalebox{0.75}{\input{ClosedDiags/StringGen/4.tex}}}}} 
\end{alignat*}
We use solid lines for the generating objects $A \in \sort$ and dotted lines to indicate any object.
Note that here undirected strings -- those strings that appear inside bracketing strings -- are really lists of up- and downstrings that can appear in any order. This does not represent the typical algebraic representation of a closed symmetric monoidal category, in which all of the contravariant terms appear on the left and all of the covariant terms appear on the right: for example, $\hom{A\tensor B, C\tensor D \tensor E}$. In Definition \ref{def:normalised} we define normalised terms and in \autoref{thm:normalisation equivalence} we prove that the existence of unordered terms is purely syntactic.

\subsection{Morphisms}

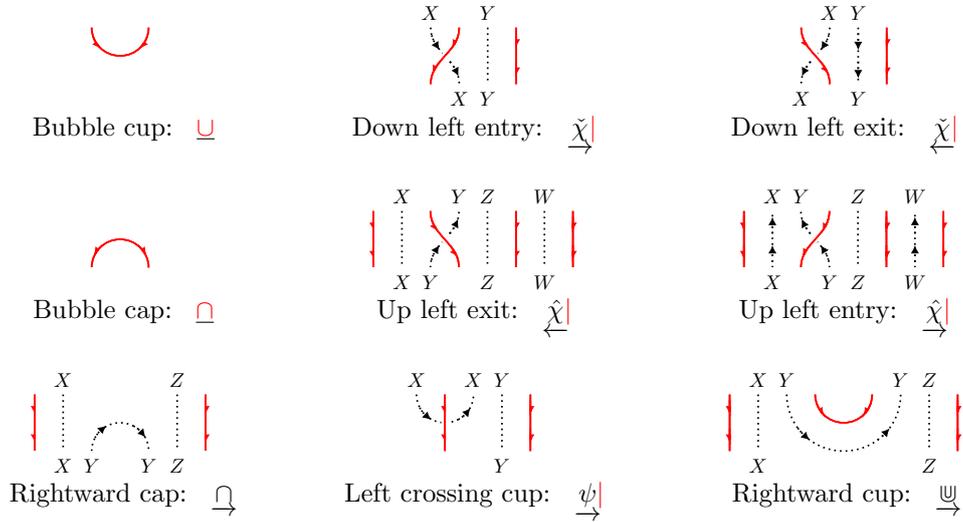
\begin{figure}[h]
  \centering
  \[
  \begin{tabular}{c@{\qquad\qquad}c@{\qquad\qquad}c}
    \ensuremath{\vcenter{\hbox{\scalebox{0.75}{\input{ClosedDiags/Generators/11a.tex}}}}}
    &
    \ensuremath{\vcenter{\hbox{\scalebox{0.75}{\input{ClosedDiags/Generators/1.tex}}}}}
    &
    \ensuremath{\vcenter{\hbox{\scalebox{0.75}{\input{ClosedDiags/Generators/1a.tex}}}}}
    \\
    \diagLabel{Bubble cup: }$\brCup$
    &
    \diagLabel{Down left entry:} $\dXin$
    &
    \diagLabel{Down left exit:} $\dXout$
    \\[0.5cm]
    \ensuremath{\vcenter{\hbox{\scalebox{0.75}{\input{ClosedDiags/Generators/12a.tex}}}}}
    &
    \ensuremath{\vcenter{\hbox{\scalebox{0.75}{\input{ClosedDiags/Generators/5.tex}}}}}
    &
    \ensuremath{\vcenter{\hbox{\scalebox{0.75}{\input{ClosedDiags/Generators/7.tex}}}}}
    \\
    \diagLabel{Bubble cap: }$\brCap$
    &
    \diagLabel{Up left exit: }$\uXout$
    &
    \diagLabel{Up left entry: }$\uXin$
    \\[0.5cm]
    \ensuremath{\vcenter{\hbox{\scalebox{0.75}{\input{ClosedDiags/Generators/9.tex}}}}}
    &
    \ensuremath{\vcenter{\hbox{\scalebox{0.75}{\input{ClosedDiags/Generators/3.tex}}}}}
    &
    \ensuremath{\vcenter{\hbox{\scalebox{0.75}{\input{ClosedDiags/Generators/13.tex}}}}}
    \\
    \diagLabel{Rightward cap:} $\rightCap$
    &
    \diagLabel{Left crossing cup: }$\Xcup$
    &
    \diagLabel{Rightward cup: }$\rightCup$
  \end{tabular}
  \]
\caption{Generators for string diagrams with brackets.}
\label{fig:term generators}
 \end{figure}
At the beginning of the section we gave examples of our string diagrams, and the manipulations and restrictions that allow us to reason about closed symmetric monoidal categories. Rather than defining these string diagrams in terms of local manipulations, and then discarding invalid diagrams, we choose to restrict to generators that can \emph{only} generate valid diagrams.

These are precisely those given in \autoref{fig:term generators}, as well as their reflections in the vertical axis and the respective contravariant generators.

\input{sections/generatingDiagsExplanation.tex}

The full diagrammatic syntax is given by the three-layer grammar in Figure~\ref{fig:full diagrammatic syntax}. Note that, as for strings, the contravariant layer is a reflection of the covariant one. 
Moreover, note that the mixed terms in Figure \ref{fig:full diagrammatic syntax} also include crossings between wires of different directions. These crossings are particularly useful for defining a normalisation procedure that sorts the components of undirected strings.

\begin{figure}[H]
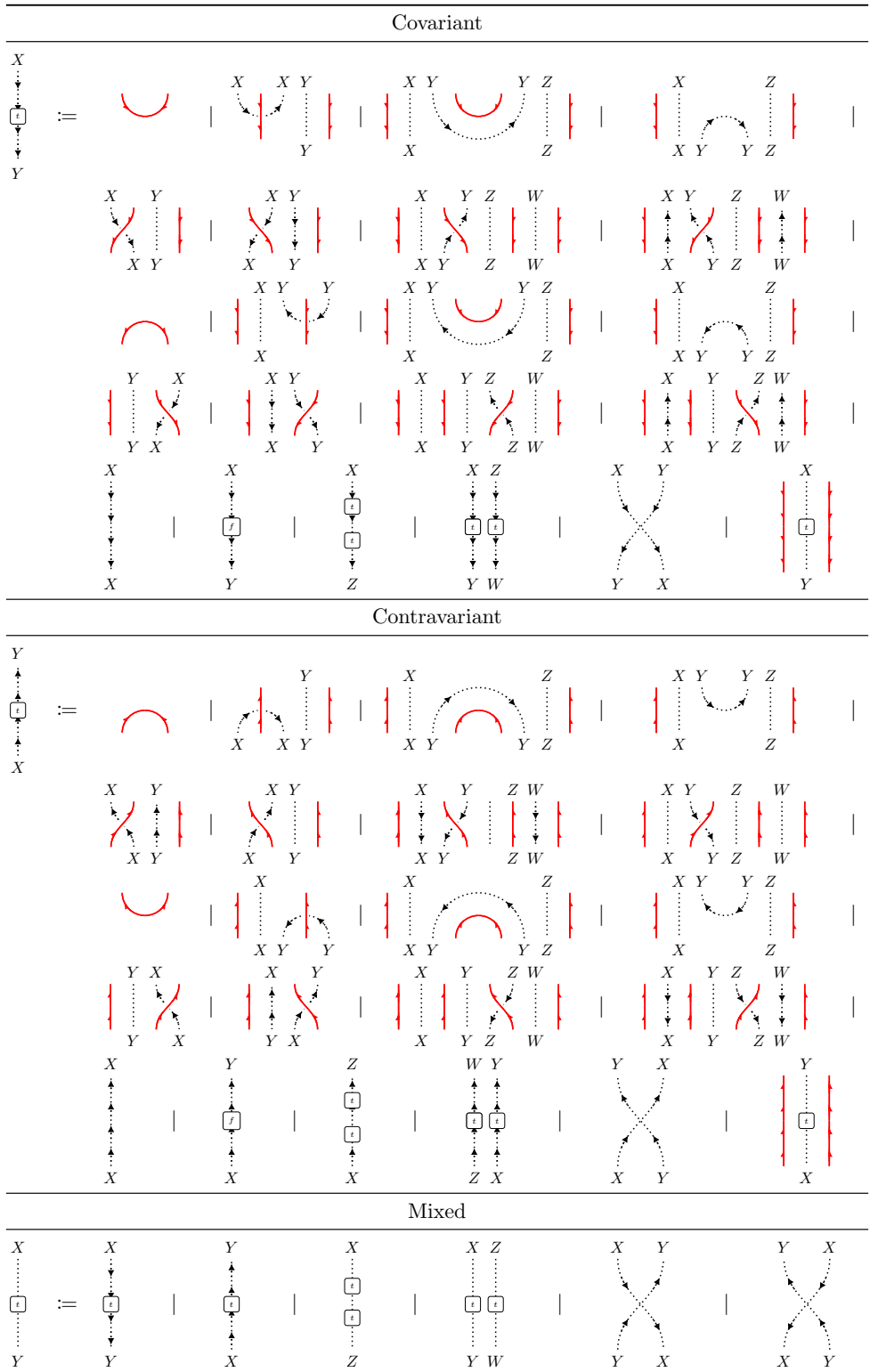

  \begin{center}
    \scalebox{0.9}{
    $\begin{array}{@{}ccc@{\;\;\mid\;\;}c@{\;\;\mid\;\;}c@{\;\;\mid\;\;}c@{\;\;\mid\;\;}c}
        \toprule
        \multicolumn{6}{c}{\text{Covariant}}
        \\
        \midrule
        \ensuremath{\vcenter{\hbox{\scalebox{0.75}{\input{ClosedDiags/CellGen/t.tex}}}}} 
        &
        \coloneqq
        &
        \ensuremath{\vcenter{\hbox{\scalebox{0.75}{\input{ClosedDiags/Generators/11a.tex}}}}}
        &
        \ensuremath{\vcenter{\hbox{\scalebox{0.75}{\input{ClosedDiags/Generators/3.tex}}}}}
        &
        \ensuremath{\vcenter{\hbox{\scalebox{0.75}{\input{ClosedDiags/Generators/13.tex}}}}}
        &
        \ensuremath{\vcenter{\hbox{\scalebox{0.75}{\input{ClosedDiags/Generators/9.tex}}}}}
        \\
        & & 
        \ensuremath{\vcenter{\hbox{\scalebox{0.75}{\input{ClosedDiags/Generators/1.tex}}}}}
        &
        \ensuremath{\vcenter{\hbox{\scalebox{0.75}{\input{ClosedDiags/Generators/1a.tex}}}}}
        &
        \ensuremath{\vcenter{\hbox{\scalebox{0.75}{\input{ClosedDiags/Generators/5.tex}}}}}
        &
        \ensuremath{\vcenter{\hbox{\scalebox{0.75}{\input{ClosedDiags/Generators/7.tex}}}}}
        \\
        & &
        \ensuremath{\vcenter{\hbox{\scalebox{0.75}{\input{ClosedDiags/Generators/12a.tex}}}}}
        &
        \ensuremath{\vcenter{\hbox{\scalebox{0.75}{\input{ClosedDiags/YReflectedGen/3.tex}}}}}
        &
        \ensuremath{\vcenter{\hbox{\scalebox{0.75}{\input{ClosedDiags/YReflectedGen/13.tex}}}}}
        &
        \ensuremath{\vcenter{\hbox{\scalebox{0.75}{\input{ClosedDiags/YReflectedGen/9.tex}}}}}
        \\
        & & 
        \ensuremath{\vcenter{\hbox{\scalebox{0.75}{\input{ClosedDiags/YReflectedGen/1.tex}}}}}
        &
        \ensuremath{\vcenter{\hbox{\scalebox{0.75}{\input{ClosedDiags/YReflectedGen/1a.tex}}}}}
        &
        \ensuremath{\vcenter{\hbox{\scalebox{0.75}{\input{ClosedDiags/YReflectedGen/5.tex}}}}}
        &
        \ensuremath{\vcenter{\hbox{\scalebox{0.75}{\input{ClosedDiags/YReflectedGen/7.tex}}}}}
        \\
        & &
        \multicolumn{4}{@{}l}{
            \begin{array}{c@{\qquad\mid\qquad}c@{\qquad\mid\qquad}c@{\qquad\mid\qquad}c@{\qquad\mid\qquad}c@{\qquad\mid\qquad}c}
                    \ensuremath{\vcenter{\hbox{\scalebox{0.75}{\input{ClosedDiags/CellGen/arrow.tex}}}}}
                    &
                    \ensuremath{\vcenter{\hbox{\scalebox{0.75}{\input{ClosedDiags/CellGen/f.tex}}}}}
                    &
                    \ensuremath{\vcenter{\hbox{\scalebox{0.75}{\input{ClosedDiags/CellGen/sequence.tex}}}}}
                    &
                    \ensuremath{\vcenter{\hbox{\scalebox{0.75}{\input{ClosedDiags/CellGen/parallel.tex}}}}}
                    &
                    \ensuremath{\vcenter{\hbox{\scalebox{0.75}{\input{ClosedDiags/CellGen/cross.tex}}}}}
                    &
                    \ensuremath{\vcenter{\hbox{\scalebox{0.75}{\input{ClosedDiags/CellGen/tBracket.tex}}}}}       
            \end{array}
        }
        \\
        \midrule
        \multicolumn{6}{c}{\text{Contravariant}}
        \\
        \midrule
        \ensuremath{\vcenter{\hbox{\scalebox{0.75}{\input{ClosedDiags/CellGen/tUp.tex}}}}} 
        &
        \coloneqq
        &
        \ensuremath{\vcenter{\hbox{\scalebox{0.75}{\input{ClosedDiags/XReflectedGen/12a.tex}}}}}
        &
        \ensuremath{\vcenter{\hbox{\scalebox{0.75}{\input{ClosedDiags/XReflectedGen/3.tex}}}}}
        &
        \ensuremath{\vcenter{\hbox{\scalebox{0.75}{\input{ClosedDiags/XReflectedGen/13.tex}}}}}
        &
        \ensuremath{\vcenter{\hbox{\scalebox{0.75}{\input{ClosedDiags/XReflectedGen/9.tex}}}}}
        \\
        & & 
        \ensuremath{\vcenter{\hbox{\scalebox{0.75}{\input{ClosedDiags/XReflectedGen/1.tex}}}}}
        &
        \ensuremath{\vcenter{\hbox{\scalebox{0.75}{\input{ClosedDiags/XReflectedGen/1a.tex}}}}}
        &
        \ensuremath{\vcenter{\hbox{\scalebox{0.75}{\input{ClosedDiags/XReflectedGen/5.tex}}}}}
        &
        \ensuremath{\vcenter{\hbox{\scalebox{0.75}{\input{ClosedDiags/XReflectedGen/7.tex}}}}}
        \\
        & &
        \ensuremath{\vcenter{\hbox{\scalebox{0.75}{\input{ClosedDiags/XYReflectedGen/11a.tex}}}}}
        &
        \ensuremath{\vcenter{\hbox{\scalebox{0.75}{\input{ClosedDiags/XYReflectedGen/3.tex}}}}}
        &
        \ensuremath{\vcenter{\hbox{\scalebox{0.75}{\input{ClosedDiags/XYReflectedGen/13.tex}}}}}
        &
        \ensuremath{\vcenter{\hbox{\scalebox{0.75}{\input{ClosedDiags/XYReflectedGen/9.tex}}}}}
        \\
        & & 
        \ensuremath{\vcenter{\hbox{\scalebox{0.75}{\input{ClosedDiags/XYReflectedGen/1.tex}}}}}
        &
        \ensuremath{\vcenter{\hbox{\scalebox{0.75}{\input{ClosedDiags/XYReflectedGen/1a.tex}}}}}
        &
        \ensuremath{\vcenter{\hbox{\scalebox{0.75}{\input{ClosedDiags/XYReflectedGen/5.tex}}}}}
        &
        \ensuremath{\vcenter{\hbox{\scalebox{0.75}{\input{ClosedDiags/XYReflectedGen/7.tex}}}}}
        \\
        & &
        \multicolumn{4}{@{}l}{
            \begin{array}{c@{\qquad\mid\qquad}c@{\qquad\mid\qquad}c@{\qquad\mid\qquad}c@{\qquad\mid\qquad}c@{\qquad\mid\qquad}c}
                    \ensuremath{\vcenter{\hbox{\scalebox{0.75}{\input{ClosedDiags/CellGen/arrowUp.tex}}}}}
                    &
                    \ensuremath{\vcenter{\hbox{\scalebox{0.75}{\input{ClosedDiags/CellGen/fUp.tex}}}}}
                    &
                    \ensuremath{\vcenter{\hbox{\scalebox{0.75}{\input{ClosedDiags/CellGen/sequenceUp.tex}}}}}
                    &
                    \ensuremath{\vcenter{\hbox{\scalebox{0.75}{\input{ClosedDiags/CellGen/parallelUp.tex}}}}}
                    &
                    \ensuremath{\vcenter{\hbox{\scalebox{0.75}{\input{ClosedDiags/CellGen/crossUp.tex}}}}}
                    &
                    \ensuremath{\vcenter{\hbox{\scalebox{0.75}{\input{ClosedDiags/CellGen/tBracketUp.tex}}}}}       
            \end{array}
        }
        \\
        \midrule
        \multicolumn{6}{c}{\text{Mixed}}
        \\
        \midrule
        \ensuremath{\vcenter{\hbox{\scalebox{0.75}{\input{ClosedDiags/CellGen/tUndirected.tex}}}}}
        &
        \coloneqq
        &
        \multicolumn{4}{@{}l}{
            \begin{array}{c@{\qquad\mid\qquad}c@{\qquad\mid\qquad}c@{\qquad\mid\qquad}c@{\qquad\mid\qquad}c@{\qquad\mid\qquad}c}
                \ensuremath{\vcenter{\hbox{\scalebox{0.75}{\input{ClosedDiags/CellGen/tUndirected1.tex}}}}}
                &
                \ensuremath{\vcenter{\hbox{\scalebox{0.75}{\input{ClosedDiags/CellGen/tUndirected2.tex}}}}}
                &
                \ensuremath{\vcenter{\hbox{\scalebox{0.75}{\input{ClosedDiags/CellGen/tUndirectedSequence.tex}}}}}
                &
                \ensuremath{\vcenter{\hbox{\scalebox{0.75}{\input{ClosedDiags/CellGen/tUndirectedParallel.tex}}}}}
                &
                \ensuremath{\vcenter{\hbox{\scalebox{0.75}{\input{ClosedDiags/CellGen/tUndirected3.tex}}}}}
                &
                \ensuremath{\vcenter{\hbox{\scalebox{0.75}{\input{ClosedDiags/CellGen/tUndirected4.tex}}}}}
            \end{array}
        }
        \\
        \bottomrule
    \end{array}$}
  \end{center}
    \vspace{-10pt}\caption{Full syntax of string diagrams with brackets.}
    \label{fig:full diagrammatic syntax}
\end{figure}

For the sake of readability we include syntactic sugar for the bracket cup and cap diagrams. In order to do this we first need a notion of normal form.

\begin{definition}
  \label{def:normalised}
  An undirected string $\ensuremath{\vcenter{\hbox{\scalebox{0.75}{\input{ClosedDiags/StringGen/1.tex}}}}}$ is normalised if it is generated as $\ensuremath{\vcenter{\hbox{\scalebox{0.75}{\input{ClosedDiags/StringGen/5.tex}}}}}$.
\end{definition}

\begin{definition}
\label{def:normaliser}
 The \emph{normaliser} term $\nu$, depicted below, sorts all strings within brackets so that the upstrings are on the left and the downstrings are on the right. 
 \[ \ensuremath{\vcenter{\hbox{\input{ClosedDiags/Normaliser/1a.tex}}}} \]
 Formally, we define the normaliser inductively as follows:
\[\begin{array}{rcl@{\qquad\quad}rcl@{\qquad\quad}rcl}
    \ensuremath{\vcenter{\hbox{\input{ClosedDiags/Normaliser/2d}}}} &\coloneqq& \ensuremath{\vcenter{\hbox{\input{ClosedDiags/StringGen/blank.tex}}}}
    &
    \ensuremath{\vcenter{\hbox{\input{ClosedDiags/Normaliser/5}}}} &\coloneqq& \ensuremath{\vcenter{\hbox{\input{ClosedDiags/Normaliser/5a}}}}
    &
    \ensuremath{\vcenter{\hbox{\input{ClosedDiags/Normaliser/6}}}} &\coloneqq& \ensuremath{\vcenter{\hbox{\input{ClosedDiags/Normaliser/6a}}}}
    \\
    \\
    & & &
    \ensuremath{\vcenter{\hbox{\input{ClosedDiags/Normaliser/7}}}} &\coloneqq& \ensuremath{\vcenter{\hbox{\input{ClosedDiags/Normaliser/7a}}}}
    &
    \ensuremath{\vcenter{\hbox{\input{ClosedDiags/Normaliser/8}}}} &\coloneqq& \ensuremath{\vcenter{\hbox{\input{ClosedDiags/Normaliser/8a}}}}
    \\
    \\
    & & &
    \ensuremath{\vcenter{\hbox{\input{ClosedDiags/Normaliser/4b}}}} &\coloneqq& \ensuremath{\vcenter{\hbox{\input{ClosedDiags/Normaliser/4c}}}}
    &
    \ensuremath{\vcenter{\hbox{\input{ClosedDiags/Normaliser/3b}}}} &\coloneqq& \ensuremath{\vcenter{\hbox{\input{ClosedDiags/Normaliser/3c}}}}
  \end{array}\]
\end{definition}
\newpage
\begin{definition}\label{fig:syntSugar}
  For the sake of keeping our language geometrically intuitive, we introduce some syntactic sugar for strings that leave and enter cups and caps, respectively. 
  \[
  \begin{array}{rcl@{\qquad\qquad}rcl}
    \ensuremath{\vcenter{\hbox{\scalebox{0.75}{\input{ClosedDiags/Sugar/cupCrossing5a}}}}} &\coloneqq&  \ensuremath{\vcenter{\hbox{\scalebox{0.75}{\input{ClosedDiags/Sugar/cupCrossing6a}}}}}
    &
    \ensuremath{\vcenter{\hbox{\scalebox{0.75}{\input{ClosedDiags/Sugar/capCrossing5a}}}}} &\coloneqq&  \ensuremath{\vcenter{\hbox{\scalebox{0.75}{\input{ClosedDiags/Sugar/capCrossing6a}}}}}
    \\
    \ensuremath{\vcenter{\hbox{\scalebox{0.75}{\input{ClosedDiags/Sugar/cupCrossing1a}}}}} &\coloneqq&  \ensuremath{\vcenter{\hbox{\scalebox{0.75}{\input{ClosedDiags/Sugar/cupCrossing2a}}}}}
    &
    \ensuremath{\vcenter{\hbox{\scalebox{0.75}{\input{ClosedDiags/Sugar/capCrossing1a}}}}} &\coloneqq&  \ensuremath{\vcenter{\hbox{\scalebox{0.75}{\input{ClosedDiags/Sugar/capCrossing2a}}}}}
  \end{array}
  \]
  Where here $\nu^{-1}$ is the reflection of $\nu$ in the horizontal axis. By the equations that we present in Section~\ref{ssec:equations}, we can show that $\nu^{-1}$ is the inverse of $\nu$.
\end{definition}

Note that by the additional naturality and symmetry equations derived in \cref{fig:derNaturalities,fig:derCrossingRel} we can show that leaving the bracketing cup on the left is equal to leaving the bracketing cup on the right, and similarly for the bracketing cap. In other words, the syntactic sugar given above is unambiguous. 

\subsection{Equations}\label{ssec:equations}
As highlighted above, reasoning in this diagrammatic language is context dependent: the manipulations that can be performed on a particular string may depend on the contents certain pairs of bracketing strings. However, for the sake of readability here we present the equations `locally'. By this we mean that, if a certain local sequence \emph{can} occur validly then we can replace it locally. Let us consider the following composite, built from the internal tensor-hom adjunction and written in the usual algebraic language.
\[
  \color{lightgray}\hom{A \color{black}\tensor B, [I \color{lightgray},C]}\color{black} \xrightarrow{\sim} \color{lightgray} \hom{A\color{black},\hom{B\color{lightgray},C}\color{lightgray}}\color{black}\xrightarrow{\sim}\color{lightgray} \hom{A\color{black}\tensor B, \hom{I\color{gray},C}\color{lightgray}}
\]
This composite is equal to the identity, and this holds regardless of what $A$, $B$ and $C$ are. We can really just see this as local rewrite of the substrings highlighted in black.

We understand that this composite of rewrites should be the identity, even though the substrings themselves aren't valid terms. The same is true of the manipulations in our string diagrams, which we treat as rewrites via a congruence. For the sake of brevity it may be assumed that for each given equation, the reflections in the horizontal and vertical axes also hold.

\begin{definition}\label{def:C Sigma}
  Given a (unbiased) closed monoidal signature $(\sort,\sign)$, we denote as $\freeCat$ the category whose objects are elements of $\sort^\sharp$ and morphisms are \emph{covariant terms} quotiented by the minimal congruence $\cong$ generated by the equations of SMCs and those in \cref{fig:equations}.
\end{definition}

\begin{remark}
    For a detailed account of how $\freeCat$ and, in particular, $\cong$ are constructed, see Appendix~\ref{app:formal c sigma}.
\end{remark}

\begin{figure}
  \subfig{1}{
  \subfig{0.5}{\ensuremath{\vcenter{\hbox{\scalebox{0.75}{\input{ClosedDiags/Relations/SymCupEntry1.tex}}}}} = \ensuremath{\vcenter{\hbox{\scalebox{0.75}{\input{ClosedDiags/Relations/SymCupEntry2.tex}}}}}}%
        \subfig{0.5}{
        \ensuremath{\vcenter{\hbox{\scalebox{0.75}{\input{ClosedDiags/Relations/SymCap1.tex}}}}}  = \ensuremath{\vcenter{\hbox{\scalebox{0.75}{\input{ClosedDiags/Relations/SymCap2.tex}}}}}}\\[10pt]
        \subfig{0.25}{
            \ensuremath{\vcenter{\hbox{\scalebox{0.75}{\input{ClosedDiags/Relations/EntryExit1.tex}}}}} =  \ensuremath{\vcenter{\hbox{\scalebox{0.75}{\input{ClosedDiags/Relations/EntryExit2.tex}}}}}}
        \subfig{0.25}{
            \ensuremath{\vcenter{\hbox{\scalebox{0.75}{\input{ClosedDiags/Relations/EntryExit1A.tex}}}}} =  \ensuremath{\vcenter{\hbox{\scalebox{0.75}{\input{ClosedDiags/Relations/EntryExit2A.tex}}}}}}
        \subfig{0.25}{
        \ensuremath{\vcenter{\hbox{\scalebox{0.75}{\input{ClosedDiags/Relations/EntryExit3.tex}}}}} =  \ensuremath{\vcenter{\hbox{\scalebox{0.75}{\input{ClosedDiags/Relations/EntryExit4.tex}}}}}}
        \subfig{0.25}{
        \ensuremath{\vcenter{\hbox{\scalebox{0.75}{\input{ClosedDiags/Relations/EntryExit3A.tex}}}}} =  \ensuremath{\vcenter{\hbox{\scalebox{0.75}{\input{ClosedDiags/Relations/EntryExit4A.tex}}}}}}
        \caption{Crossing equations.}
        \label{fig:crossings}
  }\figSpace
  \subfig{1}{
        \subfig{0.5}{
            \ensuremath{\vcenter{\hbox{\scalebox{0.75}{\input{ClosedDiags/Relations/Natural/3b}}}}} =  \ensuremath{\vcenter{\hbox{\scalebox{0.75}{\input{ClosedDiags/Relations/Natural/3c}}}}}
        }%
        \subfig{0.5}{
                \centering
            \ensuremath{\vcenter{\hbox{\scalebox{0.75}{\input{ClosedDiags/Relations/Natural/6b}}}}} =  \ensuremath{\vcenter{\hbox{\scalebox{0.75}{\input{ClosedDiags/Relations/Natural/6c}}}}}
        }\\[10pt]
        \subfig{1}{
            \ensuremath{\vcenter{\hbox{\scalebox{0.75}{\input{ClosedDiags/Relations/cupNat1a}}}}} =   \ensuremath{\vcenter{\hbox{\scalebox{0.75}{\input{ClosedDiags/Relations/cupNat2a}}}}}}
            \caption{Naturality equations.}
        \label{fig:naturality}
  }\figSpace
  \subfig{1}{
        \begin{subfigure}{0.3\textwidth}
            \centering
        \ensuremath{\vcenter{\hbox{\scalebox{0.75}{\input{ClosedDiags/Relations/yank1.tex}}}}} \!\!=\!\!
        \ensuremath{\vcenter{\hbox{\scalebox{0.75}{\input{ClosedDiags/Relations/yank2.tex}}}}} 
        \end{subfigure}%
            \begin{subfigure}{0.3\textwidth}
                \centering
        \ensuremath{\vcenter{\hbox{\scalebox{0.75}{\input{ClosedDiags/Relations/yank3.tex}}}}} \!\!=\!\!
        \ensuremath{\vcenter{\hbox{\scalebox{0.75}{\input{ClosedDiags/Relations/yank4.tex}}}}} 
        \end{subfigure}%
        \begin{subfigure}{0.4\textwidth}
                \centering
        \ensuremath{\vcenter{\hbox{\scalebox{0.75}{\input{ClosedDiags/Relations/yank5.tex}}}}} \!\!=\!\!  \ensuremath{\vcenter{\hbox{\scalebox{0.75}{\input{ClosedDiags/Relations/yank6.tex}}}}}
        \end{subfigure}
            \caption{Yanking equations.}
        \label{fig:yanking}
  }\figSpace
  \subfig{1}{
        \begin{subfigure}{0.3\textwidth}
            \centering
        \ensuremath{\vcenter{\hbox{\scalebox{0.75}{\input{ClosedDiags/Relations/pop1a.tex}}}}} \!\!= \ensuremath{\vcenter{\hbox{\scalebox{0.75}{\input{ClosedDiags/StringGen/blank.tex}}}}}%
        \end{subfigure}%
            \begin{subfigure}{0.3\textwidth}
                \centering
        \ensuremath{\vcenter{\hbox{\scalebox{0.75}{\input{ClosedDiags/Rel/merge1C.tex}}}}}  \!\!=\!    \ensuremath{\vcenter{\hbox{\scalebox{0.75}{\input{ClosedDiags/Rel/merge2C.tex}}}}}
        \end{subfigure} %
        \begin{subfigure}{0.3\textwidth}
                \centering
        \ensuremath{\vcenter{\hbox{\scalebox{0.75}{\input{ClosedDiags/Rel/bubbleYank1.tex}}}}} \!\!=\!\!   \ensuremath{\vcenter{\hbox{\scalebox{0.75}{\input{ClosedDiags/Rel/bubbleYank2.tex}}}}}
        \end{subfigure}%
        \caption{Bracket pop, merge and yanking equations.}
        \label{fig:bubble}
  }
    \caption{Equations for string diagrams with brackets.}
    \label{fig:equations}
\end{figure}

From the congruence $\cong$ we can derive a number of other equations. We do so for two reasons: firstly, to show that our geometric intuition is correct, and that almost all manipulations that we expect to be possible, are in fact possible; secondly, to provide the basis for the proofs in \autoref{sec:soundness and completeness} where we show that our language is, in fact, the free closed symmetric monoidal category on a given closed monoidal signature.
\input{figures/derEquations1.tex}
\input{figures/derEquations2.tex}
\newpage
\begin{theorem}
  \label{thm:derivedRelations}
  The equations in \cref{fig:derEquations} hold.
\end{theorem}
\begin{proof}
    The first equation in \cref{fig:derNaturalities} follows from the sequence of equal diagrams below. The lefthand side equals diagram $(a)$ by yanking, $(a)\cong(b)\cong(c)$ by naturality equations, and $(c)$ equals the right hand side by yanking again. 
\input{figures/derNatProof1.tex}
The second claim follows by the equal diagrams below.
\input{figures/derNatProof2.tex}
    The first equation in \autoref{fig:derCapRel} holds by the sequence of equalities below. We have that $(a)\cong(b)$ and $(b)\cong(c)$ by the bubble pop and crossing equations, respectively. Then $(c)\cong(d)$ and $(d)\cong(e)$ by the yanking equations.
\input{figures/derCapBubbleCross}
    The second equation in \autoref{fig:derCapRel} follows from the diagrams below. We know that $(b)\cong(c)$ and $(d)\cong(e)$ from the previous equation. All other equalities are crossing equations or naturalities from \autoref{fig:derNaturalities}.
    \input{figures/derCapSlide.tex}
    To see that the first claim of \autoref{fig:derCrossingRel} holds, consider the first sequence of equal diagrams below. We have that $(a)\cong(b)$ by the yanking equations, $(b)\cong(c)$ by \autoref{fig:derCapRel}, $(c)\cong(d)$ by the symmetry equations, $(d)\cong(e)$ by the symmetry equations, and $(e)\cong(f)$ by the yanking equations.
    \input{figures/derCrossingProofs1.tex}
    To see that the second claim holds consider the second diagram below. We have that $(a)\cong(b)$ by the crossing equations, $(b)\cong(c)$ by the first claim, and $(c)\cong(d)$ by the crossing equations.
    \input{figures/derCrossingProofs2.tex}
\end{proof}

\begin{lemma}
\label{lemma:bigCupSplit}
The following derived equation holds.
\begin{figure}[H]
\centering\ensuremath{\vcenter{\hbox{\scalebox{0.75}{\input{ClosedDiags/Generators/13}}}}} $\cong$ \ensuremath{\vcenter{\hbox{\scalebox{0.75}{\input{ClosedDiags/BigCupSplit/1}}}}}
\end{figure}
\end{lemma}
\begin{proof}
  This follows from the diagrams in \autoref{fig:bigCupSplit}. The equality $(a)\cong(b)$ holds by the syntactic sugar in \autoref{fig:syntSugar} and \autoref{fig:derCupNatural}. The equality $(b)\cong(c)$ holds by the third yanking equation. The equality $(c)\cong(d)$ holds by the naturality equation and the yanking equation. Then (d) is equal to the right hand side of the equation above by the bubble merge equation.
  \end{proof}

  \begin{theorem}
  The equations in \cref{fig:derRightwardCup} hold.
\end{theorem}

\begin{proof}
    To prove the claim in \autoref{fig:derCupInterchange} note that the first equation in \autoref{fig:derCapRel} allows us to treat caps like terms. So by the naturality axioms we can pass a cap past the cup. Then the equations hold analogously to the proof of \autoref{fig:derNaturalities}: firstly by yanking to prove the case for down entries and up exits, and then the case for down exits and up entries follow from the crossing equations. 
    
    To prove the claim in \autoref{fig:derSmoothing} consider the sequence of diagrams in \autoref{fig:derSmoothingProof}. We have that $(a)$ is syntactic sugar for $(b)$. Then $(b)\cong(c)$ and $(c)\cong(d)$ by the interchange law in \autoref{fig:derCupInterchange} and naturality respectively. Finally, $(d)\cong(e)$ by yanking and $(f)$ is syntactic sugar for $(e)$.

    We prove the claim in \autoref{fig:derCupNatural} by \cref{lemma:bigCupSplit}, the naturality of the crossing cup and the naturalities in \autoref{fig:derNaturalities}. 
    
    Finally, we prove the claim in \autoref{fig:derSymCup} by \autoref{lemma:bigCupSplit}, and the sequence of equalities in \autoref{fig:bigCupSym}.
\end{proof}
    \input{figures/derSmoothingProof.tex}

\begin{figure}[p]
    \subfig{0.5}{\ensuremath{\vcenter{\hbox{\scalebox{0.75}{\input{ClosedDiags/BigCupSplit/6.tex}}}}}\figLab{a}}  
    \subfig{0.5}{\ensuremath{\vcenter{\hbox{\scalebox{0.75}{\input{ClosedDiags/BigCupSplit/7.tex}}}}}\figLab{b}}\figSpace
    \subfig{0.5}{\ensuremath{\vcenter{\hbox{\scalebox{0.75}{\input{ClosedDiags/BigCupSplit/8a.tex}}}}}\figLab{c}}
    \subfig{0.5}{\ensuremath{\vcenter{\hbox{\scalebox{0.75}{\input{ClosedDiags/BigCupSplit/9.tex}}}}}\figLab{d}}  
    \caption{The proof of symmetry for the rightward cup.}
    \label{fig:bigCupSym}
  \end{figure}
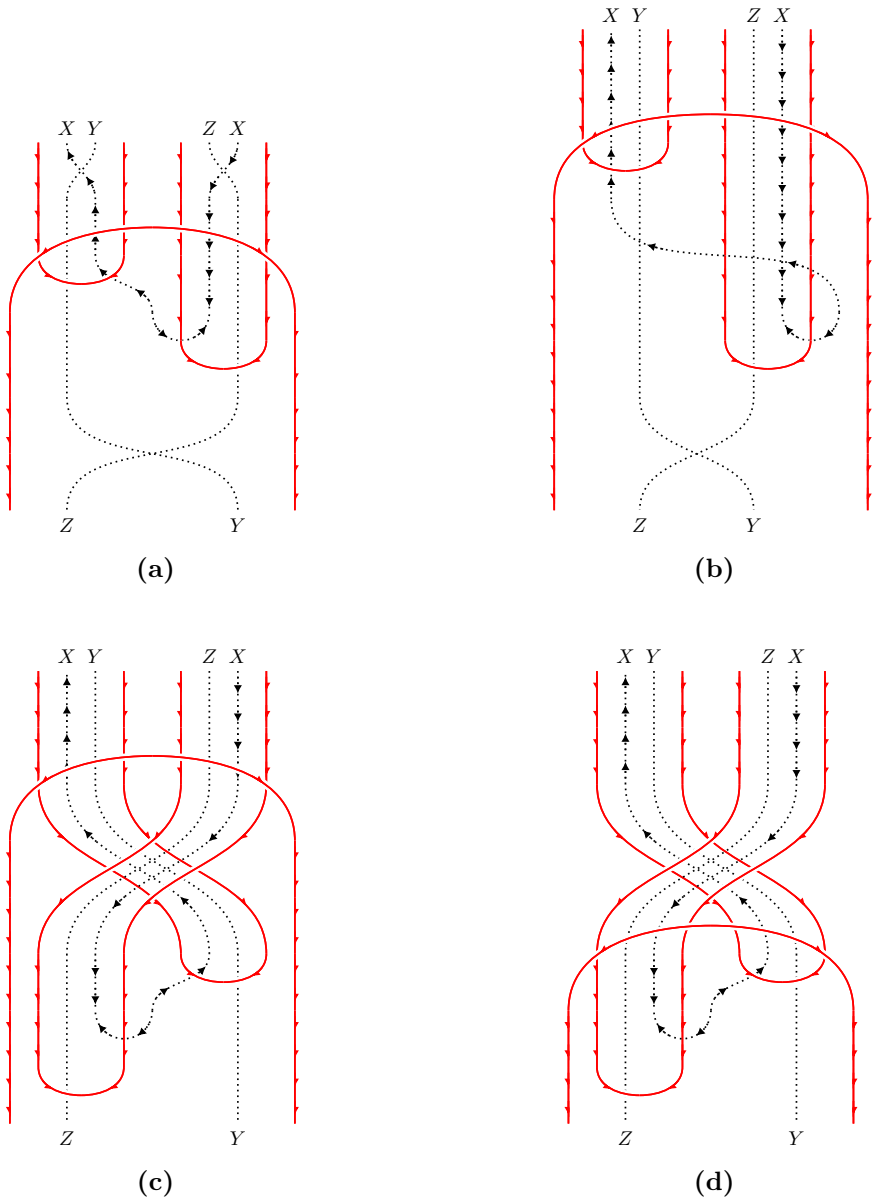

\newpage
After presenting the syntax and equations of our diagrammatic calculus, we are now ready to prove that $\freeCat$ is, in fact, a closed symmetric monoidal category.
\begin{theorem}\label{thm:string diagrams are closed monoidal}
  $\freeCat$ is a closed symmetric monoidal category.
\end{theorem}
\begin{proof}
  First, observe that $\freeCat$ is a symmetric monoidal category. The tensor product $\otimes$ is given by parallel composition of diagrams and symmetries by crossing of wires. Moreover, by definition of $\freeCat$, diagrams are subject to the laws of SMCs.

  To show that it is also closed, we define the right adjoint functor $[Y, -] \colon \freeCat \to \freeCat$ to the tensor product functor $Y \otimes - \colon \freeCat \to \freeCat$, as follows. 
  
  On objects it maps $X$ to $[Y^\ast X]$ and on morphisms it maps $\ensuremath{\vcenter{\hbox{\scalebox{0.75}{\input{ClosedDiags/tGeneral.tex}}}}}$ to $\ensuremath{\vcenter{\hbox{\scalebox{0.75}{\input{ClosedDiags/rightAdjArrows.tex}}}}}$.

  To show that $Y \otimes - \dashv  [Y, -]$, we define the unit $\eta_X$ of the adjunction as the \emph{coevaluation} diagram below on the left and the counit $\epsilon_X$ as the \emph{evaluation} diagram below on the right.
  \begin{equation}%
      \eta_X \coloneqq \ensuremath{\vcenter{\hbox{\scalebox{0.75}{\input{ClosedDiags/WellDef/coevalXY.tex}}}}}
      \qquad \qquad
      \epsilon_X \coloneqq \ensuremath{\vcenter{\hbox{\scalebox{0.75}{\input{ClosedDiags/WellDef/evalXY.tex}}}}}
  \end{equation}
  By the equations in~\autoref{fig:derNaturalities} we can slide diagrams across brackets, and thus naturalities of $\eta_X$ and $\epsilon_X$, displayed below, are easily verified.
  \[
  \ensuremath{\vcenter{\hbox{\scalebox{0.75}{\input{ClosedDiags/CoUnitNat/coevalXY.tex}}}}} \cong \ensuremath{\vcenter{\hbox{\scalebox{0.75}{\input{ClosedDiags/CoUnitNat/coevalXY2.tex}}}}}
  \qquad\qquad    
  \ensuremath{\vcenter{\hbox{\scalebox{0.75}{\input{ClosedDiags/CoUnitNat/evalXY.tex}}}}} \cong \ensuremath{\vcenter{\hbox{\scalebox{0.75}{\input{ClosedDiags/CoUnitNat/evalXY2.tex}}}}}
  \]

  We defer the proof of the triangular identities to the next section, since they arise as special cases of the currying-uncurrying bijection in Example~\ref{example:currying}, instantiated with $f,g = id_{X\otimes Y}$.
\end{proof}%

%% file: ClosedDiags/StringGen/down1.tex
\includegraphics{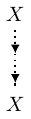}

%% file: ClosedDiags/StringGen/blank.tex
\includegraphics{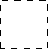}

%% file: ClosedDiags/StringGen/down3.tex
\includegraphics{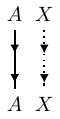}

%% file: ClosedDiags/StringGen/down4.tex
\includegraphics{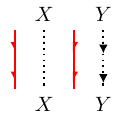}

%% file: ClosedDiags/StringGen/up1.tex
\includegraphics{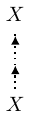}

%% file: ClosedDiags/StringGen/up3.tex
\includegraphics{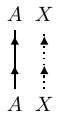}

%% file: ClosedDiags/StringGen/up4.tex
\includegraphics{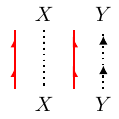}

%% file: ClosedDiags/StringGen/1.tex
\includegraphics{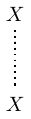}

%% file: ClosedDiags/StringGen/3.tex
\includegraphics{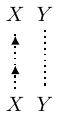}

%% file: ClosedDiags/StringGen/4.tex
\includegraphics{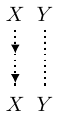}

%% file: sections/generatingDiagsExplanation.tex
In \autoref{lemma:decomposition} we give an account of how -- using the equations we impose in \autoref{fig:equations} -- these generators can be decomposed into the usual evaluation and coevaluation maps of a closed symmetric monoidal category. For the sake of intuition, here we give an account of what each of these diagrams represents in $\Set$.
\begin{longtable}{m{0.25\textwidth} p{0.7\textwidth}}
    \midrule
    \\
    \ensuremath{\vcenter{\hbox{\scalebox{0.75}{\input{ClosedDiags/Generators/11a.tex}}}}}\quad\ensuremath{\vcenter{\hbox{\scalebox{0.75}{\input{ClosedDiags/Generators/12a.tex}}}}}
    &
    The bubble cup and cap are the isomorphisms $\hom{1,1}\xrightarrow{\sim} 1$ and $1\xrightarrow{\sim} \hom{1,1}$. They simply encode the unique isomorphism
    \[\{*\}\cong\{f\colon *\to *\}.\]
    \\[-6pt] \midrule \\
    \ensuremath{\vcenter{\hbox{\scalebox{0.75}{\input{ClosedDiags/Generators/1.tex}}}}}
    &
     The down left entry represents a map $A\times \hom{B,C}\rightarrow \hom{B,A\times C}$. This takes a pair $(a,f\colon B\to C)$ and maps it to a function \[f_a(b)=(a,f(b)).\]
    \\[-6pt] \midrule \\
    \ensuremath{\vcenter{\hbox{\scalebox{0.75}{\input{ClosedDiags/Generators/1a.tex}}}}}
    &
    The down left exit represents a map $\hom{1,A\times B}\to A\times \hom{1,B}$, where each $f$ with $f(*)=(x,y)$ gets mapped to $(x,g)$ with \[g(*)=y.\]
    \\[-6pt] \midrule \\
    \ensuremath{\vcenter{\hbox{\scalebox{0.75}{\input{ClosedDiags/GeneratorsBaseline/5.tex}}}}}
    &
    The up left entry allows us to represent \emph{internal} currying. This diagram represents a map
    \[
        \hom{A\times B, \hom{C,D}} \to \hom{A,\hom{B\times C,D}}.
    \]
    This takes a function $f(a,b)=g$ and maps it to the function \[h(a)(b,c)=f(a,b)(c).\]
    \\[-6pt] \midrule \\
    \ensuremath{\vcenter{\hbox{\scalebox{0.75}{\input{ClosedDiags/Generators/7.tex}}}}}
    &
    The up left exit allows us to represent a generalised form of internal uncurrying. This diagram represents a map
    \[\hom{A, \hom{B,C}\times D}\to \hom{A\tensor B, \hom{1,C}\times D}.\]
    This takes a function $f(a)=(f_1(a),d)$ and maps it to the function 
    \[h(a,b)=(*\mapsto f_1(a)(b), d).\]
    \\[-6pt] \midrule \\
    \ensuremath{\vcenter{\hbox{\scalebox{0.75}{\input{ClosedDiags/Generators/9.tex}}}}}
    &
    The rightward cap allows the introduction of internal identities. This diagram represents a map
    \[
        \hom{A,C}\to \hom{A\times B, B\times C}.
    \]
    This takes a function $f(a)$ and maps it to the function \[f(a,b)=(b,f(a)).\]
    \\[-6pt] \midrule \\
    \ensuremath{\vcenter{\hbox{\scalebox{0.75}{\input{ClosedDiags/Generators/3.tex}}}}}
    &
    The left crossing cup allows us to partially evaluate functions internally. This diagram represents a map
    \[
        A\times \hom{A\times B, C}\to \hom{B,C}.
    \]
    This takes a pair $(a,f\colon A\times B\to C)$ and maps it to the function \[f(a,-)\colon B\to C.\]
    \\[-6pt] \midrule \\
    \ensuremath{\vcenter{\hbox{\scalebox{0.75}{\input{ClosedDiags/Generators/13.tex}}}}}
    &
    The rightward cup allows us to partially compose functions. This diagram represents a map
    \[
    \hom{A,B\times C}\times \hom{C\times D, E}\to \hom{A\times D, B\times E}
    \]
    This map takes a pair of functions $((f\colon A\to B\times C),(g\colon C\times D,E))$ and maps them to the function
    \[
    h(a,d)=(f_1(a), g(f_2(a), d)).
    \]
    \\[-6pt]
    \midrule
\end{longtable}

%% file: ClosedDiags/Generators/11a.tex
\includegraphics{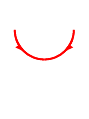}

%% file: ClosedDiags/Generators/12a.tex
\includegraphics{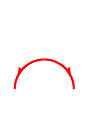}

%% file: ClosedDiags/Generators/1.tex
\includegraphics{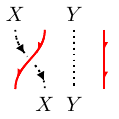}

%% file: ClosedDiags/Generators/1a.tex
\includegraphics{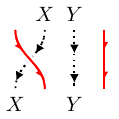}

%% file: ClosedDiags/GeneratorsBaseline/5.tex
\includegraphics{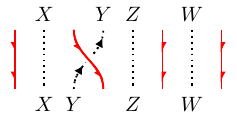}

%% file: ClosedDiags/Generators/7.tex
\includegraphics{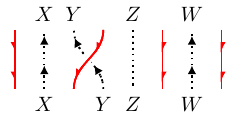}

%% file: ClosedDiags/Generators/9.tex
\includegraphics{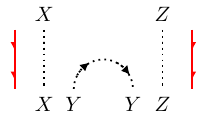}

%% file: ClosedDiags/Generators/3.tex
\includegraphics{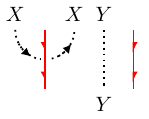}

%% file: ClosedDiags/Generators/13.tex
\includegraphics{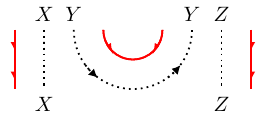}

%% file: ClosedDiags/StringGen/5.tex
\includegraphics{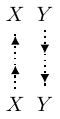}

%% file: ClosedDiags/Normaliser/1a.tex
\includegraphics{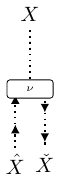}

%% file: ClosedDiags/Normaliser/2d.tex
\includegraphics{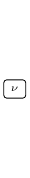}

%% file: ClosedDiags/Normaliser/5.tex
\includegraphics{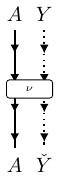}

%% file: ClosedDiags/Normaliser/5a.tex
\includegraphics{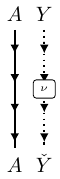}

%% file: ClosedDiags/Normaliser/6.tex
\includegraphics{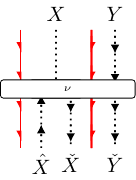}

%% file: ClosedDiags/Normaliser/6a.tex
\includegraphics{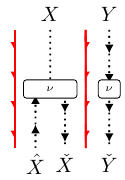}

%% file: ClosedDiags/Normaliser/7.tex
\includegraphics{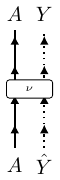}

%% file: ClosedDiags/Normaliser/7a.tex
\includegraphics{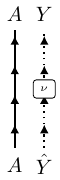}

%% file: ClosedDiags/Normaliser/8.tex
\includegraphics{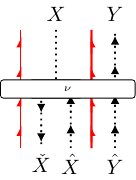}

%% file: ClosedDiags/Normaliser/8a.tex
\includegraphics{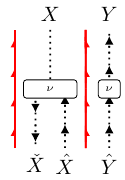}

%% file: ClosedDiags/Normaliser/4b.tex
\includegraphics{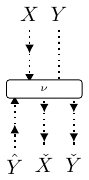}

%% file: ClosedDiags/Normaliser/4c.tex
\includegraphics{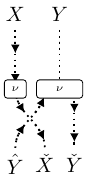}

%% file: ClosedDiags/Normaliser/3b.tex
\includegraphics{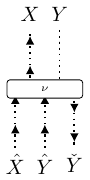}

%% file: ClosedDiags/Normaliser/3c.tex
\includegraphics{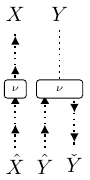}

%% file: ClosedDiags/Sugar/cupCrossing5a.tex
\includegraphics{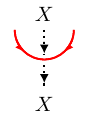}

%% file: ClosedDiags/Sugar/cupCrossing6a.tex
\includegraphics{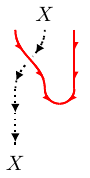}

%% file: ClosedDiags/Sugar/capCrossing5a.tex
\includegraphics{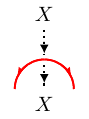}

%% file: ClosedDiags/Sugar/capCrossing6a.tex
\includegraphics{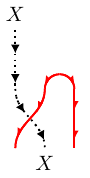}

%% file: ClosedDiags/Sugar/cupCrossing1a.tex
\includegraphics{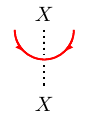}

%% file: ClosedDiags/Sugar/cupCrossing2a.tex
\includegraphics{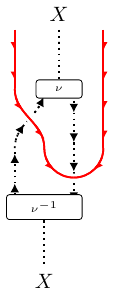}

%% file: ClosedDiags/Sugar/capCrossing1a.tex
\includegraphics{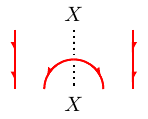}

%% file: ClosedDiags/Sugar/capCrossing2a.tex
\includegraphics{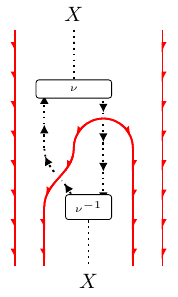}

%% file: figures/derEquations1.tex
\begin{figure}[H]
  \centering
  \subfig{1}{
  \subfig{0.5}{ \ensuremath{\vcenter{\hbox{\scalebox{0.75}{\input{ClosedDiags/Relations/Natural/1}}}}} $\cong$ \ensuremath{\vcenter{\hbox{\scalebox{0.75}{\input{ClosedDiags/Relations/Natural/1a}}}}}}
  \subfig{0.5}{ \ensuremath{\vcenter{\hbox{\scalebox{0.75}{\input{ClosedDiags/Relations/Natural/2}}}}} $\cong$ \ensuremath{\vcenter{\hbox{\scalebox{0.75}{\input{ClosedDiags/Relations/Natural/2a}}}}}
  }\figSpace
  \subfig{0.5}{ \ensuremath{\vcenter{\hbox{\scalebox{0.75}{\input{ClosedDiags/Relations/Natural/4}}}}} $\cong$ \ensuremath{\vcenter{\hbox{\scalebox{0.75}{\input{ClosedDiags/Relations/Natural/4a}}}}}}
  \subfig{0.5}{ \ensuremath{\vcenter{\hbox{\scalebox{0.75}{\input{ClosedDiags/Relations/Natural/5}}}}} $\cong$ \ensuremath{\vcenter{\hbox{\scalebox{0.75}{\input{ClosedDiags/Relations/Natural/5a}}}}}}
\caption{Derived naturality equations.}
\label{fig:derNaturalities}}
\figSpace
\subfig{1}{
  \subfig{0.5}{\ensuremath{\vcenter{\hbox{\scalebox{0.75}{\input{ClosedDiags/CapBubbleCross/1.tex}}}}} $\cong$     \ensuremath{\vcenter{\hbox{\scalebox{0.75}{\input{ClosedDiags/CapBubbleCross/6.tex}}}}}}
  \subfig{0.5}{\ensuremath{\vcenter{\hbox{\scalebox{0.75}{\input{ClosedDiags/CapSlide/5.tex}}}}}=\ensuremath{\vcenter{\hbox{\scalebox{0.75}{\input{ClosedDiags/CapSlide/1.tex}}}}}}
\caption{Derived cap equations.}
\label{fig:derCapRel}}
\figSpace
\subfig{1}{
      \subfig{0.5}{\ensuremath{\vcenter{\hbox{\scalebox{0.75}{\input{ClosedDiags/DerivedRel/SymDownEntry1.tex}}}}}  $\cong$  \ensuremath{\vcenter{\hbox{\scalebox{0.75}{\input{ClosedDiags/DerivedRel/SymDownEntry6.tex}}}}}}
      \subfig{0.5}{\ensuremath{\vcenter{\hbox{\scalebox{0.75}{\input{ClosedDiags/DerivedRel/SymDownExit1.tex}}}}} $\cong$ \ensuremath{\vcenter{\hbox{\scalebox{0.75}{\input{ClosedDiags/DerivedRel/SymDownExit4.tex}}}}}}\figSpace
\caption{Derived crossing equations.}
\label{fig:derCrossingRel}}
\caption{Derived equations for string diagrams with brackets.}
\label{fig:derEquations}
\end{figure}

%% file: figures/derEquations2.tex
\begin{figure}[p]
  \centering
  \subfig{1}{
  \subfig{1}{ 
  \ensuremath{\vcenter{\hbox{\scalebox{0.75}{\input{ClosedDiags/Relations/cupNat1}}}}} $\cong$   \ensuremath{\vcenter{\hbox{\scalebox{0.75}{\input{ClosedDiags/Relations/cupNat2}}}}}}\\
  for all
\[\ensuremath{\vcenter{\hbox{\scalebox{0.75}{\input{ClosedDiags/Relations/cupNat3}}}}} \in \left\{ \!\!\ensuremath{\vcenter{\hbox{\scalebox{0.75}{\input{ClosedDiags/Relations/cupNat9}}}}},\!\! \ensuremath{\vcenter{\hbox{\scalebox{0.75}{\input{ClosedDiags/Relations/cupNat5}}}}},\!\! \ensuremath{\vcenter{\hbox{\scalebox{0.75}{\input{ClosedDiags/Relations/cupNat6}}}}},\!\! \ensuremath{\vcenter{\hbox{\scalebox{0.75}{\input{ClosedDiags/Relations/cupNat7}}}}},\!\! \ensuremath{\vcenter{\hbox{\scalebox{0.75}{\input{ClosedDiags/Relations/cupNat8}}}}} \!\!\right\},
\ensuremath{\vcenter{\hbox{\scalebox{0.75}{\input{ClosedDiags/Relations/cupNat3a}}}}} \in \left\{ \!\!\ensuremath{\vcenter{\hbox{\scalebox{0.75}{\input{ClosedDiags/Relations/cupNat9a}}}}},\!\! \ensuremath{\vcenter{\hbox{\scalebox{0.75}{\input{ClosedDiags/Relations/cupNat5a}}}}},\!\! \ensuremath{\vcenter{\hbox{\scalebox{0.75}{\input{ClosedDiags/Relations/cupNat6a}}}}},\!\! \ensuremath{\vcenter{\hbox{\scalebox{0.75}{\input{ClosedDiags/Relations/cupNat7a}}}}},\!\! \ensuremath{\vcenter{\hbox{\scalebox{0.75}{\input{ClosedDiags/Relations/cupNat8a}}}}} \!\!\right\}\]
\caption{An interchange law for cups and crossings}
\label{fig:derCupInterchange}}\figSpace
\subfig{1}{
\subfig{1}{{\ensuremath{\vcenter{\hbox{\scalebox{0.75}{\input{ClosedDiags/Smooth/1.tex}}}}} $\cong$ \ensuremath{\vcenter{\hbox{\scalebox{0.75}{\input{ClosedDiags/Smooth/6.tex}}}}}}}
\caption{Derived smoothing equation.}
\label{fig:derSmoothing}}\figSpace
\subfig{1}{\ensuremath{\vcenter{\hbox{\scalebox{0.75}{\input{ClosedDiags/Relations/Natural/7}}}}}  $\cong$ \ensuremath{\vcenter{\hbox{\scalebox{0.75}{\input{ClosedDiags/Relations/Natural/7a.tex}}}}}
\caption{Naturality laws for the rightward cup}
\label{fig:derCupNatural}}\figSpace
\subfig{1}{\ensuremath{\vcenter{\hbox{\scalebox{0.75}{\input{ClosedDiags/Relations/SymCup1.tex}}}}} $\cong$   \ensuremath{\vcenter{\hbox{\scalebox{0.75}{\input{ClosedDiags/Relations/SymCup2.tex}}}}}
\caption{A symmetry law for the rightward cup}
\label{fig:derSymCup}}
\caption{Derived equations for the rightward cup.}
\label{fig:derRightwardCup}
\end{figure}

%% file: figures/derNatProof1.tex
  \begin{figure}[H]
    \begin{subfigure}{0.33\textwidth}
      \centering
      \ensuremath{\vcenter{\hbox{\scalebox{0.75}{\input{ClosedDiags/DerivedRel/Natural/2.tex}}}}}\\
      \textbf{(a)}
    \end{subfigure}%
    \begin{subfigure}{0.33\textwidth}
      \centering
      \ensuremath{\vcenter{\hbox{\scalebox{0.75}{\input{ClosedDiags/DerivedRel/Natural/3.tex}}}}}\\
      \textbf{(b)}
    \end{subfigure}%
        \begin{subfigure}{0.33\textwidth}
      \centering
      \ensuremath{\vcenter{\hbox{\scalebox{0.75}{\input{ClosedDiags/DerivedRel/Natural/4.tex}}}}}\\
      \textbf{(c)}
    \end{subfigure}%
    \addtocounter{figure}{-1}

  \end{figure}

%% file: figures/derNatProof2.tex
  \begin{figure}[H]
    \begin{subfigure}{0.5\textwidth}
      \centering
      \ensuremath{\vcenter{\hbox{\scalebox{0.75}{\input{ClosedDiags/DerivedRel/Natural/5.tex}}}}}\\
      \textbf{(a)}
    \end{subfigure}%
    \begin{subfigure}{0.5\textwidth}
      \centering
      \ensuremath{\vcenter{\hbox{\scalebox{0.75}{\input{ClosedDiags/DerivedRel/Natural/6.tex}}}}}\\
      \textbf{(b)}
    \end{subfigure}%
    \addtocounter{figure}{-1}

  \end{figure}

%% file: figures/derCapBubbleCross
     \begin{figure}[H]
    \begin{subfigure}{.15\textwidth}
  \centering
    \ensuremath{\vcenter{\hbox{\scalebox{0.75}{\input{ClosedDiags/CapBubbleCross/1.tex}}}}}\\
    \textbf{(a)}
\end{subfigure}%
    \begin{subfigure}{.2\textwidth}
  \centering
    \ensuremath{\vcenter{\hbox{\scalebox{0.75}{\input{ClosedDiags/CapBubbleCross/2.tex}}}}}\\
    \textbf{(b)}
\end{subfigure}%
    \begin{subfigure}{.2\textwidth}
        \centering
    \ensuremath{\vcenter{\hbox{\scalebox{0.75}{\input{ClosedDiags/CapBubbleCross/3.tex}}}}}\\
    \textbf{(c)}
\end{subfigure}
    \begin{subfigure}{.25\textwidth}
        \centering
    \ensuremath{\vcenter{\hbox{\scalebox{0.75}{\input{ClosedDiags/CapBubbleCross/4.tex}}}}}\\
    \textbf{(d)}
\end{subfigure}%
    \begin{subfigure}{.2\textwidth}
        \centering
    \ensuremath{\vcenter{\hbox{\scalebox{0.75}{\input{ClosedDiags/CapBubbleCross/6.tex}}}}}\\
    \textbf{(e)}
\end{subfigure}%
\addtocounter{figure}{-1}

\end{figure}

%% file: figures/derCapSlide.tex
  \begin{figure}[H]
    \begin{subfigure}{.2\textwidth}
  \centering
    \ensuremath{\vcenter{\hbox{\scalebox{0.75}{\input{ClosedDiags/CapSlide/1.tex}}}}}\\
    \textbf{(a)}
\end{subfigure}%
\begin{subfigure}{.2\textwidth}
  \centering
\ensuremath{\vcenter{\hbox{\scalebox{0.75}{\input{ClosedDiags/CapSlide/2.tex}}}}}\\
\textbf{(b)}
\end{subfigure}%
\begin{subfigure}{.2\textwidth}
  \centering
\ensuremath{\vcenter{\hbox{\scalebox{0.75}{\input{ClosedDiags/CapSlide/3.tex}}}}}\\
\textbf{(c)}
\end{subfigure}%
\begin{subfigure}{.2\textwidth}
  \centering
\ensuremath{\vcenter{\hbox{\scalebox{0.75}{\input{ClosedDiags/CapSlide/4.tex}}}}}\\
\textbf{(d)}
\end{subfigure}%
\begin{subfigure}{.2\textwidth}
  \centering
\ \ensuremath{\vcenter{\hbox{\scalebox{0.75}{\input{ClosedDiags/CapSlide/5.tex}}}}}\\
\textbf{(e)}
\end{subfigure}%
\addtocounter{figure}{-1}

  \end{figure}

%% file: figures/derCrossingProofs1.tex
    \begin{figure}[H]
  \subfig{1}{
    \begin{subfigure}{.33\textwidth}
  \centering
    \ensuremath{\vcenter{\hbox{\scalebox{0.75}{\input{ClosedDiags/DerivedRel/SymDownEntry1.tex}}}}}\\
    \textbf{(a)}
\end{subfigure}%
\begin{subfigure}{.33\textwidth}
  \centering
\ensuremath{\vcenter{\hbox{\scalebox{0.75}{\input{ClosedDiags/DerivedRel/SymDownEntry2.tex}}}}}\\
\textbf{(b)}
\end{subfigure}%
\begin{subfigure}{.33\textwidth}
  \centering
\ensuremath{\vcenter{\hbox{\scalebox{0.75}{\input{ClosedDiags/DerivedRel/SymDownEntry3.tex}}}}}\\
\textbf{(c)}
\end{subfigure}\vspace{0.3cm}
\begin{subfigure}{.33\textwidth}
  \centering
\ensuremath{\vcenter{\hbox{\scalebox{0.75}{\input{ClosedDiags/DerivedRel/SymDownEntry4.tex}}}}}\\
\textbf{(d)}
\end{subfigure}%
\begin{subfigure}{.33\textwidth}
  \centering
\ensuremath{\vcenter{\hbox{\scalebox{0.75}{\input{ClosedDiags/DerivedRel/SymDownEntry5.tex}}}}}\\
\textbf{(e)}
\end{subfigure}%
\begin{subfigure}{.33\textwidth}
  \centering
\ensuremath{\vcenter{\hbox{\scalebox{0.75}{\input{ClosedDiags/DerivedRel/SymDownEntry6.tex}}}}}\figLab{f}

\addtocounter{figure}{-1}
\end{subfigure}}
\end{figure}

%% file: figures/derCrossingProofs2.tex
\begin{figure}[H]
\subfig{1}{
\subfig{0.25}{
    \ensuremath{\vcenter{\hbox{\scalebox{0.75}{\input{ClosedDiags/DerivedRel/SymDownExit1.tex}}}}}\figLab{a}}
\subfig{0.25}{\ensuremath{\vcenter{\hbox{\scalebox{0.75}{\input{ClosedDiags/DerivedRel/SymDownExit2.tex}}}}}\figLab{b}}
\subfig{0.25}{
\ensuremath{\vcenter{\hbox{\scalebox{0.75}{\input{ClosedDiags/DerivedRel/SymDownExit3.tex}}}}}\figLab{c}}
\subfig{0.25}{
\ensuremath{\vcenter{\hbox{\scalebox{0.75}{\input{ClosedDiags/DerivedRel/SymDownExit4.tex}}}}}\figLab{d}}}
\addtocounter{figure}{-1}

\end{figure}

%% file: figures/derSmoothingProof.tex
  \begin{figure}[H]
  \subfig{1}{
    \begin{subfigure}{.5\textwidth}
  \centering
    \ensuremath{\vcenter{\hbox{\scalebox{0.75}{\input{ClosedDiags/BigCupSplit/1}}}}}\\
    \textbf{(a)}
\end{subfigure}%
\begin{subfigure}{.5\textwidth}
  \centering
 \ensuremath{\vcenter{\hbox{\scalebox{0.75}{\input{ClosedDiags/BigCupSplit/2}}}}}\\
\textbf{(b)}
\end{subfigure}\\
\begin{subfigure}{.5\textwidth}
  \centering
\ \ensuremath{\vcenter{\hbox{\scalebox{0.75}{\input{ClosedDiags/BigCupSplit/3}}}}}\\
\textbf{(c)}
\end{subfigure}%
\begin{subfigure}{.5\textwidth}
  \centering
\ensuremath{\vcenter{\hbox{\scalebox{0.75}{\input{ClosedDiags/BigCupSplit/4.tex}}}}}\\
\textbf{(d)}
\end{subfigure}}%
 \caption{The proof of \autoref{lemma:bigCupSplit}}
 \label{fig:bigCupSplit}
\figSpace
\subfig{1}{
    \begin{subfigure}{0.33\textwidth}
    \centering
    \ensuremath{\vcenter{\hbox{\scalebox{0.75}{\input{ClosedDiags/Smooth/1.tex}}}}}\\
    \textbf{(a)}
    \end{subfigure}%
      \begin{subfigure}{0.33\textwidth}
    \centering
    \ensuremath{\vcenter{\hbox{\scalebox{0.75}{\input{ClosedDiags/Smooth/2.tex}}}}}\\
    \textbf{(b)}
    \end{subfigure}%
      \begin{subfigure}{0.33\textwidth}
    \centering
    \ensuremath{\vcenter{\hbox{\scalebox{0.75}{\input{ClosedDiags/Smooth/3.tex}}}}}\\
    \textbf{(c)}
    \end{subfigure}\\
    \begin{subfigure}{0.33\textwidth}
    \centering
    \ensuremath{\vcenter{\hbox{\scalebox{0.75}{\input{ClosedDiags/Smooth/4.tex}}}}}\\
    \textbf{(d)}
    \end{subfigure}%
    \begin{subfigure}{0.33\textwidth}
    \centering
    \ensuremath{\vcenter{\hbox{\scalebox{0.75}{\input{ClosedDiags/Smooth/5.tex}}}}}\\
    \textbf{(e)}
    \end{subfigure}%
    \begin{subfigure}{0.33\textwidth}
    \centering
    \ensuremath{\vcenter{\hbox{\scalebox{0.75}{\input{ClosedDiags/Smooth/6.tex}}}}}\\
    \textbf{(f)}
    \end{subfigure}
}
\caption{The proof of the smoothing equation.}
 \label{fig:derSmoothingProof}
\figSpace
\end{figure}

%% file: ClosedDiags/tGeneral.tex
\begin{tikzpicture}[baseline={(current bounding box.center)}, xscale=1]
\drawStartNodes{
    B/$X$/0
}

\ddShift{B}{0,-1}
\coord{t}{B}{0,0}

\ddShift{B}{0,-1}

\bead{t}{t}{t}

\drawEndNodes{
    B/$Z$
}
\end{tikzpicture}

%% file: ClosedDiags/WellDef/coevalXY.tex
\begin{tikzpicture}[baseline={(current bounding box.center)}, xscale=1
]
\drawStartNodes{
    A/$X$/0
}
\ddShift{A}{0,-1}
\coordinate[shift={(-0.5,-1)}] (B1) at (A);
\coordinate[shift={(-1,-1)}] (B2) at (A);
\coordinate[shift={(-1.5,-1)}] (b1) at (A);
\coordinate[shift={(0.5,-1)}] (b2) at (A);

\ddCap{B2}{B1}[0.3]

\ddShift{A}{0,-1}
\bracketCap{b1}{b2}[1]

\drawEndNodes{
    A/{$X$},
    B1/{$Y$},
    B2/{$Y$}
}

\end{tikzpicture}

%% file: ClosedDiags/WellDef/evalXY.tex
\begin{tikzpicture}[baseline={(current bounding box.center)}, xscale=1
]
\drawStartNodes{
    B1/{$Y$}/0.5,
    b1//0.5,
    B2/{$Y$}/0.5,
    A/{$X$}/0.5,
    b2//0
}
\ddCup{B1}{B2}
\drawBracketShift{b1}{0,-1}[right]
\drawBracketShift{b2}{0,-1}

\ddShift{A}{0,-1}
\ddShift{A}{0,-1}

\bracketCup{b1}{b2}

\drawEndNodes{
    A/{$X$}
}

\end{tikzpicture}

%% file: ClosedDiags/CoUnitNat/coevalXY.tex
\includegraphics{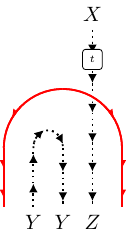}

%% file: ClosedDiags/CoUnitNat/coevalXY2.tex
\includegraphics{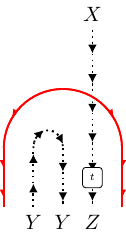}

%% file: ClosedDiags/CoUnitNat/evalXY.tex
\includegraphics{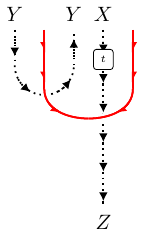}

%% file: ClosedDiags/CoUnitNat/evalXY2.tex
\includegraphics{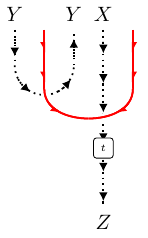}

%% file: sections/examples2.tex
\section{Examples}\label{sec:examplesApp}
Before turning to the soundness and completeness of our string diagrams with brackets, we begin with several motivating examples.

We start with the archetypal case of currying and uncurrying in a closed symmetric monoidal category. This highlights the advantages of our diagrammatic language over other approaches in the literature.

Next, we extend the calculus with copy and delete morphisms to obtain a \emph{cartesian} closed structure, and show how this allows us to reason diagrammatically about the $\lambda$-calculus.

Our third example demonstrates how the self-enrichment of a closed symmetric monoidal category can be expressed graphically, and how the associativity and unitality of the enrichment follow naturally from the diagrams.

Finally, we use our calculus to illustrate a well-known coherence issue involving units in closed monoidal categories.
\begin{example}[Currying]\label{example:currying}
  The archetypal example is that of \emph{currying}. Given two morphisms $f \colon Y \otimes X \to Z$ and $g \colon X \to [Y,Z]$, we define the currying of $f$ and the uncurrying of $g$ as the diagrams on the left and on the right, respectively, below.
\begin{figure}[H]
\begin{subfigure}{0.5\textwidth}
  \centering
  \ensuremath{\vcenter{\hbox{\scalebox{0.75}{\input{ClosedDiags/Currying/1a.tex}}}}} $\mapsto$ \ensuremath{\vcenter{\hbox{\scalebox{0.75}{\input{ClosedDiags/Currying/2a.tex}}}}}%
\end{subfigure}%
\begin{subfigure}{0.5\textwidth}
  \centering
  \ensuremath{\vcenter{\hbox{\scalebox{0.75}{\input{ClosedDiags/Currying/3.tex}}}}} $\mapsto$ \ensuremath{\vcenter{\hbox{\scalebox{0.75}{\input{ClosedDiags/Currying/4.tex}}}}}%
\end{subfigure}
\end{figure}
Note that in both cases, the operations correspond diagrammatically to bending the (un)curried wire. This visual intuition hints at an underlying bijection between the two operations, which we establish formally below.
\[
\ensuremath{\vcenter{\hbox{\scalebox{0.75}{\input{ClosedDiags/Currying/5.tex}}}}}\qquad \cong \qquad \ensuremath{\vcenter{\hbox{\scalebox{0.75}{\input{ClosedDiags/Currying/6.tex}}}}}\qquad \cong\qquad  \ensuremath{\vcenter{\hbox{\scalebox{0.75}{\input{ClosedDiags/Currying/7.tex}}}}} \qquad \cong\qquad \ensuremath{\vcenter{\hbox{\scalebox{0.75}{\input{ClosedDiags/Currying/8.tex}}}}}
\]
\[
\ensuremath{\vcenter{\hbox{\scalebox{0.75}{\input{ClosedDiags/Currying/9.tex}}}}} \quad\cong\quad \ensuremath{\vcenter{\hbox{\scalebox{0.75}{\input{ClosedDiags/Currying/10.tex}}}}} \quad\cong\quad \ensuremath{\vcenter{\hbox{\scalebox{0.75}{\input{ClosedDiags/Currying/11.tex}}}}}\quad\cong\quad  \ensuremath{\vcenter{\hbox{\scalebox{0.75}{\input{ClosedDiags/Currying/12.tex}}}}}\]
 Hence we have a bijection.
\end{example}
\begin{example}[$\lambda$-calculus]%
  Expanding on the previous example, we give a diagrammatic encoding of the simply typed $\lambda$-calculus, that we recall now.
  Fixed a set of variables $\mathcal X$ and a set of basic types $\sort$, we build types as $X \coloneqq A \mid X \to X$, where $A \in \sort$. Terms are generated according to the following grammar, where $x \in \mathcal X$:
\[    u \coloneqq x \mid \lambda x . u \mid u \; u
\]Amongst terms, we consider only those that can be typed according to type judgements of the form $\lTerm{\Gamma}{u}{X}$, where $\Gamma = x_1 \colon X_1, \ldots, x_n \colon X_n$ is a variable type assignment. Type judgements are defined according to the following rules:
\begin{center}
  \[
    \begin{array}{c@{\qquad\qquad}c@{\qquad\qquad}c}
        \lTerm{\Gamma, x \colon X}{x}{X} 
        & 
        \inferrule{
            \lTerm{\Gamma, x \colon X}{u}{Y}
        }{
            \lTerm{\Gamma}{\lambda x. u}{X \to Y}
        }
        &
        \inferrule{
            \lTerm{\Gamma}{u}{X}
            \qquad
            \lTerm{\Gamma}{v}{X \to Y}
        }{
            \lTerm{\Gamma}{v \; u}{Y}
        }
    \end{array}
    \]
\end{center}
We interpret typed terms in our category of string diagrams, additionally supplied with natural copy and discard operations that satisfy the axioms of cocommutative comonoids (see e.g. Table 8 in~\cite{Selinger2009}). By Fox's Theorem~\cite{fox1976coalgebras}, the additional structure makes the category cartesian closed.
The diagrammatic encoding is then given as follows:
\begin{figure}[H]
  \subfig{0.33}{
    \ensuremath{\vcenter{\hbox{\scalebox{0.75}{\input{ClosedDiags/Lambda/absEnc.tex}}}}}\\\vspace{0.3cm}
      $\lTerm{\Gamma}{\lambda x. u}{X \to Y}$
  }%
  \subfig{0.33}{
    \ensuremath{\vcenter{\hbox{\scalebox{0.75}{\input{ClosedDiags/Lambda/app.tex}}}}}\vspace{0.3cm}\\
    $\lTerm{\Gamma}{v \; u}{Y}$
  }
  \subfig{0.33}{
    \ensuremath{\vcenter{\hbox{\scalebox{0.75}{\input{ClosedDiags/Lambda/varEnc.tex}}}}}\vspace{0.3cm}\\
    $\lTerm{\Gamma, x \colon X}{x}{X}$
  }
\end{figure}
We conclude by showing how $\beta$-reduction is performed in our language.
\[
\ensuremath{\vcenter{\hbox{\scalebox{0.75}{\input{ClosedDiags/Lambda/beta2.tex}}}}} \qquad \cong \qquad \ensuremath{\vcenter{\hbox{\scalebox{0.75}{\input{ClosedDiags/Lambda/beta3.tex}}}}} \qquad \cong \qquad \ensuremath{\vcenter{\hbox{\scalebox{0.75}{\input{ClosedDiags/Lambda/beta4.tex}}}}} \qquad \cong \qquad \ensuremath{\vcenter{\hbox{\scalebox{0.75}{\input{ClosedDiags/Lambda/beta5.tex}}}}}
\]

The proof above consists of yanking strings and sliding or popping bubbles, emphasising the geometric nature of string diagrams with brackets. We invite the reader to compare this with the proof of $\beta$-reduction in~\cite{ghica2023string}, which, instead, relies more heavily on algebraic reasoning.

\end{example}

\begin{example}[Self enrichment]
One well-known and useful feature of closed symmetric monoidal categories is that they can be enriched over themselves using the internal hom. In other words, for any objects $X$, $Y$, and $Z$ in the category there are canonical composition and identity morphisms:
\[
\overline{\circ}_{X,Y,Z} \colon \hom{X,Y}\tensor \hom{Y,Z}\to \hom{X,Z} \text{ and } \overline{\id}_{X} \colon I \to \hom{X,X}
\]

In our language these are represented by diagrams the left- and right-hand diagrams below. These, of course, adhere to associativity and unitality axioms. 

\begin{figure}[H]
  \subfig{0.5}{\ensuremath{\vcenter{\hbox{\scalebox{0.75}{\input{ClosedDiags/Example/intComp}}}}}} 
  \subfig{0.5}{\ensuremath{\vcenter{\hbox{\scalebox{0.75}{\input{ClosedDiags/Example/intUnit}}}}}}
\end{figure}

The unitality axiom is shown below, using \autoref{fig:derCupInterchange} as well as the yanking equations for strings and bracketing strings. Associativity also follows from \autoref{fig:derCupInterchange}.

\[
    {\ensuremath{\vcenter{\hbox{\scalebox{0.75}{\input{ClosedDiags/Example/intIdentity}}}}}} 
    \cong{\ensuremath{\vcenter{\hbox{\scalebox{0.75}{\input{ClosedDiags/Example/intIdentity2}}}}}} 
\]
\end{example}

\begin{example}[Triple unit diagram]
  In their paper on coherence in closed symmetric monoidal categories, Kelly and Mac Lane~\cite{KELLY197197} show that there is no full coherence in closed symmetric monoidal categories, in the sense that we can generate non-commuting diagrams from the structural morphisms. 
  
  The typical example is the diagram below, that does not commute, for instance, in the closed monoidal category of pointed sets~\cite{MEHATS2007127, hughes2012simple}: %
  \tikzexternaldisable
  \[\begin{tikzcd}
    \hom{\hom{\hom{X,I},I},I}
    \arrow[rr, equal]
    \arrow[dr, "\hom{v_X,I}"{swap}]
        & &\hom{\hom{\hom{X,I},I},I}\\
    & \hom{X,I}
    \arrow[ru, "v_{\hom{X,I}}"{swap}]
  \end{tikzcd}
  \]
  \tikzexternalenable
  where $v_X\colon X\to \hom{\hom{X,I},I}$ is given by the composite
  \[
    X\xrightarrow{\eta^{X,I}} \hom{\hom{X,I},X\tensor \hom{X,I}} \xrightarrow{\hom{\hom{X,I},\epsilon^X_I}} \hom{\hom{X,I},I}.
  \]
  To see how this phenomenon appears in our diagrammatic language, note that $v_X$ is depicted by the string diagram below on the left. Meanwhile, the composite in the pasting diagram is represented by the string diagram on the right:
  \[
  \ensuremath{\vcenter{\hbox{\scalebox{0.75}{\input{ClosedDiags/Example/tripleUnit1.tex}}}}} \qquad\qquad\qquad \ensuremath{\vcenter{\hbox{\scalebox{0.75}{\input{ClosedDiags/Example/tripleUnit2.tex}}}}}
  \]
  By \autoref{thm:freeness}, proven in the next section, we know that this diagram does not collapse to the identity diagram, but intuitively this follows from the fact that the cups and caps in the string diagram above cannot be merged.
\end{example}

%% file: ClosedDiags/Currying/5.tex
\includegraphics{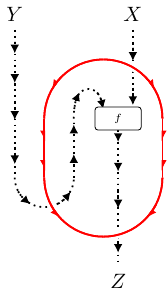}

%% file: ClosedDiags/Currying/6.tex
\includegraphics{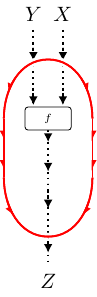}

%% file: ClosedDiags/Currying/7.tex
\includegraphics{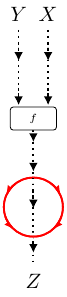}

%% file: ClosedDiags/Currying/8.tex
\includegraphics{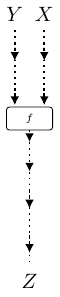}

%% file: ClosedDiags/Currying/9.tex
\includegraphics{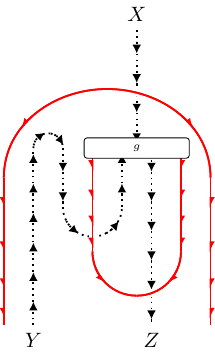}

%% file: ClosedDiags/Currying/10.tex
\includegraphics{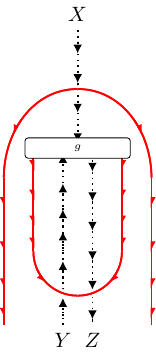}

%% file: ClosedDiags/Currying/11.tex
\includegraphics{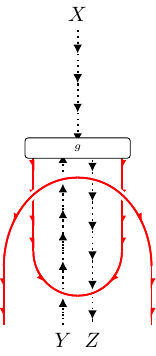}

%% file: ClosedDiags/Currying/12.tex
\includegraphics{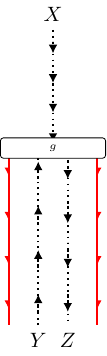}

%% file: ClosedDiags/Lambda/beta2.tex
\includegraphics{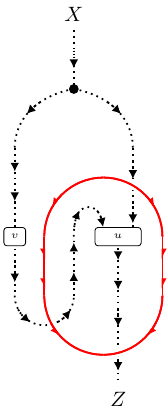}

%% file: ClosedDiags/Lambda/beta3.tex
\includegraphics{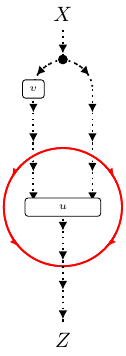}

%% file: ClosedDiags/Lambda/beta4.tex
\includegraphics{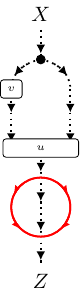}

%% file: ClosedDiags/Lambda/beta5.tex
\includegraphics{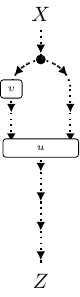}

%% file: ClosedDiags/Example/intIdentity.tex
\includegraphics{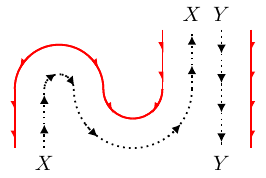}

%% file: ClosedDiags/Example/intIdentity2.tex
\includegraphics{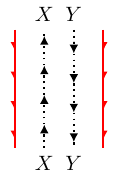}

%% file: ClosedDiags/Example/tripleUnit1.tex
\includegraphics{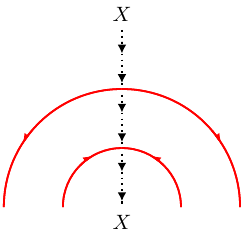}

%% file: ClosedDiags/Example/tripleUnit2.tex
\includegraphics{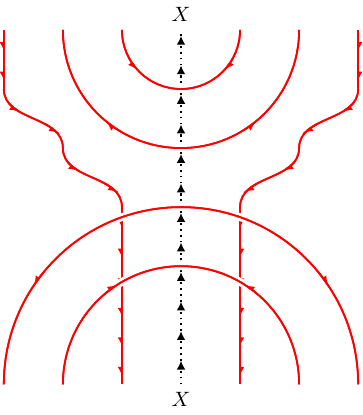}

%% file: sections/completeness.tex
\section{Soundness and Completeness}\label{sec:soundness and completeness}

In this section we prove the completeness result, namely that the category of string diagrams $\freeCat$ is \emph{the} free closed symmetric monoidal category generated by a closed monoidal signature. We begin by proving that the free closed symmetric monoidal category is equivalent to one where diagrams are normalised in a suitable sense. We then prove that each of our generators can be decomposed into evaluation and coevaluation maps, before using this result to inductively define an interpretation functor.

\input{sections/normalisation.tex}
\subsection{Decomposition}
In this section we give yet another equivalent characterisation of the category of string diagrams with brackets. Specifically, we show that normalised diagrams can be expressed as composition of the co/evaluation diagrams shown in the proof of \autoref{thm:string diagrams are closed monoidal}.

\begin{definition}
    The closed symmetric monoidal category $\minimal{\freeCat}$ is defined as the subcategory of $\freeNormalCat$ with the same objects but whose morphisms are generated by compositions of the co/evaluation diagrams.
\end{definition}
    \begin{lemma}
        \label{lemma:capInsidetoOutside}
The following is an equality of diagrams.  
\[
    \ensuremath{\vcenter{\hbox{\scalebox{0.75}{\input{ClosedDiags/Relations/Prop1/1.tex}}}}} \cong  \ensuremath{\vcenter{\hbox{\scalebox{0.75}{\input{ClosedDiags/Relations/Prop1/2.tex}}}}}
  \]
\end{lemma}
\begin{proof}
  The left hand side is equal the diagram below by merge and naturality, the right hand side is equal to the diagram below by merge.
  \[
   \ensuremath{\vcenter{\hbox{\scalebox{0.75}{\input{ClosedDiags/Relations/Prop1/3.tex}}}}}
  \]
\end{proof}
\begin{lemma}\label{lemma:decomposition}
    There is an equality $\freeNormalCat = \minimal{\freeCat}$.
\end{lemma}
\begin{proof} 
    We prove the equality by showing that each generator of $\freeNormalCat$ decomposes into composites, tensors and bracketings of evaluation and coevaluation terms. Firstly, note that the bracket cups and caps correspond to $\eta$ and $\epsilon$ when $X,Y = I$.
    
    The decomposition of down right exit is given by the sequence of diagrams below where $(a)\cong(b)$ by the merge equation, $(b)\cong(c)$ by naturality and $(c)\cong(d)$ by syntactic sugar. Hence down right exit decomposes as the following:
    \[
        \hom{I,X\tensor Y} \xrightarrow{\epsilon^I_{X\tensor Y}} X\tensor Y \xrightarrow{\eta^I_{X} \tensor Y} \hom{I,X}\tensor Y.
    \]
    \input{figures/decompDownRightExit.tex}

    The decomposition of down right entry is given by the sequence of diagrams below where $(a)\cong(b)$ by the merge equation, $(b)\cong(c)$ by the yanking equation, and $(c)\cong(d)$ by syntactic sugar. Hence down right entry decomposes as the following
    \[
        \hom{X,Y}\tensor Z \xrightarrow{\eta^X} \hom{X,X\tensor \hom{X,Y}\tensor Z} \xrightarrow{\hom{X, \epsilon^X\tensor Z}} \hom{X,Y\tensor Z}.
    \]
    \input{figures/decompDownRightEntry.tex}

The decomposition of the left crossing cup is given by the sequence of diagrams below, where $(a)\cong(b)$ by the merge equation, $(b)\cong(c)$ by the naturality in \autoref{fig:derNaturalities}, and $(c)\cong(d)$ by the yanking equations. Hence the left cup crossing decomposes as the following
  \[
  X\tensor \hom{X\tensor Y, Z} \xrightarrow{\eta^Y} \hom{Y,Y\tensor X\tensor \hom{X\tensor Y, Z}} \xrightarrow{\hom{Y,\epsilon^{X\tensor Y}}} \hom{Y,Z}.
  \]
  \input{figures/decompCupCrossing.tex}

  The decomposition of the up left exit is given by the sequence of diagrams below, where $(a)\cong(b)$ be the merge equation, and $(b)\cong(c)$ by the yanking equation. Note that this is not a full decomposition: the second left crossing cup is not part of an evaluation map. However, since we have a decomposition for the left crossing cup above, diagram $(c)$ can be further decomposed using that result. Hence the up left exit decomposes as the following:
  \begin{align*}
      \hom{X, \hom{Y\tensor Z,W}\tensor U} &\xrightarrow{\eta^{Y\tensor X}} \hom{X\tensor Y, Y\tensor X \tensor \hom{X, \hom{Y\tensor Z, W}\tensor U}}\\
        &\xrightarrow{\hom{X\tensor Y, Y\tensor\epsilon^X}} \hom{X\tensor Y, Y\tensor \hom{Y\tensor Z, W}\tensor U}\\
        &\xrightarrow{\hom{Y\tensor X,\left(\normalised\Xcup\right)\tensor U}} \hom{X\tensor Y, \hom{Z,W}\tensor U}.
  \end{align*}
  \input{figures/decompUpLeftExit.tex}

    The decomposition of the rightward cap is given by the two diagrams below, where $(a)\cong(b)$ by the equality in \autoref{fig:derCapRel}. Note again that this is not a full decomposition, here we use the decomposition of up left exit. Hence the rightward cap decomposes as the following: 
    \begin{align*}
      \hom{X,Z} \xrightarrow{\hom{X, \eta_I^Y}\tensor Z} \hom{X,\hom{Y,Y}\tensor Z} &\xrightarrow{(\normalised\uXout)} \hom{X\tensor Y, \hom{I,Y}\tensor Z}\\
        &\xrightarrow{\hom{X\tensor Y, \epsilon^I \tensor Z}} \hom{X\tensor Y, Y\tensor Z}.
    \end{align*}
    \input{figures/decompRightCap.tex}

  The decomposition of the up left entry is given by the sequence of diagrams below, where $(a)\cong(b)$ by the merge equation, $(b)\cong(c)$ by \autoref{lemma:capInsidetoOutside} and $(c)\cong(d)$ by the crossing equations. Once more, this is not a full decomposition, here we use the decomposition of up left exit. Hence the up left entry decomposes as the following:  
  \begin{align*}
    \hom{X\tensor Y, \hom{Z,W}} &\xrightarrow{\eta^X} \hom{X, X\tensor \hom{X\tensor Y, \hom{Z,W}}}\\
    & \xrightarrow{\hom{X,\normalised\Xcup}} \hom{X,\hom{Y,\hom{Z,W}}}\\
    &\xrightarrow{\hom{X, \normalised\uXout }} \hom{X,\hom{Y\tensor Z, \hom{I,W}}}\\
    &\xrightarrow{\hom{X,\hom{Y\tensor Z, \epsilon^I_W}}}  \hom{X,\hom{Y\tensor Z, W}}.
  \end{align*}
  \input{figures/decompUpLeftEntry.tex}

  The decomposition of the rightward cup is given by the equality of the diagrams below, where $(a)\cong(b)$ by \autoref{lemma:bigCupSplit}. Hence the rightward cup decomposes as the following:
  \begin{align*}
    \hom{X,Y\tensor Z} \tensor \hom{Z\tensor W,U} &\xrightarrow{\eta^I}
        \hom{I,\hom{X, Y\tensor Z}\tensor \hom{Z\tensor W, U}}\\
        &\xrightarrow{\normalised\uXout }\hom{X, \hom{I,Y\tensor Z}\tensor \hom{Z\tensor W,U}}\\
        &\xrightarrow{\hom{X, \epsilon^I\tensor \hom{Z\tensor W,U}}}\hom{X, Y\tensor Z\tensor \hom{Z\tensor W, U}}\\
        &\xrightarrow{\hom{X,Y\tensor \normalised\Xcup}}\hom{X,Y\tensor \hom{W,U}}\\
        &\xrightarrow{\normalised\uXout}\hom{X\tensor W,Y\tensor \hom{I,U}}\\
        &\xrightarrow{\hom{X\tensor W, \epsilon^I\tensor Y}} \hom{X\tensor W, U\tensor Y}.
  \end{align*}
  \input{figures/decompRightCup.tex}
\end{proof}

In order to state the completeness theorem it is necessary to introduce the notion of \emph{interpretation} of a closed monoidal signature into a closed symmetric monoidal category.

\begin{definition}\label{def:interpretation}
    An interpretation $\interpretation{I} = (\alpha_\sort, \alpha_\sign)$ of a closed monoidal signature $(\sort, \sign, \arity, \coarity)$ in a closed symmetric monoidal category $\cate{C}$, consists of a pair of functions $\alpha_\sort \colon \sort \to \objects{\C}$ and $\alpha_\sign \colon \sign \to \arrows{\C}$ such that, for all $f \in \sign$, $\alpha_\sign(f) \colon \alpha^\sharp_\sort(\arity(f)) \to \alpha^\sharp_\sort(\coarity(f))$ is a morphism of $\C$, where $\alpha^\sharp_\sort \colon \sort^\sharp \to \objects{\C}$ is the inductive extension of $\alpha_\sort$.

    A closed symmetric monoidal category $\C$ is \emph{freely generated} by $\sign$ if, for all closed symmetric monoidal categories $\cate{D}$ and interpretations $\interpretation{I}$ of $\sign$ in $\cate{D}$, there is a \emph{unique} closed symmetric monoidal functor $\semFunct{-}_\interpretation{I} \colon \C \to \cate{D}$ extending $\interpretation{I}$, i.e. $\semFunct{f}_\interpretation{I} = \alpha_\sign(f)$ for every $f \in \sign$.
\end{definition}

\begin{theorem}\label{thm:freeness}
    $\C_\sign$ is -- up to equivalence -- the free closed symmetric monoidal category generated by the closed monoidal signature $(\sort, \sign)$.
\end{theorem}
\begin{proof}
    By normalisation (\autoref{thm:normalisation equivalence}) and decomposition (\autoref{lemma:decomposition}), we can safely work with the category $\minimal{\freeCat}$ and define a closed symmetric monoidal functor $\semFunct{-}_\interpretation{I} \colon \minimal{\freeCat} \to \cate{D}$, for any interpretation $\interpretation{I} = (\alpha_\sort, \alpha_\sign)$ and any closed symmetric monoidal category $\cate{D}$.

    Recall that since objects of $\freeNormalCat$, and thus of $\minimal{\freeCat}$, have a certain shape. In particular, now the formal dual of an object may only appear on the left part of a bracketed string. Thus, we define the interpretation functor on objects inductively as follows:
    \begin{center}
    \begin{tabular}{r c l r c l}
        $\semFunct{I}_\interpretation{I}$ & $\coloneqq$  & $1$ \qquad&\qquad $\semFunct{A}_\interpretation{I}$ & $\coloneqq$ & $\alpha_\sort(A)$\\
        $\semFunct{XY}_\interpretation{I}$ & $\coloneqq$ & $\semFunct{X}_\interpretation{I} \otimes \semFunct{Y}_\interpretation{I}$
    &
    $\semFunct{[X^\ast Y]}_\interpretation{I}$ & $\coloneqq$ & $[\semFunct{X}_\interpretation{I}, \semFunct{Y}_\interpretation{I}]$ 
    \end{tabular}
\end{center}

    To define $\semFunct{-}_\interpretation{I}$ on morphisms, it suffices to say where to map the generators in $\sign$ and the co/evaluation diagrams. As expected, we map the generators according to the interpretation $\interpretation{I}$ and the co/evaluation diagrams to the co/unit pair of $\cate{D}$.
    \begin{center}
        \scalebox{0.8}{$
        \begin{array}{c@{\qquad\qquad}c@{\qquad\qquad}c}
            \semFunct{\ensuremath{\vcenter{\hbox{\scalebox{0.75}{\input{ClosedDiags/CellGen/f.tex}}}}}}_\interpretation{I} \coloneqq \alpha_\sign(f)
            &
            \semFunct{\ensuremath{\vcenter{\hbox{\scalebox{0.75}{\input{ClosedDiags/WellDef/coevalXY.tex}}}}}}_\interpretation{I} \coloneqq \eta_X
            &
            \semFunct{\ensuremath{\vcenter{\hbox{\scalebox{0.75}{\input{ClosedDiags/WellDef/evalXY.tex}}}}}}_\interpretation{I} \coloneqq \epsilon_X
        \end{array}
    $}
    \end{center}
    To establish soundness of the interpretation we need to ensure that $\semFunct{-}_\interpretation{I}$ is well-defined. Specifically we need to check that, the (normalised version of the) equations in \autoref{fig:crossings,fig:naturality,fig:yanking,fig:bubble} hold after applying the interpretation assignment.

    Take, for example, the bubble pop equation we have that 
    \[
    \semFunct{\ensuremath{\vcenter{\hbox{\scalebox{0.75}{\input{ClosedDiags/Relations/pop1a.tex}}}}}}_\interpretation{I} = \epsilon_1 \circ \eta_1 = \id_1 =
    \semFunct{\ensuremath{\vcenter{\hbox{\scalebox{0.75}{\input{ClosedDiags/StringGen/blank.tex}}}}}}_\interpretation{I}
    \]
    by the triangular identities for $\eta$ and $\epsilon$ in $\cate{D}$. Similarly, for the crossing identity we have that, by definition, the interpretation on the left is equal to the interpretation on the right below.
    \[
        \semFunct{\ensuremath{\vcenter{\hbox{\scalebox{0.75}{\input{ClosedDiags/yankSound.tex}}}}}}_\interpretation{I} 
        =
        \semFunct{\ensuremath{\vcenter{\hbox{\scalebox{0.75}{\input{ClosedDiags/yankSound2a.tex}}}}}}_\interpretation{I} 
    \]
    Note that this interpretation is given by the top right path in the diagram below. 
    \tikzexternaldisable
    \[\begin{tikzcd}[column sep=40pt]
        \hom{1,X\tensor Y}
        \arrow[r, "\epsilon^1_{X\tensor Y}"]
        \arrow[rd, equal]
            & X\tensor Y
            \arrow[r, "\eta^1_X \tensor Y"]
            \arrow[d, "\eta^1_{X\tensor Y}"]
                &\hom{1,X}\tensor Y
                \arrow[d, "\eta^1_{\hom{1,X}\tensor Y}"]\\
        \blank 
            & \hom{1,X\tensor Y}
            \arrow[r, "\hom{1, \eta^1_{X}\tensor Y}"]
            \arrow[rd, equal]
                & \hom{1,\hom{1,X}\tensor Y}
                \arrow[d, "\hom{1,\epsilon^1_X\tensor Y}"]\\
        \blank
            &\blank
                &\hom{1,X\tensor Y}
    \end{tikzcd}
    \]
    \tikzexternalenable    
    This diagram commutes by naturality and the triangle identities. Thus we have that the interpretation above is, in fact, equal to the interpretation below, as expected.
    \[ \semFunct{\ensuremath{\vcenter{\hbox{\scalebox{0.75}{\input{ClosedDiags/yankSound3.tex}}}}}}_\interpretation{I}
    \]
    The other equalities follow similarly. To conclude, observe that $\semFunct{-}_\interpretation{I}$ is a functor of closed symmetric monoidal categories since, by definition, it preserves the tensor, closed structure, symmetries and the co/evaluation morphisms. Moreover, observe that any other closed symmetric monoidal functor $F \colon \minimal{\freeCat} \to \cate{D}$ extending $\interpretation{I}$ must also preserve the co/evaluation morphisms and hence $\semFunct{-}_\interpretation{I}$ is the only such functor.
\end{proof}

%% file: sections/normalisation.tex
\subsection{Normalisation}

As previously mentioned, the closed symmetric monoidal category generated by a signature, defined above, is in some sense, unbiased. This means that we will show it to be free in a 2-categorical sense -- in other words we will show that it is free up to equivalence. In this subsection we define the \emph{normalised} subcategory of the  closed symmetric monoidal category on a signature -- consisting only of normalised objects and terms. We prove that it is closed symmetric monoidal equivalent to the closed symmetric monoidal category on a signature. This ultimately means that, in the proof of completeness and soundness, we can ignore a lot of the purely syntactic generators, such as the normaliser. We begin by proving that, although we include the reflections of the generating diagrams in \autoref{fig:term generators}, this is purely for syntactical convenience.
\begin{definition}
        The \emph{unreflected} closed symmetric monoidal category $\unref{\freeCat}$ is defined as the subcategory of $\freeCat$ with the same objects but whose morphisms are generated:
        \begin{itemize}
        \item on downstrings by \autoref{fig:term generators}, without reflection in the vertical axis;
        \item on upstrings by the generators of \autoref{fig:term generators} after a reflection in the horizontal axis.
        \end{itemize}
\end{definition}
To see that this is closed symmetric monoidal note that it is closed under tensors, internal hom, and contains the evaluation and coevaluation diagrams.
\begin{lemma}\label{lemma:unrefEq}
        $\freeCat$ is equal to $\unref{\freeCat}$.
\end{lemma}
\begin{proof}
        By definition the two categories have the same objects, composition and closed symmetric monoidal structures. To see that they have the same morphisms, note that by the crossing equations, \autoref{fig:derCrossingRel} and \autoref{fig:derSymCup} we know that every generator that is a reflection of those in \autoref{fig:term generators} can be written as a composition of symmetries and its corresponding generator in \autoref{fig:term generators}. 
\end{proof}
\begin{definition}
        The \emph{normalised} closed symmetric monoidal category, $\normalised{\freeCat}$, on a signature $\Sigma$, is defined as the subcategory of $\unref{\freeCat}$ whose objects are exactly the normalised ones, and whose morphisms are those generated by the diagrams of \autoref{fig:normalised term generators}, and, on upstrings, their reflections in the horizontal axis.
\end{definition}
To see that this is closed symmetric monoidal note that it is closed under tensors, internal hom, and contains the evaluation and coevaluation diagrams.
\setcounter{figure}{11}
\begin{figure}[H]
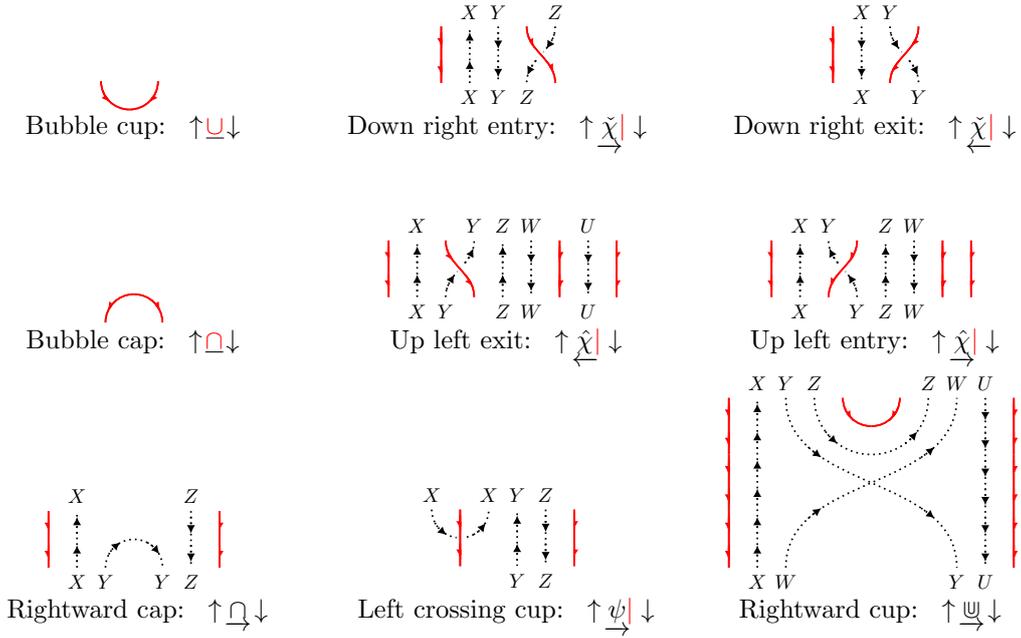

  \centering
\subfig{0.33}{\ensuremath{\vcenter{\hbox{\scalebox{0.75}{\input{ClosedDiags/NormalGen/brackCup.tex}}}}}\\
\diagLabel{Bubble cup: }$\normalised\brCup$}
\subfig{0.33}{\ensuremath{\vcenter{\hbox{\scalebox{0.75}{\input{ClosedDiags/NormalGen/downIn.tex}}}}}\\
\diagLabel{Down right entry:} $\normalised\dXin$
}\figSpace
\subfig{0.33}{
        \ensuremath{\vcenter{\hbox{\scalebox{0.75}{\input{ClosedDiags/NormalGen/downOut.tex}}}}}\\
\diagLabel{Down right exit:} $\normalised\dXout$
}\\
\subfig{0.33}{
      \ensuremath{\vcenter{\hbox{\scalebox{0.75}{\input{ClosedDiags/NormalGen/brackCap.tex}}}}}\\
\diagLabel{Bubble cap: }$\normalised\brCap$
}
\subfig{0.33}{
        \ensuremath{\vcenter{\hbox{\scalebox{0.75}{\input{ClosedDiags/NormalGen/upOut.tex}}}}}\\
\diagLabel{Up left exit: }$\normalised\uXout$}
\subfig{0.33}{
        \ensuremath{\vcenter{\hbox{\scalebox{0.75}{\input{ClosedDiags/NormalGen/upIn.tex}}}}}\\
\diagLabel{Up left entry: }$\normalised\uXin$}
\\
\subfig{0.33}{
        \ensuremath{\vcenter{\hbox{\scalebox{0.75}{\input{ClosedDiags/NormalGen/cap.tex}}}}}\\
\diagLabel{Rightward cap:} $\normalised\rightCap$}
\subfig{0.33}{
        \ensuremath{\vcenter{\hbox{\scalebox{0.75}{\input{ClosedDiags/NormalGen/crossCup.tex}}}}}\\
\diagLabel{Left crossing cup: }$\normalised\Xcup$}
\subfig{0.33}{
        \ensuremath{\vcenter{\hbox{\scalebox{0.75}{\input{ClosedDiags/NormalGen/cup}}}}}\\
\diagLabel{Rightward cup: }$\normalised\rightCup$
}
\caption{Normalised term generators}
\label{fig:normalised term generators}
 \end{figure}
\label{sec:normal}
Besides the obvious inclusion $\normalised{\freeCat}\hookrightarrow \unref{\freeCat}$, we have a functor in the opposite direction. In order to define this, we first need a definition of the normaliser terms.

\begin{definition}
The normalisation functor is a strong closed, strong symmetric monoidal functor $\normF \colon \unref{\freeCat} \to \normalised{\freeCat}$ defined as follows
\begin{itemize}
        \item On objects it sends every object to its normal form.
        \item On morphisms it is defined inductively as follows: 
        \begin{itemize}
                \item $\alpha\mapsto (\normalised{\alpha})$ for all $\alpha$ in \autoref{fig:term generators};
                \item $f\mapsto \nu \circ f \circ \nu^{-1}$ for all $f\in \Sigma$;
                \item all other generators are mapped in the only possible way.
        \end{itemize}
\end{itemize}
\end{definition}
The fact that this is indeed a functor, as well as being strong closed and strong symmetric monoidal, immediately follows from the inductive definition.

In essence, we will give a theorem which proves that the functor above gives a kind of equivalence between the ``unbiased'' presentation of closed symmetric monoidal categories, and the ``biased' presentation. First we need the following elementary result.
\begin{lemma}
\label{lemma:normaliserFlip}
The normaliser adheres to the following two equalities.
\[
\ensuremath{\vcenter{\hbox{\scalebox{0.75}{\input{ClosedDiags/NormalLemma/4.tex}}}}} \cong
\ensuremath{\vcenter{\hbox{\scalebox{0.75}{\input{ClosedDiags/NormalLemma/3.tex}}}}} 
\qquad 
\ensuremath{\vcenter{\hbox{\scalebox{0.75}{\input{ClosedDiags/NormalLemma/2.tex}}}}} \cong
\ensuremath{\vcenter{\hbox{\scalebox{0.75}{\input{ClosedDiags/NormalLemma/1.tex}}}}}
\]
\end{lemma}
\begin{proof}
        Straightforward induction using the symmetry equations for the crossing of upstrings with downstrings.
\end{proof}
\begin{theorem}\label{thm:normalisation equivalence}
There is a closed symmetric monoidal equivalence $\freeCat \simeq \normalised{\freeCat}$. %
\end{theorem}
\begin{proof}
By \autoref{lemma:unrefEq}, it suffices to show that the normalisation functor forms an equivalence with the inclusion. Note that the composite
\tikzexternaldisable
\[\begin{tikzcd}
(\normalised{\freeCat})
\arrow[r, hookrightarrow]
         &  (\unref{\freeCat})
         \arrow[r,"\normF"]
                 & (\normalised{\freeCat})
\end{tikzcd}
\]
\tikzexternalenable
is the identity. Thus it suffcies to show that there is a closed monoidal natural isomorphism in the following diagram.
\tikzexternaldisable
\[\begin{tikzcd}
(\unref{\freeCat})
\arrow[r,"\normF"]
\arrow[rr, bend right, "{\id}"{swap}, ""{name=bot}]
        &  (\normalised{\freeCat})
        \arrow[r, hookrightarrow]
                 & (\unref{\freeCat})
\arrow[from=bot, to=1-2, phantom, "\cong"]
\end{tikzcd}
\]
\tikzexternalenable
We claim that such an isomorphism exists and is the normaliser. Since the normaliser is invertible all that remains to be shown is that it is natural, and to prove this we proceed inductively on generators. Clearly we have, for any $f\in \Sigma$
\[
\normF(f) \circ \nu = \nu \circ f\circ  \nu^{-1}\circ \nu = \nu \circ f.
\]
Now consider the generators in \autoref{fig:term generators}. The naturality with respect to the rightward cup follows from %
the diagrams below. We have that $(a)\cong(b)$ by the inductive definition of the normaliser, $(b)\cong(c)$ by the naturality proven in \autoref{fig:derCupNatural} and the equalities in \autoref{lemma:normaliserFlip}, and $(c)\cong(d)$ by the inductive definition of the normaliser. The naturality of the normaliser for the other generators in \autoref{fig:term generators} follow similarly, using the naturality axioms and the results of \autoref{fig:derNaturalities}. The normaliser is natural with respect to the other term generators either trivially or by inductive hypothesis, and the fact that this natural transformation is closed monoidal follows from the inductive definition, completing the proof.
\begin{figure}[H]
\subfig{0.5}{\ensuremath{\vcenter{\hbox{\scalebox{0.75}{\input{ClosedDiags/NormalProof/1}}}}}\figLab{a}}
\subfig{0.5}{\ensuremath{\vcenter{\hbox{\scalebox{0.75}{\input{ClosedDiags/NormalProof/2}}}}}\figLab{b}}\\
\subfig{0.5}{\ensuremath{\vcenter{\hbox{\scalebox{0.75}{\input{ClosedDiags/NormalProof/3}}}}}\figLab{c}}
\subfig{0.5}{\ensuremath{\vcenter{\hbox{\scalebox{0.75}{\input{ClosedDiags/NormalProof/4}}}}}\figLab{d}}
\end{figure}
\end{proof}

%% file: ClosedDiags/NormalLemma/4.tex
\includegraphics{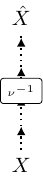}

%% file: ClosedDiags/NormalLemma/3.tex
\includegraphics{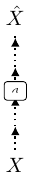}

%% file: ClosedDiags/NormalLemma/2.tex
\includegraphics{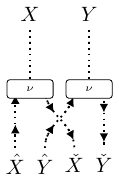}

%% file: ClosedDiags/NormalLemma/1.tex
\includegraphics{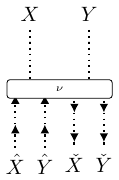}

%% file: ClosedDiags/Relations/Prop1/1.tex
\begin{tikzpicture}[baseline={(current bounding box.center)}, xscale=1
]
\drawStartNodes{
    b1//0.5,
    A/{$A$}/0.5,
    B/{$B$}/0.5,
    b2//0
}
\udShift{A}{0,-1}
\ddShift{B}{0,-1}
\bShift{b2}{0,-1}
\bShift{b1}{0,-1}[right]

\coordinate[shift={(-0.5,0)}] (b0) at (b1);
\coordinate[shift={(0.5,0)}] (b3) at (b2);
\bCap{b0}{b3}

\udShift{A}{-0.5,-1}
\bShift{b1}{0.5,-1}[right]

\ddShift{B}{0,-1}
\bShift{b2}{0,-1}
\bShift{b0}{0,-1}[right]
\bShift{b3}{0,-1}

\drawEndNodes{
    A/{$A$},
    B/{$B$}
}
\end{tikzpicture}

%% file: ClosedDiags/Relations/Prop1/2.tex
\begin{tikzpicture}[baseline={(current bounding box.center)}, xscale=1
]
\drawStartNodes{
    b1//0.5,
    A/{$A$}/1,
    B/{$B$}/1,
    b2//0
}
\udShift{A}{0,-1}
\ddShift{B}{0,-1}
\bShift{b2}{0,-1}
\bShift{b1}{0,-1}[right]

\coordinate[shift={(-0.5,0)}] (b0) at (B);
\coordinate[shift={(0.5,0)}] (b3) at (B);
\bCap{b0}{b3}

\udShift{A}{0,-1}
\bShift{b1}{0,-1}[right]

\ddShift{B}{0,-1}
\bShift{b2}{0,-1}
\bShift{b0}{0,-1}[right]
\bShift{b3}{0,-1}

\drawEndNodes{
    A/{$A$},
    B/{$B$}
}
\end{tikzpicture}

%% file: ClosedDiags/Relations/Prop1/3.tex
\begin{tikzpicture}[baseline={(current bounding box.center)}, xscale=1
]
\drawStartNodes{
    b1//0.5,
    A/{$A$}/0.5,
    B/{$B$}/0.5,
    b2//0
}
\udShift{A}{0,-1}
\ddShift{B}{0,-1}
\bShift{b2}{0,-1}
\bShift{b1}{0,-1}[right]

\coordinate[shift={(-0.5,0)}] (b0) at (b1);
\coordinate[shift={(0.5,0)}] (b3) at (b2);
\bCap{b0}{b3}

\udShift{A}{-0.5,-1}
\bShift{b1}{0.5,-1}[right]

\ddShift{B}{0,-1}
\bShift{b2}{0,-1}
\bShift{b0}{0,-1}[right]
\bShift{b3}{0,-1}

\ddShift{B}{0,-1.5}
\udShift{A}{0,-1.5}
\bShift{b0}{0,-1.5}[right]
\bShift{b3}{0,-1.5}

\bCup{b1}{b2}

\coordinate[shift={(0,-1.5)}] (b1) at (b1);
\coordinate[shift={(0,-1.5)}] (b2) at (b2);

\bCap{b1}{b2}

\drawEndNodes{
    A/{$A$},
    B/{$B$}
}
\end{tikzpicture}

%% file: figures/decompDownRightExit.tex
    \begin{figure}[H]
  \subfig{0.25}{
    \ensuremath{\vcenter{\hbox{\scalebox{0.75}{\input{ClosedDiags/NormalGen/downOut.tex}}}}} 
 \figLab{a}
  }      
  \subfig{0.25}{
    \ensuremath{\vcenter{\hbox{\scalebox{0.75}{\input{ClosedDiags/WellDef2/downLeave/2a}}}}}  \figLab{b}
  }
  \subfig{0.25}{
    \ensuremath{\vcenter{\hbox{\scalebox{0.75}{\input{ClosedDiags/WellDef2/downLeave/2}}}}}  \figLab{c}
  }
\subfig{0.25}{
    \ensuremath{\vcenter{\hbox{\scalebox{0.75}{\input{ClosedDiags/WellDef2/downLeave/1}}}}} \figLab{d}
  }
  \end{figure}

%% file: figures/decompDownRightEntry.tex
  \begin{figure}[H]
  \subfig{0.5}{
    \ensuremath{\vcenter{\hbox{\scalebox{0.75}{\input{ClosedDiags/NormalGen/downIn.tex}}}}} \figLab{a}
  }      
  \subfig{0.5}{
    \ensuremath{\vcenter{\hbox{\scalebox{0.75}{\input{ClosedDiags/WellDef2/downEnter/1}}}}} \figLab{b}
  }\figSpace
  \subfig{0.5}{
    \ensuremath{\vcenter{\hbox{\scalebox{0.75}{\input{ClosedDiags/WellDef2/downEnter/2}}}}} \figLab{c}
  }
\subfig{0.5}{
    \ensuremath{\vcenter{\hbox{\scalebox{0.75}{\input{ClosedDiags/WellDef2/downEnter/3}}}}} \figLab{d}
  }
  \end{figure}

%% file: figures/decompCupCrossing.tex
  \begin{figure}[H]
  \subfig{0.25}{
    \ensuremath{\vcenter{\hbox{\scalebox{0.75}{\input{ClosedDiags/NormalGen/crossCup.tex}}}}}
 \figLab{a}
  }
    \subfig{0.25}{
      \ensuremath{\vcenter{\hbox{\scalebox{0.75}{\input{ClosedDiags/WellDef2/cup/3.tex}}}}} \figLab{b}
    }%
  \subfig{0.25}{
      \ensuremath{\vcenter{\hbox{\scalebox{0.75}{\input{ClosedDiags/WellDef2/cup/2.tex}}}}} \figLab{c}
  }     
  \subfig{0.25}{
      \ensuremath{\vcenter{\hbox{\scalebox{0.75}{\input{ClosedDiags/WellDef/cup/1.tex}}}}} \figLab{d}
  }
  \end{figure}

%% file: figures/decompUpLeftExit.tex
  \begin{figure}[H]
    \subfig{0.33}{
        \ensuremath{\vcenter{\hbox{\scalebox{0.75}{\input{ClosedDiags/NormalGen/upOut.tex}}}}} \figLab{a}
    }
  \subfig{0.33}{
      \ensuremath{\vcenter{\hbox{\scalebox{0.75}{\input{ClosedDiags/WellDef2/upLeaveLeft/2.tex}}}}} \figLab{b}
  }      
  \subfig{0.33}{
      \ensuremath{\vcenter{\hbox{\scalebox{0.75}{\input{ClosedDiags/WellDef2/upLeaveLeft/1.tex}}}}} \figLab{c}
  }
  \end{figure}

%% file: figures/decompRightCap.tex
  \begin{figure}[H]
  \subfig{0.5}{
    \ensuremath{\vcenter{\hbox{\scalebox{0.75}{\input{ClosedDiags/NormalGen/cap.tex}}}}} \figLab{a}
  }      
  \subfig{0.5}{
    \ensuremath{\vcenter{\hbox{\scalebox{0.75}{\input{ClosedDiags/DecompProofs/cap/1.tex}}}}}\figLab{b}
  }
  \end{figure}

%% file: figures/decompUpLeftEntry.tex
  \begin{figure}[H]
    \subfig{0.23}{
        \ensuremath{\vcenter{\hbox{\scalebox{0.7}{\input{ClosedDiags/NormalGen/upIn.tex}}}}}  \figLab{a}
    }
  \subfig{0.24}{
      \ensuremath{\vcenter{\hbox{\scalebox{0.7}{\input{ClosedDiags/WellDef2/upEnter/3.tex}}}}} \figLab{b}
  }
  \subfig{0.28}{
      \ensuremath{\vcenter{\hbox{\scalebox{0.7}{\input{ClosedDiags/WellDef2/upEnter/2a.tex}}}}} \figLab{c}
  }
  \subfig{0.25}{
      \ensuremath{\vcenter{\hbox{\scalebox{0.7}{\input{ClosedDiags/WellDef/upEnter/1.tex}}}}} \figLab{d}
  }
  \end{figure}

%% file: figures/decompRightCup.tex
   \begin{figure}[H]
    \subfig{0.5}{
  \ensuremath{\vcenter{\hbox{\scalebox{0.75}{\input{ClosedDiags/NormalGen/cup}}}}} \figLab{a}
    }
  \subfig{0.5}{
       \ensuremath{\vcenter{\hbox{\scalebox{0.75}{\input{ClosedDiags/BigCupSplit/1N}}}}} \figLab{b}
  }
  \end{figure}

%% file: ClosedDiags/CellGen/f.tex
\includegraphics{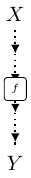}

%% file: ClosedDiags/Relations/pop1a.tex
\includegraphics{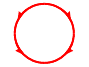}

%% file: ClosedDiags/yankSound.tex
\includegraphics{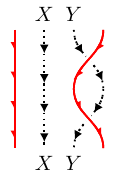}

%% file: ClosedDiags/yankSound2a.tex
\includegraphics{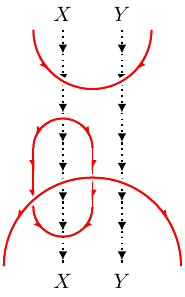}

%% file: ClosedDiags/yankSound3.tex
\includegraphics{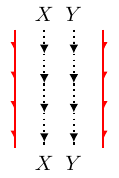}

%% file: sections/conclusions.tex
\section{Conclusions and Future Work}
We introduced a string diagram formalism which, by virtue of Theorem~\ref{thm:freeness}, serves as a universal language for closed symmetric monoidal categories.

Compared to other approaches in the literature, our contribution offers (1) a completeness result and (2) a fully diagrammatic language that avoids mixed algebraic notation and requires no additional machinery beyond strings.

Given the ubiquity of closed symmetric monoidal categories and their cartesian variants, we believe our formalism constitutes a valuable tool for reasoning across a range of domains, from pure mathematics to logic and theoretical computer science.

As future work, we plan to investigate the hypergraph interpretation of our string diagrams, following the approach of~\cite{DBLP:journals/jacm/BonchiGKSZ22}. This characterization, beyond offering an alternative definition, also has computational significance: hypergraphs can serve as data structures for potential implementations in, for example, proof assistants or compiler optimizations.

Some inspiration may also be drawn from the bigraphical approach taken in~\cite{garner2008graphical,garner2009variable}.

For the non-symmetric setting, variants of these diagrams can be used in the study of biclosed bicategories, such as bicategories of relations, or profunctors. In particular, the recent trend of adopting bicategorical models as proof-relevant denotational semantics for programming languages~\cite{clairambault2023cartesian,fiore2008cartesian} and linear logic~\cite{fiore2024stabilized, mellies2019template} opens up new possibilities for applying our diagrams in this context.

In topology there are studies of fixed point theorems using traces in bicategories~\cite{ponto2008fixedpointtheorytrace,PONTO2014248}. If a compact closed bicategory is also biclosed compositionally, then the biclosed structure can be used to define a \emph{co}trace~\cite{barhite2024bicategoricaltracescotraces,reader2024scalarenrichmentcotracesbicategories}. A wide variety of mathematical phenomena can be expressed as a cotrace -- including ends and Hochschild cohomology -- and we believe that results about these phenomena can now be expressed and derived purely graphically.

%% file: sections/appendix/presentation2.tex
\section{Formal construction of $\freeCat$}\label{app:formal c sigma}

In this section we give a detailed construction of the category $\freeCat$ of string diagrams with brackets. For convenience, we rely on a term-based syntax equivalent to the diagrammatic syntax, presented in Section~\ref{sec:presentation}.

Given a (unbiased) closed monoidal signature $(\sort, \sign)$, we generate terms from the following three layer grammar:
    \[
    \begin{array}{c c c@{\quad \mid \quad}c@{\quad \mid \quad}c@{\quad \mid \quad}c@{\quad \mid \quad}c@{\quad \mid \quad}c@{\quad \mid \quad}c@{\quad \mid \quad}c@{\quad \mid \quad}c}
        \toprule
        \multicolumn{10}{c}{\text{Covariant}}
        \\
        \midrule
        t & \coloneqq & 
        \brCup & \dXin & \dXout & \uXin & \uXout & \Xcup & \rightCap & \rightCup
        \\[5pt]
        & &
        \brCap & \YdXin & \YdXout & \YuXin & \YuXout & \YXcup & \YrightCap &\YrightCup 
        \\[5pt]
        & & \id_X & f & t ; t & t \otimes t & \sigma_{X,Y} & \multicolumn{3}{@{}l}{[\mathbf{t}]}
        \\
        \midrule
        \multicolumn{10}{c}{\text{Contravariant}}
        \\
        \midrule
        t^\ast & \coloneqq & 
        {\flip{\brCup}} & {\flip{\dXin}} & {\flip{\dXout}} & {\flip{\uXin}} & {\flip{\uXout}} & {\flip{\Xcup}} & {\flip{\rightCap}} & {\flip{\rightCup}}
        \\[5pt]
        & & {\XYflip{\brCap}} & {\XYflip{\dXin}} & {\XYflip{\dXout}} & {\XYflip{\uXin}} & {\XYflip{\uXout}} & {\XYflip{\Xcup}} & {\XYflip{\rightCap}} & {\XYflip{\rightCup}}
        \\[5pt]
        & & \id_{X^\ast} & f^\ast & t^\ast ; t^\ast & t^\ast \otimes t^\ast & \sigma_{Y^\ast,X^\ast} & \multicolumn{3}{@{}l}{[\mathbf{t}]^\ast}
        \\
        \midrule
        \multicolumn{10}{c}{\text{Mixed}}
        \\
        \midrule
        \mathbf{t} & \coloneqq & t & t^\ast & \mathbf{t} ; \mathbf{t} & \mathbf{t} \otimes \mathbf{t} & \sigma_{X,Y^\ast} & \multicolumn{3}{@{}l}{\sigma_{X^\ast,Y}}
        \\
        \bottomrule
    \end{array}
    \]
    where $f \colon X \to Y$ is a symbol in $\Sigma$ and $f^\ast \colon Y^\ast \to X^\ast$ is its corresponding dual.

    \medskip

    As conveyed by the diagrammatic representation of the grammar above (see Section~\ref{sec:presentation}), the contravariant terms are the covariant ones, reflected in the horizontal axis. This correspondence is made formal by the function $(-)^\ast$ that maps each term $t \colon X \to Y$ to its dual $t^\ast \colon Y^\ast \to X^\ast$ and vice versa, defined inductively as follows:
    \begin{itemize}
        \item on objects:
        \begin{itemize}
        \item for each $A \in \sort$, $(A)^\ast \coloneqq A^\ast$ and $(A^\ast)^\ast \coloneqq A$;
        \item for the empty word $I$, $(I)^\ast \coloneqq I$;
        \item for concatenation $(X\mathbf{Y})^\ast \coloneqq (\mathbf{Y})^\ast (X)^\ast$ and $(\mathbf{Y}X)^\ast \coloneqq (X)^\ast(\mathbf{Y})^\ast$; 
        \item for bracketing $([\mathbf{X}])^\ast \coloneqq [(\mathbf{X})^\ast]^\ast$ and $([\mathbf{X}]^\ast)^\ast \coloneqq [(\mathbf{X})^\ast]$.
        \end{itemize}
        \item on terms:
        \begin{itemize}
        \item for each covariant bracket symbol $s$, $(s)^\ast \coloneqq s^\ast$, where $s^\ast$ is the corresponding contravariant symbol. e.g. $(\brCup)^\ast \coloneqq \flip{\brCup}$;
        \item for each contravariant bracket symbol $s^\ast$, $(s^\ast)^\ast \coloneqq s$, where $s$ is the corresponding covariant symbol; e.g. $(\flip{\brCup})^\ast \coloneqq \brCup$;
        \item for symbols $f \in \Sigma$, $(f)^\ast \coloneqq f^\ast$ and $(f^\ast)^\ast \coloneqq f$;
        \item for identities $(\id_X)^\ast \coloneqq \id_{X^\ast}$ and $(\id_{X^\ast})^\ast \coloneqq \id_X$;
        \item for symmetries $(\sigma_{X,Y})^\ast \coloneqq \sigma_{Y^\ast, X^\ast}$, $(\sigma_{Y^\ast, X^\ast})^\ast \coloneqq \sigma_{X,Y}$, $(\sigma_{X,Y^\ast})^\ast \coloneqq \sigma_{Y, X^\ast}$ and $(\sigma_{X^\ast,Y})^\ast \coloneqq \sigma_{Y^\ast, X}$;
        \item for compositions $(t_1 ; t_2)^\ast \coloneqq (t_2)^\ast ; (t_1)^\ast$, $(t_1 \otimes t_2)^\ast \coloneqq (t_2)^\ast \otimes (t_1)^\ast$;
        \item and for bracketings $([\mathbf{t}])^\ast \coloneqq [(\mathbf{t})^\ast]^\ast$ and $([\mathbf{t}]^\ast)^\ast \coloneqq [(\mathbf{t})^\ast]$.
        \end{itemize}
    \end{itemize}

    Amongst the terms generated by the grammar above, we consider only those that can be typed according to the following typing rules:

    \noindent\scalebox{0.9}{$\renewcommand{\arraystretch}{1.5}
        \begin{array}{r@{\;}c@{\;}l r@{\;}c@{\;}l}
            \brCup &\colon& [I] \to I 
            & 
            \dXin_{X,\mathbf{Y}} &\colon& X [\mathbf{Y}] \to [X\mathbf{Y}]
            \\
            \dXout_{X,Y} &\colon& [XY] \to X[Y]
            &
            \uXin_{X^\ast,Y^\ast,\mathbf{Z},W^\ast} &\colon& [X^\ast Y^\ast [\mathbf{Z}] W^\ast] \to [X^\ast [Y^\ast \mathbf{Z}] W^\ast]
            \\
            \uXout_{\mathbf{X}, Y^\ast, \mathbf{Z}, \mathbf{W}} &\colon& [\mathbf{X} [Y^\ast \mathbf{Z}] \mathbf{W}] \to [\mathbf{X} Y^\ast [\mathbf{Z}] \mathbf{W}]
            & 
            \Xcup_{X,\mathbf{Y}} &\colon& X [X^\ast \mathbf{Y}] \to [\mathbf{Y}]
            \\ 
            \rightCap_{\mathbf{X}, Y, \mathbf{Z}} &\colon& [\mathbf{X} \mathbf{Z}] \to [\mathbf{X} Y^\ast Y \mathbf{Z}]
            & 
            \rightCup_{\mathbf{X}, Y, \mathbf{Z}} &\colon& [\mathbf{X} Y] [Y^\ast \mathbf{Z}] \to [\mathbf{X} \mathbf{Z}]
            \\
            \brCap &\colon& I \to [I]
            & 
            \YdXin_{\mathbf{Y}, X} &\colon& [\mathbf{Y}] X \to [\mathbf{Y} X]
            \\
            \YdXout_{Y,X} &\colon& [YX] \to [Y]X
            & 
            \YuXin_{W^\ast, \mathbf{Z}, Y^\ast, X^\ast} &\colon& [W^\ast [\mathbf{Z}] Y^\ast X^\ast] \to [W^\ast [\mathbf{Z} Y^\ast] X^\ast]
            \\
            \YuXout_{\mathbf{W}, \mathbf{Z}, Y^\ast, \mathbf{X}} &\colon& [\mathbf{W} [\mathbf{Z} Y^\ast] \mathbf{X}] \to [\mathbf{W} [\mathbf{Z}] Y^\ast \mathbf{X}]
            & 
            \YXcup_{\mathbf{Y}, X} &\colon& [\mathbf{Y} X^\ast] X \to [\mathbf{Y}]
            \\
            \YrightCap_{\mathbf{Z}, Y, \mathbf{X}} &\colon& [\mathbf{Z} \mathbf{X}] \to [\mathbf{Z} Y Y^\ast \mathbf{X}]
            &
            \YrightCup_{\mathbf{Z}, Y, \mathbf{X}} &\colon& [\mathbf{Z} Y^\ast] [Y \mathbf{X}] \to [\mathbf{Z} \mathbf{X}]
            \\
            \multicolumn{6}{@{}l}{
                \renewcommand{\arraystretch}{1.5}
                \begin{array}{l@{\qquad}l@{\qquad}l@{\qquad}l}
                    \multicolumn{4}{c}{
                        \begin{array}{l@{\qquad\qquad}l}
                            \id_X \colon X \to X
                            &
                            f \colon \arity(f) \to \coarity(f)
                        \end{array}
                    }
                    \\
                    \multicolumn{4}{c}{
                        \begin{array}{l@{\qquad\qquad}l@{\qquad\qquad}l}
                            \sigma_{X,Y} \colon XY \to YX
                            &
                            \sigma_{X,Y^\ast} \colon XY^\ast \to Y^\ast X
                            &
                            \sigma_{X^\ast,Y} \colon X^\ast Y \to Y X^\ast
                        \end{array}
                    }
                    \\
                    \inferrule{t \colon X \to Y}{(t)^\ast \colon (Y)^\ast \to (X)^\ast}
                    &
                    \inferrule{t_1 \colon X \to Y \quad t_2 \colon Y \to Z}{t_1 ; t_2 \colon X \to Z}
                    &
                    \inferrule{t_1 \colon X \to Y \quad t_2 \colon Z \to W}{t_1 \otimes t_2 \colon XZ \to YW}
                    &
                    \inferrule{\mathbf{t} \colon \mathbf{X} \to \mathbf{Y}}{[\mathbf{t}] \colon [\mathbf{X}] \to [\mathbf{Y}]}
                \end{array}
            }
        \end{array}$}

        \medskip

Terms are then quotiented by the equations of symmetric monoidal categories and those governing the additional structure of bracketing strings:

\medskip

\begin{center}
\noindent\scalebox{0.9}{$
    \begin{array}{@{}cc}
        \toprule
        \multicolumn{2}{c}{\text{Symmetric Monoidal Categories}}
        \\
        \midrule
        \id_X ; t = t = t ; \id_X
        &
        (t_1 ; t_2) ; t_3 = t_1 ; (t_2 ; t_3)
        \\
        \id_I \otimes t = t = t \otimes \id_I
        &
        (t_1 \otimes t_2) \otimes t_3 = t_1 \otimes (t_2 \otimes t_3)
        \\
        \multicolumn{2}{c}{
            (t_1 \otimes t_2) ; (t_3 \otimes t_4) = (t_1 ; t_3) \otimes (t_2 ; t_4)
        }
        \\
        \sigma_{XY} ; \sigma_{YX} = \id_{XY}
        &
        (t \otimes id_Z) ; \sigma_{YZ} = \sigma_{XZ} ; (\id_Z \otimes t)
        \\
        \multicolumn{2}{c}{
            (\mathbf{t}_1 \otimes \mathbf{t}_2) ; (\mathbf{t}_3 \otimes \mathbf{t}_4) = (\mathbf{t}_1 ; \mathbf{t}_3) \otimes (\mathbf{t}_2 ; \mathbf{t}_4)
        }
        \\
        \sigma_{XY^\ast} ; \sigma_{Y^\ast X} = \id_{XY^\ast}
        &
        \sigma_{X^\ast Y} ; \sigma_{YX^\ast} = \id_{X^\ast Y}
        \\
        (t \otimes id_{Z^\ast}) ; \sigma_{Y{Z^\ast}} = \sigma_{X{Z^\ast}} ; (\id_{Z^\ast} \otimes t)
        &
        (t^\ast \otimes id_Z) ; \sigma_{{Y^\ast}Z} = \sigma_{{X^\ast}Z} ; (\id_Z \otimes t^\ast)
        \\
        \midrule
        \multicolumn{2}{c}{\text{Crossing}}
        \\
        \midrule
        (\id_X \otimes \sigma_{\mathbf{Y} X^\ast}) ; \Xcup_{X, \mathbf{Y}} = \sigma_{X[\mathbf{Y} X^\ast]} ; \YXcup_{\mathbf{Y},X}
        &
        \rightCap_{\mathbf{X},Y,\mathbf{Z}} ; [\id_{\mathbf{X}} \otimes \sigma_{Y^\ast, Y} \otimes \id_{\mathbf{Z}}] = \YrightCap_{\mathbf{X}, Y, \mathbf{Z}}
        \\
        \dXin_{X,Y} ; \dXout_{X,Y} = \id_{X[Y]}
        &
        \dXout_{X,Y} ; \dXin_{X,Y} = \id_{[XY]}
        \\
        \uXin_{X^\ast, Y^\ast, \mathbf{Z}, W^\ast} ; \uXout_{X^\ast, Y^\ast, \mathbf{Z}, W^\ast} = \id_{[X^\ast Y^\ast [\mathbf{Z}] W^\ast]}
        &
        \uXout_{X^\ast, Y^\ast, \mathbf{Z}, W^\ast} ; \uXin_{X^\ast, Y^\ast, \mathbf{Z}, W^\ast} = \id_{[X^\ast [Y^\ast \mathbf{Z}] W^\ast]}
        \\
        \midrule
        \multicolumn{2}{c}{\text{Naturality}}
        \\
        \midrule
        \multicolumn{2}{c}{
            \begin{array}{c}
                (\mathbf{t}_0 \otimes [\id_{Y^\ast} \otimes \mathbf{t}_1]) ; \Xcup_{Y,\mathbf{W}} = (\id_X \otimes [\mathbf{t}_0^\ast \otimes \id_{\mathbf{Z}}]) ; \Xcup_{X,\mathbf{Z}} ; [\mathbf{t}_1]
                \\
                {[\id_{\mathbf{U}} \otimes \mathbf{t}_1] ; \rightCap_{\mathbf{U},Y,\mathbf{W}} ; [\id_{\mathbf{U}} \otimes t_0^\ast \otimes \id_{Y\mathbf{W}}] = \rightCap_{\mathbf{U},X,\mathbf{Z}} ; [\id_{\mathbf{U} X^\ast} \otimes t_0 \otimes \mathbf{t}_1]}
                \\
                (\Xcup_{X, \mathbf{Y}W} \otimes \id_{W^\ast \mathbf{U}}) ; ([\mathbf{t}_0 \otimes \id_{W}] \otimes [\id_{W^\ast} \otimes \mathbf{t}_1]) ; \rightCup_{\mathbf{Z},W,\mathbf{V}} = (\id_X \otimes \rightCup_{X^\ast \mathbf{Y},W,\mathbf{U}}) ; (\Xcup_{X,\mathbf{YU}}) ; [\mathbf{t}_0 \otimes \mathbf{t}_1]
            \end{array}
        }
        \\
        \midrule
        \multicolumn{2}{c}{\text{Yanking}}
        \\
        \midrule
        \rightCap_{I, X, \mathbf{Y}} ; \Xcup_{X,X\mathbf{Y}} = \dXin_{X,\mathbf{Y}}
        &
        \rightCap_{\mathbf{X},Y,[Y^\ast,\mathbf{Z}]\mathbf{W}} ; [\id_{{\mathbf{X}}Y^\ast} \otimes \Xcup_{Y, \mathbf{Z}} \otimes \id_{\mathbf{W}}] = \dXout_{\mathbf{X}, Y^\ast, \mathbf{Z}, \mathbf{W}}
        \\
        \multicolumn{2}{c}{
            (\id_{[X]} \otimes \rightCap_{I,X,\mathbf{Z}}) ; \rightCup_{I, X, X\mathbf{Z}} = \rightCup_{I,X,X\mathbf{Z}}
        }
        \\
        \midrule
        \multicolumn{2}{c}{\text{Pop, merge, yanking and functoriality of brackets}}
        \\
        \midrule
        \brCap ; \brCup = \id_I
        &
        \brCup_{\mathbf{X}} ; \brCap_{\mathbf{X}} = \id_{[\mathbf{X}]}
        \\
        (\brCap \otimes \id_{\mathbf{X}}) ; \rightCup_{I,I,\mathbf{X}} = [\id_{\mathbf{X}}]
        &
        [\mathbf{t}_0] ; [\mathbf{t}_1] = [\mathbf{t}_0 ; \mathbf{t}_1] \qquad\qquad [\id_{\mathbf{X}}] = \id_{[\mathbf{X}]}
        \\
        \bottomrule
    \end{array}
$}
\end{center}

\medskip

In particular, for each equation $t_1 = t_2$ above, we consider the set $\mathbb{E}$ of pairs $(t_1, t_2)$ and form the smallest congruence closure $\cong$ (w.r.t. $;$, $\otimes$, $(-)^\ast$ and $[-]$), inductively as follows:
\[
    \begin{array}{c}
        \begin{array}{c@{\qquad\quad}c@{\qquad\quad}c@{\qquad\quad}c}
        \inferrule{(t_1, t_2) \in \mathbb{E}}{t_1 \cong t_2}
        &
        \inferrule{-}{t \cong t}
        &
        \inferrule{t_1 \cong t_2}{t_2 \cong t_1}
        &
        \inferrule{t_1 \cong t_2 \qquad t_2 \cong t_3}{t_1 \cong t_3}
        \end{array}
        \\[15pt]
        \begin{array}{c@{\qquad\quad}c@{\qquad\quad}c@{\qquad\quad}c}
            \inferrule{t_1 \cong t_2 \qquad t_3 \cong t_4}{t_1 ; t_3 \cong t_2 ; t_4}
            &
            \inferrule{t_1 \cong t_2 \qquad t_3 \cong t_4}{t_1 \otimes t_3 \cong t_2 \otimes t_4}
            &
            \inferrule{t_1 \cong t_2}{(t_1)^\ast \cong (t_2)^\ast}
            &
            \inferrule{t_1 \cong t_2}{[t_1] \cong [t_2]}
        \end{array}    
    \end{array}
\]

\begin{definition}
    Given a (unbiased) closed monoidal signature $(\sort, \sign)$, we denote as $\freeCat$ the category whose objects are elements of $\sort^\sharp$ and morphisms are the typed covariant terms quotiented by $\cong$.
\end{definition}